\documentclass[a4paper,10pt,twoside]{amsart}

\usepackage[affil-it]{authblk}
\usepackage{titling}
\usepackage[latin1]{inputenc}
\usepackage{calligra}
\usepackage[OT1]{fontenc}
\usepackage[english]{babel}
\usepackage{amsfonts}
\usepackage{amsmath}
\usepackage{amssymb}
\usepackage{amsthm}
\usepackage{faktor}
\usepackage{float}
\usepackage{enumerate}
\usepackage{color}
\usepackage{pdfpages}
\usepackage{esint}
\usepackage{hyperref}
\usepackage{cite}
\usepackage{fancyhdr}

\newcommand\testname{Abstract}
\makeatletter
\newenvironment{abs}{%
    \small
    \begin{center}%
        {\textsc \testname\vspace{-.2em}\vspace{\z@}}%
    \end{center}%
    \quote
    }
   {\endquote}
\makeatother

\DeclareMathOperator*{\Id}{Id}

\newcommand{\Div}{\mathrm{div}}
\newcommand{\PP}{\mathbb{P}}
\newcommand{\ee}{\varepsilon}

\newcommand{\Aa}{\mathcal{A}}
\newcommand{\Cc}{\mathcal{C}}
\newcommand{\Bb}{\mathcal{B}}
\newcommand{\dd}{\mathrm{d}}

\newcommand{\X}{\mathfrak{X}}

\newcommand{\Y}{\mathfrak{Y}}
\newcommand{\SSS}{\mathbb{S}}

\newcommand{\RR}{\mathbb{R}}
\newcommand{\ZZ}{\mathbb{Z}}
\newcommand{\NN}{\mathbb{N}}

\newcommand{\BB}{\dot{B}}
\newtheorem{theorem}{Theorem}[section]
\newtheorem{cor}{Corollary}[theorem]
\newtheorem{prop}[theorem]{Proposition}
\newtheorem{lemma}[theorem]{Lemma}
\newtheorem{definition}[theorem]{Definition}
\newtheorem{remark}[theorem]{Remark}

\title{\Large 
	\textbf{\uppercase{
		Global Solvability of the inhomogeneous Ericksen-Leslie system 
		with only bounded density }}}
	
\author{Francesco De Anna$\quad$}
\affil{	\textsc{Universit\'e de Bordeaux} \\ 
		\small\textsc{Institut de Math\'ematiques de Bordeaux}\\ 
		\small{F-33405 Talence Cedex, France}
		\\ \vspace{0cm}\\
		\small\textnormal{Francesco.Deanna@math.u-bordeaux1.fr}}
\date{\today}
\markleft{ FRANCESCO DE ANNA}
\begin{document}
\maketitle
\begin{abs}
	Ericksen and Leslie established a theory to model the flow of nematic liquid crystals.
	This paper is devoted to the Cauchy Problem of a simplified version of their system, 
	which retains most of the properties of the original one. 
	We consider the density-dependent case and we establish the global existence of solutions in the 
	whole space for small initial data. The initial density only has to be bounded and kept far from vacuum, 
	while the initial velocity belongs to some critical Besov Space. 
	Under a little bit more regularity for the initial velocity, we prove also that those solutions are unique. 
\end{abs}
\vspace{0.2cm}

\vspace{0.1cm}
\section{Introduction and main results}
\noindent Over the period of 1958 through 1968, Ericksen \cite{MR0137403} and Leslie \cite{MR1553506} developed the theory of liquid crystal materials. A liquid crystal is a compound of fluid molecules, which has a state of matter between the ordinary liquid one and the crystal solid one. The molecules have not a positional order but they  assume an orientation, which can be modified by the velocity flow. At the same time a variation of this alignment can induce a velocity field.

\noindent The present paper is devoted to the global solvability issue for the following system, which characterizes the liquid crystal hydrodynamics:
\begin{equation}\label{Navier_Stokes_system_1}
	\begin{cases}
		\;\partial_t\rho + \Div\, (\rho u)=0									& \RR_+ \times\RR^N,\\
		\;\partial_t (\rho u)+ \Div (\rho u\otimes u)-\nu \Delta u + \nabla\Pi =
		-\rho \lambda \, \Div\big(\nabla d\odot \nabla d \big)		& \RR_+ \times\RR^N,\\
		\;\Div\, u = 0															& \RR_+ \times\RR^N,\\
		\; \partial_t d + u\cdot\nabla d-\gamma \Delta d= \gamma |\nabla d|^2 d	& \RR_+ \times\RR^N,\\
		\;|d|=1																	& \RR_+ \times\RR^N,\\
		\;(u,\rho, d)_{|t=0} = (u_0,\,\rho_0, \,  d_0)							& \;\;\quad \quad\RR^N.\\
	\end{cases}
\end{equation}
This is a strongly coupled system between the inhomonegenous and incompressible Navier-Stokes equation and the transported heat flow of harmonic maps into sphere. Here $\rho=\rho(t,x)\in \RR_+$ denotes the density, $u=u(t,x)\in \RR^N$ represents the velocity field, 
$\Pi=\Pi(t,x)\in \RR$ stands for the pressure and $d=d(t,x)\in \SSS^{N-1}$ (the unit sphere in $\RR^N$) represents the molecular orientation of the liquid crystal, depending on the time variable $t\in \RR_+$ and on the space variables $x\in \RR^N$, with $N\geq 2$. The symbol $\nabla d\odot \nabla d$ denotes the 
$N\times N$ matrix whose $(i,\,j)$-th entry is given by $\partial_i d\cdot \partial_j d$, for $i,j=1,\dots, N$.
The positive constants $\nu$, $\lambda$ and $\gamma$ stand for the viscosity, the competition between kinetic energy and potential energy and the microscopic elastic relaxation time for the molecular orientation field respectively. 

\noindent We immediately remark that in the present paper we are going only to suppose the initial density $\rho_0$ to be bounded, thus it can present discontinuities along an interface. Such condition should be a model for an important physical case, namely the description of a mixture of fluids with different densities filled with crystals.

\subsubsection*{Some developments in the crystal liquid Theory}
System \eqref{Navier_Stokes_system_1} has been presented by Lin (see \cite{MTOLC,  MR1003435, MR1883167}) to model the  flow of nematic liquid crystals with variable degree of orientation. It has been formulated as a simplification (see \cite{MR1329830} appendix) of the original equations in the theory of liquid crystal proposed by Ericksen in \cite{MR0158610, MR0137403} and Leslie \cite{MR1553506}.

\noindent The solvability problem of the general Ericksen-Leslie equations has been studied by F. Lin and Liu in \cite{MR1784963}, proving the existence of weak and strong solutions under specifically conditions and also by Wu, Xu and Liu in \cite{MR3021544}, establishing the local well-posedness, and the global well-posedness for small intial data. However, because of the complexity of such model, only special cases have been treated in literature (see for example \cite{GennesProst199508}). On the other hand, the simplified system \eqref{Navier_Stokes_system_1} is more maneuverable, even if it preserves the mainly characteristics of the first one.

\noindent In \cite{MR1329830} F. Lin and Liu consider the homogeneous sub-case of system \eqref{Navier_Stokes_system_1} (namely with constant density), replacing the term $|\nabla d|^2 d$ with the forcing term given by
$f(d)=\nabla F(d)$. They study the wellposedness of the system, establishing the following basic energy law:
\begin{equation*}
	\frac{1}{2}\frac{\dd}{\dd t}\int_\Omega (\|u\|^2 + \lambda \|\nabla d\|^2 + 2 \lambda F(d))\dd x = 
	-\int_\Omega (\nu \|\nabla u\|^2 +\lambda \gamma \|\Delta d-f(d)\|^2)\dd x.
\end{equation*}
Then, with a modified Galerkin method, they are able to prove the existence of a weak solution $(u,\,d)$, with  
$u\in L^2(0,T; H^1(\Omega))\cap C([0,T], L^2(\Omega))$ and $d\in L^2(0,T;H^2(\Omega))\cap C([0,T];H^1(\Omega))$, provided $u_0\in L^2(\Omega)$, $d_0\in H^1(\Omega)$ and $d_{0|\partial \Omega} \in 
H^{3/2}(\partial \Omega)$. They also get uniqueness in the 2D case or in the 3D case on the condition 
$\nu \geq \nu ( \gamma , \,\lambda,\,u_0,\,d_0)$. Furthermore they prove a stability result for the equilibria.
Moreover, in \cite{MR1367385} F. Lin and Liu show a regularity result for the system, namely if the initial-boundary condition are smooth enough, then there exists a suitable weak solution whose singular set has one-dimension Hausdorff measure zero in space time.

\noindent In \cite{MR2561116}, Jiang and Tan consider the 3D system \eqref{Navier_Stokes_system_1} on a bounded domain with $f(d) = \nabla F(d)$ instead of $|\nabla d|^2 d$. They derive the global existence of weak solutions provided that the boundary is $C^{3+\delta}$, for a positive $\delta<1$, $d_0\in H^1(\Omega)$ with $d_{0|\partial \Omega} \in H^{3/2}(\partial \Omega)$, the initial density $\rho_0\in L^{\gamma}(\Omega)$ with $\gamma\geq 3/2$, $u_0=0$ whenever $\rho_0=0$ and $|u_0|^2/\rho_0$ belongs to $L^1(\Omega)$.

\noindent In dimension two, the homogeneous sub-case of system \eqref{Navier_Stokes_system_1} on a bounded domain has also been treated by F. Lin, J. Lin and C. Wang in \cite{MR2646822} and by Hong in \cite{MR2745194}.
Even if the velocity field is supposed to be defined on a plane, the authors suppose the director field $d$ to assume values on $S^2$, that is in the three dimensional framework. The first authors prove both interior and boundary regularity Theorems under smallness condition, which allow to obtain the existence of global weak solutions on a bounded smooth domain. Such solutions are smooth except for a finite time-set. An equivalent result has been proved simultaneously by Hong with a different approach, getting an $L^2$-estimate of $\nabla^2 d$ and 
$\nabla u$ under a small energy condition on the initial data.

\noindent In \cite{MR2917124}, D. Wang and Yu studied system \eqref{Navier_Stokes_system_1} on a bounded domain of class $C^{2+\nu}$, where $\nu>0$, with $f(d) = \nabla F(d)$ instead of $|\nabla d|^2 d$. They consider the compressible case, namely with the pressure $\Pi$ dependent by the density and without the free divergence condition on the velocity field. They proved existence and large-time behavior of a global weak solution.

\noindent Still in the compressible case,  F. Jiang, S. Jiang and D. Wang in \cite{MR3110506} established the global existence of weak solutions in a bounded domain in dimension three, under a smallness condition on the N-th component of the initial director field. Moreover in \cite{MR3255696} they determined the existence of weak solutions in dimension two, overcoming the supercritical nonlinearity of $|\nabla d|^2d$ and applying a three-level approximation scheme. Such results hold under some restriction on the initial energy (including the case of small initial energy). They also proved a global existence result with large initial data, provided that the second component of the initial director field fulfils some geometric angle condition. 

\noindent In \cite{MR2745211} F. Lin and C. Wang develop some uniqueness results in the homogeneous setting of \eqref{Navier_Stokes_system_1}. In the two dimensional case they prove that uniqueness holds provided that $u$ belongs to $L^\infty_t L^2_x\cap L^2_t H^1_x$ (where with the notation $L^p_t L^q_x$ stands for $L^p(\RR_+, L^q(\RR^N))$), $\nabla \Pi\in L^{4/3}_t L^{4/3}_x$ and $d\in L^\infty_t\dot{H}^1_x\cap L^2_t \dot{H}^2_x$. In the three dimensional case they prove a similar result on the condition $u\in 
L^\infty_t L^2_x\cap L^2_t H^1_x\cap C([0,T), L^N_x)$, $\Pi\in L^{N/2}_t L^{N/2}_x$ and 
$d\in L^2_t \dot{H}^1_x\cap C([0,T), \dot{W}^{1,N})$, where $\dot{H}^1_x$ and $\dot{W}^{1,N}$ denote the 
homogeneous Sobolev spaces on $\RR^N$.

\noindent In \cite{MR3035314}, Zhou, Fan and Nakamura establish an existence and uniqueness result of the 
two dimensional inhomogeneous \eqref{Navier_Stokes_system_1} on a smooth bounded domain, for large  initial velocity $u_0$ and small $\nabla d_0$ in $L^2$, while the initial density $\rho_0$ is supposed to be smooth enough, namely in $W^{1,r}(\Omega)$, with $r\in (2,\infty)$. 

\noindent Recently, Hieber, Nesensohn, Pr\"uss and Schade \cite{Hieber2014} have developed a complete dynamic Theory for the homogeneous sub-system of \eqref{Navier_Stokes_system_1} (namely with constant density) in a bounded domain $\Omega$ of $\RR^N$ with a $C^2$-boundary $\partial \Omega$. Their approach is to consider such system as a quasilinear parabolic evolution equation, proving the existence and uniqueness of stronq solutions on a maximal time interval. They also show that the equilibria are normally stable, i.e. for an initial data close to their set, there exists a global solution which converges exponentially in time to an equilibrium. Moreover they ascertain the analytic regularity of their solutions.

\begin{remark}
	In order to simplify the nonlinear term $|\nabla d|^2 d$ in the molecular orientation equation, which 
	allows the constraint $d\in \SSS^{N-1}$, in \cite{MR1329830} Lin and Liu have introduced a penalty 
	approximation of Ginzburg-Landau type. More precisely they have replaced the classical energy 
	$1/2\|\nabla d\|_{L^2(\Omega)}^2$ by the following one:
	\begin{equation*}
		\int_{\Omega} \Big( \frac{1}{2}|\nabla d|^2 + \frac{(1-|d|^2)^2}{4 \ee^2}\Big)\dd x,
	\end{equation*}
	where $\ee$ is a positive parameter. It is still a challenging problem to prove the convergence of 
	the approximate solutions to the original one as $\ee$ goes to zero. However, in this paper we are going 
	to prove global solvability results for system \eqref{Navier_Stokes_system_1} with the direct constraint 
	$d\in \SSS^{N-1}$.
\end{remark}

\noindent At first, let us observe that system \eqref{Navier_Stokes_system_1} contains the incompressible inhomogeneous Navier-Stokes equations (imposing the molecular orientation field to be constant), thus we cannot expect to obtain better results than those of this sub-system. We mention the paper of Huang, Paicu, and Zhang \cite{MR3056619} where the authors establish existence and uniqueness of solutions in the whole space and moreover the paper of Danchin and Zhang \cite{MR3250369} where similar results are obtained in the half space setting. In this paper we aim to develop analogous Theorems of \cite{MR3056619} to the liquid crystal framework.

\vspace{0.2cm}
\noindent There is no loss of generality if in \eqref{Navier_Stokes_system_1} we consider the constant viscosity $\nu=1$. In the same line we impose the constants $\lambda$ and $\gamma$ to be $1$, for the convenience of the reader. When the density function $\rho>0$ assumes a value different from $0$, we can define $a:=1/\rho-1$ and reformulate 
system \eqref{Navier_Stokes_system_1} by
\begin{equation}\label{Navier_Stokes_system}
	\begin{cases}
		\;\partial_ta + \Div\, (a u)=0											& \RR_+ \times\RR^N,\\
		\;\partial_t u + u\cdot \nabla u +(1+a)\big\{ \nabla\Pi-\Delta u\big\}
		 =-\Div\big(\nabla d\odot \nabla d\big)													& \RR_+ \times\RR^N,\\
		 \;\Div\, u = 0															& \RR_+ \times\RR^N,\\
		\; \partial_t d + u\cdot \nabla d-\Delta d= |\nabla d|^2 d					& \RR_+ \times\RR^N,\\
		\;|d|=1																	& \RR_+ \times\RR^N,\\		
		\;(u,a, d)_{|t=0} = (u_0,\,a_0, \, d_0)									& \;\;\quad \quad\RR^N,\\
	\end{cases}
\end{equation}

\begin{remark}
The liquid crystal system \eqref{Navier_Stokes_system} holds a scaling property, just as the classical Navier-Stokes one. Namely, if $(a,\,u,\,d,\,\nabla \Pi)$ solves \eqref{Navier_Stokes_system} with initial 
data $(a_0,\,u_0,\,d_0)$, then for 
every positive $\lambda$, 
\begin{equation*}
	(a,\,u,\,d,\,\nabla \Pi)_\lambda := 
	(a(\lambda^2 t,\, \lambda x), \,\lambda u(\lambda^2 t,\, \lambda x),\, d(\lambda^2 t,\, \lambda x),\,
	\lambda^2 \nabla \Pi(\lambda^2 t,\, \lambda x))
\end{equation*}
is still solution with initial data $(a_0(\lambda x),\, \lambda u_0(\lambda x), \, d_0(\lambda x))$. Hence 
it is natural to consider the initial data in a Banach space which has an invariant norm under the previous scaling. Moreover, let us note that $\nabla d_0$ has the same scaling property has $u_0$, thus it is natural to impose them in the same functional space. An example of invariant by scaling space is $(a_0,\,u_0,\,\nabla d_0)\in L^\infty_x\times \BB^{N/p-1}_{p,r}\times \BB^{N/p-1}_{p,r}$, where $\BB^{N/p-1}_{p,r}$ stands for the homogeneous Besov space (see the next section for more details and for the definition of Besov spaces). We are going to consider initial data of this type and we still remark that the case of bounded density can also include discontinuities along an interface, which is an important physical case describing a mixture of fluids with different densities, filled with crystals.
\end{remark}

\noindent In this paper we will consider initial data of the following type:
\begin{equation}\label{initial_data}
	a_0\in L^\infty_x,\quad (u_0,\,\nabla d_0)\in \BB_{p,r}^{\frac{N}{p}-1} \quad\text{with}\quad
	d_0:\RR^N\rightarrow \SSS^{N-1},
\end{equation}
where $\BB_{p,r}^{N/p-1}$ is the critical homogeneous Besov space, with indexes $1<r<\infty$ and $1<p<N$.  
From here on, we suppose that our initial data verify the following smallness condition:
\begin{equation}\label{smallness_condition}
	\eta := 
	\| a_0 \|_{L^\infty_x}+
	\|u_0\|_{\BB_{p,r}^{\frac{N}{p}-1}}+
	\|\nabla d_0\|_{\BB^{\frac{N}{p}-1}_{p,r}}
	\leq c_0,
\end{equation}
where $c_0$ is a positive constant, small enough. The index of integrability $p$ is supposed to be in $(1,N)$ and the value of $r$ in $(1,\infty)$. As in the case of the inhomogeneous Navier-Stokes equation, we can not assume $u$ with a better regularity than $\tilde{L}^1_t \BB_{p,r}^{N/p+1}$ (see \cite{MR2768550} for a complete explanation of such space). Hence, the product $a\Delta u$ between $L^\infty_x$ and $\BB_{p,r}^{N/p-1}$ assumes a distributional sense only if $p<N$, and this explains the restriction for $p$. If the index $r$ is supposed to be equal to $1$ then we expect to obtain a velocity field to be in $L^1_t \mathcal{L}ip_x$ 
(where $L^1_t \mathcal{L}ip_x$ stands for $L^1(\RR_+, \mathcal{L}ip(\RR^N))$) which is very useful to solve the transport equation on the density by Lagrangian coordinates. Our condition $r>1$ is general enough to include the case of non-Lipschitz velocity field.

\vspace{0.2cm}
\noindent Before introducing our main Theorems, let us explain the meaning of weak solution for system \eqref{Navier_Stokes_system}.
\begin{definition}
	We define $(a,\,u,\,d)$ a weak solution for \eqref{Navier_Stokes_system} if $|d|=1$ almost everywhere and
	\begin{itemize}
		\item[$\rhd$] for any test function $\phi\in C^\infty_c (\RR_+\times \RR^N)$ the following equalities 
		are well-defined and fulfilled:
					\begin{equation*}
					  \int_{\RR_+\times \RR^N} 
					  	a(t,x)\left(
					  		\partial_t \varphi(t,x) + u(t,x)\cdot \nabla \varphi(t,x)
					  	 \right)\dd t\dd x  + 
					  	 \int_{\RR^N}a_0(x) \varphi(0,x)\dd x = 0
					\end{equation*}
					\begin{equation*}
						\int_{\RR^N} u\cdot \nabla \varphi = 0.
					\end{equation*}
		\item[$\rhd$] for any vector valued function $\Phi=(\Phi_1,\dots,\Phi_N)\in 
					C^\infty_c (\RR_+\times \RR^N)$ the following identities 
					are well-defined and satisfied:
					\begin{equation*}
						\int_{\RR_+\times \RR^N} 
							u\cdot \partial_t\Phi + 
						\left\{	
							u \cdot \nabla u + 
							(1+a)(\nabla \Pi - \Delta u)  
						\right\}\cdot \Phi + 
						(\nabla d\odot \nabla d)\cdot \nabla \Phi + 
						\int_{\RR^N}u_0 \cdot \Phi(0,\cdot) = 0, 
					\end{equation*}
					\begin{equation*}
						\int_{\RR_+\times \RR^N} 
							d\cdot \partial_t\Phi + 
						\left\{	
							u \cdot \nabla d 
							- \Delta d - 
							|\nabla d|^2 d  
						\right\}\cdot \Phi + 
						\int_{\RR^N}d_0 \cdot \Phi(0,\cdot) = 0,	
					\end{equation*}
	\end{itemize}
\end{definition}

\subsubsection*{The functional Framework: the smooth case}
The maximal regularity Theorem (see Theorem \ref{Maximal_regularity_theorem}) and the characterization of the homogeneous Besov spaces (see Theorem \ref{Characterization_of_hom_Besov_spaces}) play an important role for the study of \eqref{Navier_Stokes_system}, since we can reformulate the second and the third equations of \eqref{Navier_Stokes_system} in the following integral form:
\begin{equation*}
\begin{aligned}
	u(t) &= e^{t\Delta}u_0 + \int_0^t e^{(t-s)\Delta}\big\{ -u\cdot \nabla u + (1+a)\nabla \Pi+a\Delta u
	-\Div\big(\nabla d\odot \nabla d\big)	\big\}(s)\dd s,\\
	d(t) &= e^{t\Delta}d_0 + \int_0^t  e^{(t-s)\Delta}\big\{-u\cdot\nabla d +|\nabla d|^2 d\big\}(s)
	\dd s.
\end{aligned}
\end{equation*}
It is reasonable to suppose the solution having the same regularity as for the linear heat equation given by the heat kernel convoluted with the initial data. Moreover, due to the low regularity of the initial density, which is supposed to be a general bounded function, the transport equation on the density forces us to suppose $a$ only bounded. The classical maximal regularizing effect for heat kernel (see Theorem \ref{Maximal_regularity_theorem}) suggests us to look for a solution in a $L^{r}_t L^q_x$ setting. Now in the simpler case where $u$ just solves the heat equation with initial data $u_0$, having $\Delta u$ in 
$L^{r}_tL^q_x$ is equivalent to $u_0\in \BB_{q,r}^{N/q-1}$ on the condition 
$N/q-1=2-2/r$ (see Corollary \ref{Cor_Characterization_of_hom_Besov_spaces}). From the immersion 
$\BB_{p,r}^{N/p-1}\hookrightarrow \BB^{N/q-1}_{q,r}$ for every $q\geq p$,  we understand that this strategy requires $p\leq Nr/(3r-2)$. Furthermore, since the velocity field $u$ may be seen as solution of the Stokes equation
\begin{equation*}
	\partial_t u -\Delta u +\nabla \Pi = -u\cdot \nabla u +a(\Delta u -\nabla \Pi) -
	\Div\{\nabla d\odot \nabla d\},
\end{equation*}
it turns out that
\begin{equation*}
	\|(\partial_t u,\,\nabla^2 u,\,\nabla\Pi)\|_{L^{\bar{r}}_tL^q_x}\lesssim 
	\|u_0\|_{\BB_{q,\bar{r}}^{\frac{N}{q}-1}} 
	+\|(u\cdot\nabla u,\, a\Delta u,\,a\nabla \Pi,\,\Div\{\nabla d\odot \nabla d\})
	\|_{L^{\bar{r}}_tL^q_x},
\end{equation*}
Here the first relation we expect between the regularities of $u$ and $d$, namely $u\cdot\nabla u= \Div(u\otimes u)$, $\Delta u$ and $\Div\{\nabla d\odot \nabla d\}$ in the same $L^{r}_tL^q_x$ space. 
Thus, according to the previous remark, it is natural to look for a solution such that $u$ and $\nabla d$ fulfils the same functional properties and this explains why we suppose $\nabla d_0$ and $u_0$ belonging to the same critical Besov space 
$\BB_{p,r}^{N/p-1}$. 

\vspace{0.2cm}
\noindent According to the above heuristics, imposing $q=Nr/(3r-2)$, we aim to find a solution in the following space: $a,\,d\in L^{\infty}_{t,x}$ and $(u,\,\nabla d,\,\nabla \Pi)\in \mathfrak{X}_{T,r}$, where
\begin{equation*}
\begin{aligned}
	\X_{T,r}:=\big\{\,(u,\, \nabla d, \nabla\Pi) \quad&\text{with}\quad 
	\nabla d\in L^{3r}_t L^{\frac{3Nr}{3r-2}}_x,\quad (u,\,\nabla d)\in 
	 L^{2r}_TL^{\frac{Nr}{r-1}}_x,\\
	&\nabla (u,\,\nabla d)\in L^{2r}_TL^{\frac{Nr}{2r-1}}_x\cap L^{r}_t L^{\frac{Nr}{2(r-1)}}_x,
	\quad(\nabla^2u,\,\nabla^3 d,\,
	 \,\nabla\Pi)\in  L^{r}_TL^{\frac{Nr}{3r-2}}_x\,\big\}.
\end{aligned}
\end{equation*}
We also  define the following norm
\begin{equation*}
\begin{split}
	\|(u,\,\nabla d, \nabla \Pi)\|_{\X_{r,T}}:=
	\| \nabla d \|_{L^{3r}_T L^{\frac{3Nr}{3r-2}}_x}&+
	\| \nabla( u,\, \nabla d) \|_{L^{2r}_T L^{ \frac{Nr}{2r-1}}_x}+
	\| \nabla( u,\, \nabla d) \|_{L^{r}_T L^{ \frac{Nr}{2(r-1)}}_x}+\\&+
	\| (u,\, \nabla d) \|_{L^{2r}_T L^{\frac{Nr}{r-1}}_x}+ 
	\|(\nabla^2 u,\,\nabla^3 d,\, \nabla \Pi)\|_{L^{r}_T L^{\frac{Nr}{3r-2}}_x},
\end{split}
\end{equation*}
and impose $\X_{r}=\X_{r,\infty}$. Thus, our first result reads as follows:
\begin{theorem}\label{Main_Theorem}
	Let $1<r< 2$ and $p\in (1,Nr/(3r-2)]$. Suppose that the initial data $(a_0,\,u_0,\,d_0)$ are determined by  
	\eqref{initial_data}. 
	There exists a positive constant $c_0$ such that, if \eqref{smallness_condition} is fulfilled, 
	then there exists a global weak solution $(a,\, u,\,  d,\, \nabla \Pi)$  of \eqref{Navier_Stokes_system}, 
	such that $(u,\,\nabla d,\nabla \Pi)\in \X_r$, $(u,\, \nabla d)\in L^2_tL^\infty_x$ and 
	$a,\,d \in L^\infty_{t,x}$. 
	Furthermore $\|a\|_{L^\infty_{t,x}}\leq \|a_0\|_{L^\infty_{x}}$ and the 
	following inequality is satisfied:	
	\begin{equation}\label{main_theorem_inequality}
	\begin{split}
		\| (u,\,\nabla d,\,\nabla \Pi) \|_{\X_r}+
		\| (u,\, \nabla d) \|_{L^{2}_t L^{\infty}_x}
		\lesssim \eta.
	\end{split}
	\end{equation}
\end{theorem}
\begin{remark}
	In this first Theorem we have supposed the constriction $1<r< 2$. To explain this condition, 
	we anticipate that the proof will be based on an iterate scheme which includes the following one:
	\begin{equation*}
		\partial_t d^n  -\Delta d^n= 
		|\nabla d^{n-1}|^2 d^{n-1}-u^{n-1}\cdot \nabla d^{n-1}. 
	\end{equation*}
	The condition $(u^n, \, \nabla d^n)_\NN\subset \X_r$ is not enough to easily control the 
	$L^{\infty}_{t,x}$-norm of $(d^n)_\NN$. Thus, we have added the condition 
	$(u^n, \, \nabla d^n)_\NN\subset L^2_tL^\infty_x$ which requires $(u_0,\,\nabla d_0)\in 
	\BB_{\infty,2}^{-1}$. We conjecture that such restriction is not necessary, however we 
	have imposed it to simplify the proof for the reader. Indeed the case 
	$r\geq 2$ is treated in our second result, Theorem \ref{Main_Theorem_2}.
\end{remark}

\subsubsection*{The functional framework: the general case}
As we have already pointed out, the choice of a $L^{\bar{r}}_t L^q_x$ functional setting requires the condition 
$p\leq Nr/(3r-2)$. The more general case $1<p<N$ can be handled by the addiction of a weight in time. Indeed the simpler case where $u$ just solves the heat equation with initial data $u_0$, having $u_0\in \BB_{p_1, \bar{r}}^{N/p_1-1}$ for some $p_1\in [p,N)$ and $\bar{r}\in [r, \infty]$ is equivalent to impose $t^{1/2(3-N/p_1)}\Delta u$ in $L^{\bar{r}}(\RR_+,\dd t/t)$. Hence, with similar heuristics proposed in the first case, adding such weights in time, we aim to find a solution in the following functional space: $a,\,d\in L^{\infty}_{t,x}$ and $(u,\,\nabla d,\,\nabla \Pi)\in \mathfrak{Y}_{T,r}$, where $\Y_{r,T}=\Y_{r,T}(p_1,\,p_2,\,p_3)$ is the set of 
$(u,\,\nabla d,\,\nabla \Pi)$ such that
\begin{equation*}
\begin{aligned}
	\Big\{ 
	t^{\gamma_1 }&		   (u,\, \nabla d) 	 	\in L^{2r    }_t	L^{p_3			}_x,\quad
	t^{\gamma_2 }		   (u,\, \nabla d) 	 	\in L^{\infty}_t 	L^{p_3			}_x,\quad
	t^{\gamma_1 }		   (u,\, \nabla d)	 	\in L^{2r    }_t	L^{p_3			}_x,\quad
	t^{\gamma_2 }		   (u,\, \nabla d)	 	\in L^{\infty}_t 	L^{p_3			}_x,\quad\\&
	t^{\gamma_3 }     		     \nabla d	 	\in L^{2r    }_t	L^{3 p_1		}_x,\quad
	t^{\gamma_4 }			     \nabla d	 	\in L^{\infty}_t 	L^{3 p_1		}_x,\quad
	t^{\beta_1	}\nabla	   (u,\, \nabla d) 	 	\in L^{2r    }_t	L^{p_2			}_x,\quad
	t^{\beta_2	}\nabla	   ( u,\,\nabla d)	 	\in L^{\infty}_t 	L^{p_2			}_x,\quad\\&\hspace{1.5cm}
	t^{\beta_3	}\nabla	   ( u,\,\nabla d) 	 	\in L^{2r	 }_t 	L^{\frac{p_3}{2}}_x,\quad
	t^{\beta_4  }\nabla	   ( u,\,\nabla d) 	 	\in L^{\infty}_t 	L^{\frac{p_3}{2}}_x,\quad
	t^{\alpha_1 }\nabla^2  ( u,\,\nabla d,\,\Pi)\in L^{2r}	  _t	L^{p_1}			 _x
	\Big\}.
\end{aligned}
\end{equation*}
where we have fixed $\max\{p,Nr/(2r-1)\}<p_1<N$, $\,Nr/(r-1)<p_3\leq \infty\,$ and $p_2$ such that $\,1/p_1=1/p_2+1/p_3$. Furthermore, the weight in time exponents are defined by
\begin{equation*}
	\begin{array}{llll}
	\alpha_1:=\frac{1}{2}\big(3-\frac{N}{p_1} \big) -\frac{1}{2r},\quad
	&\beta_1:=\frac{1}{2}\big(2-\frac{N}{p_2} \big)	-\frac{1}{2r},\quad
	&\gamma_1:=\frac{1}{2}\big(1-\frac{N}{p_3} \big)-\frac{1}{2r},\quad
	&\gamma_2:=\frac{1}{2}\big(1-\frac{N}{p_3} \big),\\
	\alpha_2:=\frac{1}{2}\big(3 -\frac{N}{p_1}\big)	-\frac{1}{r},\quad
	&\beta_2:=\frac{1}{2}\big(2 -\frac{N}{p_2}\big)	,\quad
	&\gamma_3:=\frac{1}{2}\big(1-\frac{N}{3p_1}\big)-\frac{1}{2r},\quad
	&\gamma_4:=\frac{1}{2}\big(1-\frac{N}{3p_1}\big),\\
	&\beta_3:=\frac{1}{2}\big(2-\frac{2N}{p_3} \big)	-\frac{1}{2r},\quad
	&\beta_4:=\frac{1}{2}\big(2-\frac{2N}{p_3} \big),
	\quad
	\end{array}
\end{equation*}
We also denote by $\|(u,\,d,\,\nabla \Pi)\|_{\Y_{r,T}}$ the following norm: 
\begin{equation*}
\begin{split}
	\|t^{\beta _1} &\nabla (u,\,  \nabla d)				   \|_{ L^{2r}    _T L^{p_2}		  _x} + 
	\|t^{\beta _2}  \nabla (u,\,  \nabla d)				   \|_{ L^{\infty}_T L^{p_2}		  _x} +
	\|t^{\beta _3}  \nabla (u,\,  \nabla d)				   \|_{ L^{2r}    _T L^{\frac{p_3}{2}}_x} +\\&+
	\|t^{\beta _4}  \nabla (u,\,  \nabla d)				   \|_{ L^{\infty}_T L^{\frac{p_3}{2}}_x} +
	\|t^{\gamma_1}         (u,\,  \nabla d) 			   \|_{ L^{2r}    _T L^{p_3}		  _x} +
	\|t^{\gamma_2}         (u,\,  \nabla d)                \|_{ L^{\infty}_T L^{p_3}		  _x} + 
	\|t^{\gamma_3}				  \nabla d 				   \|_{L^{2r}     _T L^{3 p_1}		  _x} +
	\\&\hspace{1cm}+	
	\|t^{\gamma_4}                \nabla d 				   \|_{L^{\infty} _T L^{3 p_1}		  _x} +      
	\|t^{\alpha_1}(\nabla^2 u,\,  \nabla \Pi)              \|_{L^{2r}     _T L^{p_1}		  _x} +
	\|t^{\alpha_2}( \nabla^2 u,\,\nabla^3 d,\, \nabla \Pi) \|_{L^{r}      _T L^{p_1}		  _x} 	
\end{split}
\end{equation*}
and impose $\Y_{r}:=\Y_{r,\infty}$. Hence our second and more general result concerning the existence of a solution reads as follows:
\begin{theorem}\label{Main_Theorem_2}
	Let $1<r<\infty$ and $p\in (1,N)$. Suppose that the initial data $(a_0,\,u_0,\,d_0)$ are determined by  
	\eqref{initial_data}. 
	There exists a positive constant $c_0$  such that, if \eqref{smallness_condition} is fulfilled, 
	then there exists a global weak-solution $(a,\, u,\, d)$  of \eqref{Navier_Stokes_system}, 
	such that $(a,\,d) \in L^\infty_{t,x}$ and $(u,\,\nabla d, \nabla \Pi)$ belongs to $\Y_r$.
	Furthermore $\|a\|_{L^\infty_{t,x}}\leq \|a_0\|_{L^\infty_{x}}$, $|d(t,x)|=1$ for almost every 
	$(t,x)\in \RR_+\times \RR^N$ and the following 
	inequality is satisfied:	
	\begin{equation*}
	\begin{split}
		\|(u,\,\nabla d,\,\nabla \Pi)\|_{\Y_r} &\lesssim\eta
	\end{split}
	\end{equation*}
\end{theorem}

\subsubsection*{Uniqueness}
In order to recover the uniqueness of the constructed global weak-solutions, we need to add an extra regularity on the initial data for the velocity field and the director field. Namely we add to \eqref{initial_data} the following hypotheses
\begin{equation*}
	(u_0,\,\nabla d_0)\in \BB_{p,r}^{\frac{N}{p}-1+\ee},
\end{equation*}
for a sufficient small positive constant $\ee$. With this extra-regularity, we are able to obtain the velocity field $u$ to be in $L^1_{t, loc} \mathcal{L}ip_x$. This allows us to reformulate system \eqref{Navier_Stokes_system_1} in Lagrangian coordinates. Such coordinates simplify in some way our problem, granting the density $a$ to be constant, since it is governed by a transport equation. 
Therefore, we proceed in the same line of \cite{MR3056619}, proving the uniqueness of the constructed solutions for a initial time interval. Thus we conclude by a boot strap method in order to recover the global uniqueness.

\noindent First, let us introduce the functional frameworks. Fixing the value of $\ee$ in $(0,1)$ and let us assume also the constriction $r<2/(2-\ee)$. We define the following space:
\begin{equation*}
\begin{split}
	\X_{r}^\ee:=\big\{\,(u,\, \nabla d, \nabla\Pi) \quad&\text{with}\quad 
	\nabla d\in L^{3r}_t L^{\frac{3Nr}{(3-\ee)r-2}}_x,\quad \nabla (u,\,\nabla d)\in 
	 L^{\frac{2}{2-\ee}}_tL^{\infty}_x,\\
	&\nabla (u,\,\nabla d)\in L^{2r}_tL^{\frac{Nr}{(2-\ee)r-1}}_x,
	\quad(\nabla^2u,\,\nabla^3 d,\,
	 \,\nabla\Pi)\in  L^{r}_tL^{\frac{Nr}{(3-\ee)r-2}}_x\,\big\}.
\end{split}
\end{equation*}
\begin{remark}
 Let us immediately remark that for $r<2/(2-\ee)$ we get $Nr/((3-\ee)r-2)>N$. Thus, it will 
 be possible to apply the Sobolev inequality in order to get the velocity field to be Lipschitz in space, 
 which plays an important role, as we have already mentioned.
\end{remark}

\noindent The first uniqueness result of this article reads as follows:

\begin{theorem}\label{Main_Theorem_3}
	Let $\ee$ be a positive constant in $(0,1]$. Suppose that the hypotheses of Theorem \ref{Main_Theorem} are 
	satisfied with $r<2/(2-\ee)$  
	and let $(a,\,u,\,d)$ be the 
	solution generated. Let us assume that $(u_0,\,\nabla d_0)$ also belongs to $\BB_{N,p}^{N/p-1+\ee}$,  
	then we have $(u,\,\nabla d,\,\nabla \Pi)\in \X_r^\ee$ with
	\begin{equation*}
		\|(u,\,\nabla d,\,\nabla \Pi)\|_{\X_r^\ee}\lesssim\|(u_0,\nabla d_0)\|_{\BB_{p,r}^{\frac{N}{p}-1+\ee}}
		+\|\nabla d_0\|_{\BB_{p,r}^{\frac{N}{p}-1+\frac{2}{3}\ee}}.
	\end{equation*}
	and the uniqueness holds in this functional framework.
\end{theorem}

\noindent As we have already exposed, the results of Theorem \ref{Main_Theorem} and \ref{Main_Theorem_3} require the constriction $1<p\leq Nr/(3r-2) $. Here, the existence and the uniqueness hold in a setting of type $L^{\bar{r}}_t L^q_x$. However, to recover the uniqueness for the general case $1<p<N$, we need again to add a weight in time. More precisely, fixing $q_1\in (N,\,N/(1-\ee))$, $q_3>Nr/((1-\ee)r-1)$ and imposing $q_2$ such that $1/q_1 = 1/q_2 + 1/q_3$, we define the following space:
\begin{equation*}
\begin{aligned}
	\Y_{r,T}^\ee=\Y_{r,T}^\ee(\,q_1,\,q_2,\,q_3): = 
	\Big\{
	(u,\, \nabla d, \nabla\Pi) \quad\text{with}\quad 
	t^{\gamma_1^\ee }		( u,\,	\nabla d)				\in L^{2r}_t     L^{q_3}_x,\quad
	t^{\gamma_2^\ee } 		(u,\, 	\nabla d)				\in L^{\infty}_t L^{q_3}_x,\quad\\
	t^{\gamma_1^\ee } 		(u,\, 	\nabla d)				\in L^{2r}_t     L^{q_3}_x, \quad
	t^{\gamma_2^\ee } 		(u,\, 	\nabla d)				\in L^{\infty}_t L^{q_3}_x,\quad
	t^{\gamma_3^\ee } 				\nabla d				\in	L^{2r}_t     L^{3 q_1}_x,\quad
	t^{\gamma_4^\ee } 				\nabla d				\in	L^{\infty}_t L^{3 q_1}_x,\quad\\
	t^{\beta_1^\ee  } \nabla (u,\,	\nabla d) 				\in L^{2r}_t 	 L^{q_2}_x,\quad
	t^{\beta_2^\ee  } \nabla ( u,\, \nabla d) 				\in L^{\infty}_t L^{q_2}_x,\quad
	t^{\alpha_1^\ee } (\nabla^2  u,\,\nabla^3 d,\,\nabla\Pi)\in L^{2r}_t	 L^{q_1}_x
	\Big\},
\end{aligned}
\end{equation*}
where the exponents of the weights in time are defined by
\begin{equation*}
	\begin{array}{lll}
	\alpha_1^\ee:=\frac{1}{2}\big(3-\frac{N}{q_1}-\ee \big) -\frac{1}{2r},\quad	
	&\beta_1^\ee:=\frac{1}{2}\big(2-\frac{N}{q_2}-\ee \big)	-\frac{1}{2r},\quad
	&\gamma_1^\ee:=\frac{1}{2}\big(1-\frac{N}{q_3}-\ee \big)-\frac{1}{2r},
	\\
	\alpha_2^\ee:=\frac{1}{2}\big(3 -\frac{N}{q_1}-\ee\big)	-\frac{1}{r},\quad
	&\beta_2^\ee:=\frac{1}{2}\big(2 -\frac{N}{q_2}-\ee\big)	,\quad
	&\gamma_2^\ee:=\frac{1}{2}\big(1-\frac{N}{q_3}-\ee \big),\\
	\gamma_3^\ee:=\frac{1}{2}\big(1-\frac{N}{3q_1}-\ee\big)-\frac{1}{2r}
	&\gamma_4^\ee:=\frac{1}{2}\big(1-\frac{N}{3q_1}-\ee\big).
	&
	\end{array}
\end{equation*}

\noindent Therefore, our main uniqueness result reads as follows:

\begin{theorem}\label{Main_Theorem_4}
	Let us assume that the hypotheses of Theorem \ref{Main_Theorem_2} are satisfied and suppose also that 
	$(u_0,\,\nabla d_0)\in \BB_{p,r}^{N/p-1+\varepsilon}$ for a positive $\varepsilon$ bounded by 
	$\min\{1/r,\,1-1/r,\,N/p-1\}$. Then the solution $(a,\,u,\,d,\nabla\Pi)$ determined by Theorem 
	\ref{Main_Theorem_2} fulfils also $(u,\,\nabla d,\,\nabla \Pi)\in \Y_r^\ee$ and we have 
	\begin{equation*}
		\|(u,\,\nabla d,\,\nabla \Pi)\|_{\Y^\ee_r}\lesssim 
		\|(u_0,\,\nabla d_0)\|_{\BB_{p,r}^{\frac{N}{p}-1+\ee}}.
	\end{equation*}
	Moreover, such solution is unique in this functional framework.
\end{theorem}

\noindent Let us briefly describe the structure of this paper. In the next paragraph we briefly recall some properties and characterizations about Besov Spaces, while we prove in detail some technical Lemmas and Theorems concerning the regularizing effects for the heat kernel both in the functional framework with and without weight in time. In the third section we prove the existence of solutions for system \eqref{Navier_Stokes_system} on the condition of more regular initial data with respect to \eqref{initial_data}. Such results will play an important role in the proofs of our main results, both for the existence part, regularizing the initial data, and the uniqueness part, where we are able to reformulate \eqref{Navier_Stokes_system} by Lagrangian coordinates. The fourth section is devoted to the proof of the existence part, namely Theorem \ref{Main_Theorem} and Theorem \ref{Main_Theorem_2}. Regularizing the initial data we construct a sequence of approximate solutions and we pass to the limit thanks to some uniform estimates. In the fifth section we present the uniqueness results. We suppose the initial data with a little bit more regularity, which allows us to obtain the Lagrangian coordinates. Thus, we are able to prove the uniqueness of the solution for system \eqref{Navier_Stokes_system} in a small initial time-interval. Then we conclude by a bootstrap method, obtaining Theorem \ref{Main_Theorem_3} and \ref{Main_Theorem_4}. Finally in the appendix, for the convenience of the reader, we prove some technical results which are useful in the main proofs.
\section{Preliminaries}
\noindent
This section is devoted to the study of several regularizing effects for the heat kernel, which will be useful for the proof of the main Theorems. At first step let us recall the well-known Hardy-Littlewood-Sobolev inequality.
\begin{theorem}[Hardy-Littlewood-Sobolev inequality]\label{HLS_Theorem}
	Let  $f$ belongs to $L^p_x$, with $1< p <\infty $, $\alpha\in ]0,N[$ and suppose $r\in ]0,\infty[$ satisfies
	\begin{equation*}
		\frac{1}{p}+\frac{\alpha}{N}= 1+\frac{1}{r}\text{.}
	\end{equation*}
	Then
	\begin{equation*}
		|\cdot |^{-\alpha}* f \in L^r_x\text{,}
	\end{equation*}
	and there exists a positive constant $C$ such that
	\begin{equation*}
		\left\|
			|\cdot |^{-\alpha}* f
		\right\| _{L^r_x}
		\leq 
		C
		\left\| 
			f
		\right\|_{L^p_x}\text{.}
	\end{equation*}		
\end{theorem}

\noindent
Let us now enunciate the well-known $L^pL^q$-Maximal Regularity Theorem, whose proof is available in \cite{MR1938147}.  

\begin{theorem}[Maximal $L^p_tL^q_x$ regularity for the heat kernel]\label{Maximal_regularity_theorem}
	Let $T\in ]0,\infty]$, $1<p,q<\infty$ and $f\in L^p(0,T;L^q_x)$. Let the operator $\Aa$ be 
	defined by
	\begin{equation*}
		\Aa f(t,\cdot):=\int_0^t \Delta e^{(t-s)\Delta}f(s,\cdot)\dd s\text{.}
	\end{equation*}
	Then $\Aa$ is a bounded operator from $L^p(0,T;L^q_x)$ to $L^p(0,T;L^q_x)$.
\end{theorem}

\begin{lemma}\label{Lemma2}
	Let $T\in\, ]0,\infty]$, $1<r_1, \,r_2<\infty$, $q_1\in [1,\infty]$ and 
	$q_2\in [q_1, \infty]$, such that 
	\begin{equation}\label{condition_lemma2a}
		\frac{N}{2}\Big(\frac{1}{q_1}-\frac{1}{q_2}\Big)+ \frac{1}{2}<1. 
	\end{equation}
	Let the operator $\mathcal{B}$ Let the operator $B$ be defined by
	\begin{equation*}
		\mathcal{B}f(t,\cdot):=\int_0^t \nabla e^{(t-s)\Delta}f(s,\cdot)\dd s\text{.}
	\end{equation*}
	Then, we have that $\mathcal{B}$ is a bounded operator from 
	$L^{r_1}(0,T;L^{q_1}_x)$ with values to $L^{r_2}(0,T;L^{q_2}_x)$, if the following equality is fulfilled:
	\begin{equation}\label{condition_lemma2}
	\frac{1}{r_1} + \frac{N}{2}\Big( \frac{1}{q_1} - \frac{1}{q_2}\Big) = 
	\frac{1}{2} + \frac{1}{r_2}.
	\end{equation}
\end{lemma}
\begin{proof} 
	At first let us observe that, if $K$ denotes the heat kernel, than for all $1\leq \lambda \leq \infty$ we have
	\begin{equation}\label{estimate_nablaheatkernel}
		\|\nabla K(t,\cdot)\|_{L^\lambda_x}=\frac{1}{t^{\frac{d}{2}\frac{1}{\,\lambda'}+\frac{1}{2}}}
		\|K(1,\cdot)\|_{L^\lambda_x}\text{.}
	\end{equation}
	Observe that, for every $t\in \RR_+$,
	\begin{equation*}
		\big\|\int_0^t \nabla e^{(t-s)\Delta}f(s)\dd s\,\big\|_{L^{q_2}_x}
		\leq \int_0^t \|\nabla K(t-s,\cdot)*f(s,\cdot)\,\|_{L^{q_2}_x}\dd s
		\leq \int_0^t \|\,\nabla K(t-s)\|_{L^{\tilde{q}}_x}\|\,f(s)\,\|_{L^{q_1}_x}\dd s,
	\end{equation*}
	with $1/\tilde{q}+1/q_1=1/q_2+1$. Thus, by \eqref{estimate_nablaheatkernel}, we obtain
	\begin{equation*}
		\mathcal{B}f(t) = 
		\big[1_{[0,T]}(s)|s|^{-\frac{N}{2}\Big( \frac{1}{q_1} - \frac{1}{q_2}\Big)-\frac{1}{2}}
		*_s
		1_{[0,T]}(s)\|f(s)\|_{L^{q_1}_x}\big](t)
	\end{equation*}
	and by virtue of Theorem \ref{HLS_Theorem} we conclude the proof of Lemma \ref{Lemma2}. 
\end{proof}

\begin{lemma}\label{Lemma3}
	Let $T\in\, ]0,\infty]$, $1<r_1, \,r_2<\infty$, $q_1\in [1,\infty]$ and $q_2\in [q_1, \infty]$, 
	such that 
	\begin{equation}\label{condition_lemma3a}
		\frac{N}{2}\Big(\frac{1}{q_1}-\frac{1}{q_2}\Big)<1. 
	\end{equation} 
	Let $f\in L^{r_1}_TL^{q_1}_x$ and let the operator $\mathcal{C}$ be defined by
	\begin{equation*}
		\mathcal{C}f(t,\cdot):= \int_0^t e^{(t-s)\Delta}f(s,\cdot)\dd s\text{,}
	\end{equation*}
	Then, $\Cc $ is a bounded operator from 
	$L^{r_1}(0,T;L^{q_1}_x)$ with values to $L^{r_2}(0,T;L^{q_2}_x)$, where
	\begin{equation}\label{condition_lemma3}
	\frac{1}{r_1} + \frac{N}{2}\Big( \frac{1}{q_1} - \frac{1}{q_2}\Big) = 1 + \frac{1}{r_2}.
	\end{equation}
\end{lemma}
\begin{proof}
	The proof is basically equivalent to the previous one, observing that 
	\begin{equation}\label{estimate_heatkernel}
		\|K(t,\cdot)\|_{L^\lambda_x}=\frac{1}{t^{\frac{N}{2}\frac{1}{\,\lambda'}}}
		\|K(1,\cdot)\|_{L^\lambda_x},
	\end{equation}
	for every $1\leq \lambda\leq \infty$.
\end{proof}

\noindent
The next Theorem is a variation of Theorem \ref{Maximal_regularity_theorem} for functions which belong to some $L^p_t L^q_x$-space, up to a weight in time. This Theorem has already been presented in \cite{MR3056619} and \cite{MR3250369}, however we briefly prove it for the convenience of the reader.

\begin{theorem}\label{Maximal_regularity_Thm_weight_time}
	Let $T\in ]0,\infty]$, $1<\bar{r},q<\infty$ and  $\alpha\in (0,1-1/\bar{r})$. Let the operator 
	$\Aa$ be defined as in Theorem \ref{Maximal_regularity_theorem}. Suppose that 
	$t^\alpha f(t)$ 
	belongs to $L^{\bar{r}}(0,T;L^q_x)$. Then $t^\alpha \Aa f(t)$ belongs to $L^{\bar{r}}(0,T;L^q_x)$ and there 
	exists $C>0$ such that
	\begin{equation*}
		\|t^\alpha \Aa f(t)\|_{L^{\bar{r}}(0,T;L^q_x)}\leq C \| t^\alpha f(t)\|_{L^{\bar{r}}(0,T;L^q_x)}.
	\end{equation*}	 
\end{theorem}
\begin{proof}
	Let us define $v(t):=t^\alpha \Cc f(t)$, where $\Cc$ is defined by Lemma \ref{Lemma3}. 
	Since $v(t)$ is solution of
	\begin{equation*}
	\begin{cases}
		\partial_t v(t,x)  -\Delta v (t,x) = t^{\alpha}f(t,x) + \alpha t^{\alpha-1}f(t,x)
		&(t,x)\in \RR_+\times \RR^N,\\
		v(0,x)=0	&x\in\RR^N,
	\end{cases}
	\end{equation*}
	by Theorem \ref{Maximal_regularity_theorem}, we deduce that
	\begin{equation*}
		\|t^\alpha \Aa f(t)\|_{L^{\bar{r}}(0,T;L^q_x)} \lesssim 
		\|t^\alpha f(t)\|_{L^{\bar{r}}(0,T;L^q_x)}+
		\|\alpha t^{\alpha-1} v(t)\|_{L^{\bar{r}}(0,T;L^q_x)}.
	\end{equation*}
	Observing that 
	\begin{equation*}
		 \|t^{\alpha-1} v(t)\|_{L^q_x}\lesssim t^\alpha \int_0^t \|f(\tau)\|_{L^q_x}\dd \tau
		 = \int_0^t \Big(\frac{t}{\tau}\Big)^\alpha  \tau^\alpha \|f(\tau)\|_{L^q_x}\frac{\dd \tau}{t}
		 =\int_0^1 (\tau')^{-\alpha} (t\tau')^\alpha \|f(t \tau')\|_{L^q_x}\dd \tau',
	\end{equation*}
	for every $t\in (0,T)$, then, by virtue of Minkowski inequality and because $\alpha<1-1/\bar{r}$, 
	\begin{align*}
		\|t^{\alpha-1} v(t)\|_{L^{\bar{r}}(0,T;L^q_x} 
		&\lesssim 
		\int_0^1 (\tau')^{-\alpha} 
		\Big(
			\int_0^T  (t\tau')^{\alpha \bar{r}} \|f(t \tau')\|_{L^q_x}^{\bar{r}} \dd t 
		\Big)^{\frac{1}{\bar{r}}}\dd \tau'\\
		&\lesssim
		\int_0^1 (\tau')^{-\alpha-\frac{1}{\bar{r}}} 
		\Big(
			\int_0^{T\tau'}  (t')^{\alpha \bar{r}} \|f(t ')\|_{L^q_x}^{\bar{r}} \dd t' 
		\Big)^{\frac{1}{\bar{r}}}\dd \tau'
		\lesssim \|t^\alpha f(t)\|_{L^{\bar{r}}(0,T;L^q_x)}.
	\end{align*}			 
\end{proof}

\noindent The next two Lemmas are a particular case of Lemma \ref{Lemma4} and Lemma \ref{Lemma5}, therefore we postpone the proof to the appendix.

\begin{lemma}\label{Lemma4b}
	Let the operator $\Cc$ be defined as in Lemma \ref{Lemma3}. 
	Consider $T\in(0,\infty]$, $1<\bar{r}<\infty$, and moreover suppose that 
	$q,\,\tilde{q}$ satisfy $N/2<q<N$, 
	$\,\max\{N, q\}<\tilde{q}\leq \infty$. Let $\alpha$, $\gamma$ and 
	$\bar{\gamma}$	be defined by
	\begin{equation*}
	\alpha:=\frac{1}{2}\big(3-\frac{N}{q} \big) -\frac{1}{\bar{r}},\quad	
	\gamma:=\frac{1}{2}\big(1-\frac{N}{\tilde{q}} \big)	-\frac{1}{\bar{r}}\quad\text{and}\quad
	\bar{\gamma}:=\frac{1}{2}\big(1-\frac{N}{\tilde{q}} \big).
	\end{equation*}
	If $t^{\alpha} f(t)$ belongs to $L^{\bar{r}}(0,T;L^{q}_x)$ then $t^{\gamma^\ee}\Cc f(t)$ 
	belongs to $L^{\bar{r}}(0,T;L^{\tilde{q}}_x)$. 
	Furthermore there exist $C=C(q, \tilde{q}, \bar{r})>0$ such that
	\begin{equation*}
		\|t^{\gamma} \Cc f(t) \|_{L^{\bar{r}}(0,T; L^{\tilde{q}}_x)}
		\leq 		
		C\|t^{\alpha}  f(t) \|_{L^{\tilde{r}}(0,T; L^{q}_x)}.
	\end{equation*}
	Moreover, if $\bar{r}>2$ and $N\bar{r}/(2\bar{r}-2)<q$,  
	then $t^{\bar{\gamma}}\Cc f(t)$ belongs to 
	$L^{\infty}_T L^{\tilde{q}}_x$ and there exists a positive constant 
	$\bar{C} = \bar{C} (q, \tilde{q}, \bar{r})$ 
	such that
	\begin{equation*}
		\|t^{\bar{\gamma}} \Cc f(t) \|_{L^{\infty}_T L^{\tilde{q}}_x}\leq 
		\bar{C} \|t^{\alpha}  f(t) \|_{L^{\bar{r}}_T L^{q}_x}.
	\end{equation*} 
\end{lemma}

\begin{lemma}\label{Lemma5b}
	Let the operators $\Bb$ be defined as in Lemma \ref{Lemma2}. 
	Consider $T\in (0,\infty]$, $\ee\geq 0$ small enough, $1<\bar{r}<\infty$, and moreover suppose that 
	$q,\,\bar{q}$ satisfy $N/2<q<N$ and $q\leq \bar{q}$ such that $1/q-1/\bar{q}<1/N$. Let 
	$\alpha$ be defined as in Lemma \ref{Lemma4b} and $\beta$ and $\bar{\beta}$ be defined by
	\begin{equation*}
	\bar{\beta}:=\frac{1}{2}\big(2-\frac{N}{\bar{q}} \big)	\quad\text{and}\quad
	\beta:=\frac{1}{2}\big(1-\frac{N}{\bar{q}} \big)-\frac{1}{\bar{r}}.
	\end{equation*}
	If $t^{\alpha} f(t)$ belongs to $L^{\bar{r}}_T L^{q}_x$ then 
	$t^{\beta}\mathcal{B}f(t) $ 
	belongs to $ L^{\bar{r}}_TL^{\bar{q}}_x$ and there exists a positive constant
	$C=C (q, \bar{q}, \bar{r})$ such that
	\begin{equation*}
		\|t^{\beta} \mathcal{B} f(t) \|_{L^{\bar{r}}_T L^{\bar{q}}_x}\leq 
		C \|t^{\alpha}  f(t) \|_{L^{\bar{r}}_T L^{q}_x}.
	\end{equation*}
	Moreover, if $\bar{r}>2$, $N\bar{r}/(2\bar{r}-2)<q$ and $\bar{q}<Nr$ 
	then $t^{\bar{\beta}}\mathcal{B}f(t)$ belongs to 
	$L^{\infty}_TL^{\bar{q}}_x$ and there exists a positive constant 
	$\bar{C} = \bar{C} (q, \bar{q}, \bar{r})$ 
	such that
	\begin{equation*}
		\|t^{\bar{\beta}} \mathcal{B} f(t) \|_{L^{\infty}_T L^{\bar{q}}_x}\leq 
		\bar{C} \|t^{\alpha}  f(t) \|_{L^{\bar{r}}_T L^{q}_x}.
	\end{equation*}
\end{lemma}

\noindent
For the main properties of homogeneous Besov Spaces we refer to \cite{MR2768550}. However, let us briefly recall the definition and two important results which characterize such spaces in relation to the heat kernel. 
\begin{definition}
	Let $\chi$ be a smooth nonincreasing radial function which has support in $B(0,1)$ and such that 
	$\chi\equiv 1$ on $B(0,1/2)$. Imposing $\varphi_q(\xi) := \chi(\xi2^{-q-1}) - \chi(\xi2^{-q})$ for every 
	$q\in \ZZ$, we define 
	the homogeneous Lettlewood-Paley dyadic block $\dot{\Delta}_q$ by
	\begin{equation*}
		\dot{\Delta}_q u := \mathcal{F}^{-1}(\varphi_q \mathcal{F}u),
	\end{equation*}	 
	where $u$ is a temperate distribution and $\mathcal{F}$ is the Fourier transform on $\RR^N$.
\end{definition}
\noindent The homogeneous Besov Space is defined as follows:
\begin{definition}
	For any $s\in \RR$ and $(p,r)\in [1,\infty]^2$, let us define $\BB_{p,r}^s$ as the set of tempered 
	distribution $f$ such that
	\begin{equation*}
		\|f\|_{\BB_{p,r}^s}:= \|2^{sq}\|\dot{\Delta}_q f\|_{L^p_x}\|_{l^r(\ZZ)}
	\end{equation*}
	and for all smooth compactly supported function $\theta$ on $\RR^N$ we have
	\begin{equation*}
		\lim_{\lambda\rightarrow +\infty} \theta (\lambda D)f = 0\quad\text{in}\quad L^\infty(\RR^N).
	\end{equation*}
\end{definition}

\begin{theorem}[Characterization of Homogeneous Besov Spaces]\label{Characterization_of_hom_Besov_spaces}
	Let $s$ be a negative real number and $(p,r)\in [1,\infty]^2$. $u$ belongs to $\dot{B}_{p,r}^s$ 
	if and only if $e^{t\Delta}u$ belongs to $L^p_x$ for almost every $t\in \RR_+$ and 
	\begin{equation*}
		t^{-\frac{s}{2}}\left\|e^{t\Delta}u\right\|_{L^p_x}\in L^r\Big(\RR_+;\frac{\dd t}{t}\Big).
	\end{equation*}
	Moreover, there exists a positive constant $C$ such that
	\begin{equation*}
		\frac{1}{C}\left\| u\right\|_{\dot{B}_{p,r}^s} \leq
		\left\|\left\|t^{-\frac{s}{2}}e^{t\Delta}u\right\|_{L^p_x}
		\right\|_{L^r(\RR_+;\frac{\dd t}{t})}\leq
		C \left\| u\right\|_{\dot{B}_{p,r}^s}\text{.}
	\end{equation*}
\end{theorem}
\noindent
An immediate consequence is the following Corollary:
\begin{cor}\label{Cor_Characterization_of_hom_Besov_spaces}
	Let $p\in [1,\infty]$ and $r\in [1,\infty)$. $u$ belongs to $\dot{B}_{p,r}^{-\frac{2}{r}}$ 
	if and only if $e^{t\Delta}u\in L^r_t L^p_x$.
	Moreover, there exists a positive constant $C$ such that
	\begin{equation*}
		\frac{1}{C}\left\| u\right\|_{\dot{B}_{p,r}^{-\frac{2}{r}}} \leq
		\left\|e^{t\Delta}u
		\right\|_{L^r_tL^p_x} \leq
		C \left\| u\right\|_{\dot{B}_{p,r}^{-\frac{2}{r}}}\text{.}
	\end{equation*}
\end{cor}

\begin{theorem}\label{Theorem_embedding_Besov}
	Let $1\leq p_1\leq p_2 \leq \infty$ and $1\leq r_1\leq r_2\leq \infty$. Then for any real 
	number $s$, the space $\dot{B}_{p_1,r_1}^{s}$ is continuously embedded in 
	$\dot{B}_{p_2,r_2}^{s-N(1/p_1-1/p_2)  }$. 
\end{theorem}
\section{Smooth initial data}

\noindent In order to prove Theorem \ref{Main_Theorem} and Theorem \ref{Main_Theorem_2}, in this section we are going to establish the global existence of a solution for system \eqref{Navier_Stokes_system}, considering more regular initial data. More precisely we are going to consider an initial Lipschitz density and moreover we suppose the initial velocity and the initial director field with a little bit more regularity with respect to the one of Theorem \ref{Main_Theorem} and Theorem \ref{Main_Theorem_2}. 

\begin{prop}\label{Theorem_solutions_smooth_dates} 
	Let $\ee\in (0,\,1]$, $r\in(1,\,2/(2-\ee)]$, $p\in (\,1,\,Nr/(3r-2)\,]$. Suppose that the initial 
	condition \eqref{initial_data} is fulfilled and moreover 
	$\nabla a_0\in L^\infty_x$ and $(u_0,\,\nabla d_0) \in \BB^{N/p-1+\ee}_{p,r}$. If the smallness condition \eqref{smallness_condition} 
	holds, then \eqref{Navier_Stokes_system_1} admits a global weak-solution $(a,\, u,\,d,\, \nabla\Pi)$ 
	such that $(a,\,d)\in L^{\infty}_{t,x}$, $\nabla a\in L^\infty_{t,loc}L^\infty_x$, 
	$(u,\,\nabla d,\,\nabla \Pi)$ belongs to $\X_r\cap \X_r^\ee$ and 	$(u,\nabla d)$ belongs to 
	$L^2_t L^\infty_x$, with 
	\begin{equation*}
	\begin{aligned}
		&\|(u,\,\nabla d,\,\nabla \Pi)\|_{\X_r}+\|( u,\,\nabla d)\|_{L^{2}_tL^\infty_x} 
		\lesssim \eta\\
		&\|(u,\,\nabla d,\,\nabla \Pi)\|_{\X_r^\ee}
		\lesssim 	
		\|(u_0,\,\nabla d_0)\|_{\BB_{p,r}^{\frac{N}{p}-1+\ee}}+
		\|\nabla d_0\|_{\BB_{p,r}^{\frac{N}{p}-1+\frac{2}{3}\ee}}^3.
	\end{aligned}
	\end{equation*}
\end{prop}
\begin{proof}
The mainly idea is to proceed by iterate scheme. We solve a sequence of linear systems which comes from \eqref{Navier_Stokes_system} and we prove that their solutions converge to the one we are looking for. We set $(a^0,\, u^0, \nabla d^0,  \nabla \Pi^0)=(0,0,0,0)$ and we solve inductively the following two systems:
\begin{equation}\label{prop_smooth_dates_system1}
\begin{cases}
		\;\partial_ta^{n} +  u^{n-1}\cdot\nabla a^{n}=0				& \RR_+ \times\RR^N,\\
		\;a^n_{|t=0} = a_0																& \;\;\quad \quad\RR^N,\\
	\end{cases}
\end{equation}
\begin{equation}\label{prop_smooth_dates_system2}
	\begin{cases}
		\; \partial_t d^n  -\Delta d^n=
		|\nabla d^{n-1}|^2 d^{n-1}-u^{n-1}\cdot \nabla d^{n-1}						& \RR_+ \times\RR^N,\\
		\;\partial_t u^n + u^{n-1}\cdot \nabla u^{n}-\Delta u^n + \nabla \Pi^n =F^n	& \RR_+ \times\RR^N,\\
		\;\Div\, u^n = 0															& \RR_+ \times\RR^N,\\
		\;(u^n,d^n)_{|t=0} = (u_0,\,d_0)											& \;\;\quad \quad\RR^N,\\
	\end{cases}
\end{equation}
where $F^n$ is defined by:
\begin{equation}\label{def_Fn}
	F^n:= 	(1+a^n)\Div\big(\nabla d^n \odot \nabla d^n \big)+
			a^n(\Delta u^{n-1} - \nabla \Pi^{n-1}).
\end{equation}
\noindent
The global existence and uniqueness of a solution to \eqref{prop_smooth_dates_system1} is standard, since 
$u^{n-1}$ belongs to $L^1_{loc}\mathcal{L}ip$ and so it is possible to construct the Lagrangian coordinates. Furthermore, 
\begin{equation}\label{ineq_proof_prop_smooth_date2}
	\|a^n\|_{L^\infty_{t,x}}\leq \|a_0\|_{L^\infty_x}
\end{equation}
and
\begin{equation}\label{ineq_proof_prop_smooth_date_nablaa}
		\| \nabla a^n(t) \|_{L^\infty_x}\leq \|\nabla a_0\|_{L^\infty_x}
		\exp\Big\{\int_0^t \|\nabla u^{n-1}(s)\|_{L^\infty_x}\dd s\Big\}
\end{equation}
are fulfilled for every natural number $n$. For system \eqref{prop_smooth_dates_system2} we apply Proposition \ref{Appx_Prop_exist_sol_scheme}.

\noindent
Now let us prove by induction that the following inequalities are satisfied for every $n\in\NN$:
\begin{equation}\label{ineq_d_proof_prop_smooth_date}
	\|d^n\|_{L^\infty_{t,x}}\leq (1+C \eta)e^{C\eta},
\end{equation}
\begin{equation}\label{ineq_proof_prop_smooth_date1}
	\begin{split}
		\|(u^n,\,\nabla d^n,\,\nabla \Pi^n)\|_{\X_r} + \|(u^n,\,\nabla d^n) \|_{L^{2}_t L^\infty_x}
		\lesssim \eta,
	\end{split}
\end{equation}
where $C$ is a positive constant. We consider initially system \eqref{prop_smooth_dates_system2} and we want to estimate $d^n$. By the Mild formulation for the heat equation, we obtain that
\begin{equation}\label{prop_smooth_dates_explicit_formula_d}
	d^n(t)= e^{t\Delta}d_0 + 
	\int_0^t e^{(t-s)\Delta}\big\{ -u^{n-1}\cdot \nabla d^{n-1} + |\nabla d^{n-1}|^2 d^n\big\}(s)\dd s.
\end{equation}
Hence, we deduce
\begin{equation}\label{ineq_prop_smooth_dates_dn_bound}
\begin{aligned}
	\|d^n(t)\|_{L^\infty_x}
	&\leq 
	\|d_0\|_{L^\infty_x} + 
	\int_0^t	\|u^{n-1}(s)\|_{L^\infty_x}\| \nabla d^{n-1}(s)\|_{L^\infty_x} + 
	\|\nabla d^{n-1}(s)\|_{L^\infty_x}^2 \|d^n(s) \|_{L^\infty_x}
	\dd s\\
	&\leq 
	1+  C \eta  + \int_0^t \|\nabla d^{n-1}(s)\|_{L^\infty_x}^2 \|d^n(s) \|_{L^\infty_x}
	\dd s.
\end{aligned}
\end{equation}
Applying the Gronwall inequality and by the induction hypotheses, we obtain \eqref{ineq_d_proof_prop_smooth_date}. 
We want now to estimate $\nabla d^n$ and $u^n$. From \eqref{prop_smooth_dates_explicit_formula_d} we get 
\begin{equation}\label{prop_smooth_dates_explicit_formula_nablad}
	\nabla d^n(t)= e^{t\Delta}\nabla d_0 + 
	\int_0^t e^{(t-s)\Delta}\nabla \big\{ u^{n-1}\cdot \nabla d^{n-1} + |\nabla d^{n-1}|^2 d^n\big\}(s)\dd s
\end{equation} 
First, let us estimate $\nabla d^n$ in $L^{3r}_t L^{3Nr/(3r-2)}$. By Corollary \ref{Cor_Characterization_of_hom_Besov_spaces} and Lemma \ref{Lemma2}, with $r_1=6r/5$, $r_2=3r$, $q_1=3Nr/(6r-5)$ and $q_2=3Nr/(3r-2)$ which verify \eqref{condition_lemma2a} and \eqref{condition_lemma2},  we get 
\begin{equation}\label{prop_smooth1_estimate_dn_L3r}
	\|\nabla d^n\|_{L^{3r}_t L^{\frac{3Nr}{3r-2}}_x}\lesssim 
	\|\nabla d_0\|_{\BB_{p,r}^{\frac{N}{p}-1}} + 
	\|(u^{n-1},\nabla d^{n-1})\|_{L^{2r}_t L^{\frac{Nr}{r-1}}_x} 
	\|\nabla d^{n-1}\|_{L^{3r}_t L^{\frac{3Nr}{3r-2}}_x}\lesssim \eta
\end{equation}
Moreover, applying Corollary \ref{Cor_Characterization_of_hom_Besov_spaces}, Lemma \ref{Lemma2}, Lemma \ref{Lemma3} and Theorem \ref{Maximal_regularity_theorem}, we obtain
\begin{equation}\label{prop_smooth1_estimate_dn}
\begin{split}
	\|&\nabla d^n\|_{L^2_t L^\infty_x}+
	\| \nabla d^n\|_{L^{2r}_t L^{\frac{Nr}{r-1}}_x}+
	\|\nabla^2 d^n \|_{L^{2r}_t L^{ \frac{Nr}{2r-1}}_x}+ 
	\|\nabla^2 d^n\|_{L^{r}_t L^{\frac{Nr}{2(r-1)}}_x}+
	\|\nabla^3 d^n\|_{L^r_t L^{\frac{Nr}{3r-2}}_x}
	\lesssim\\
	&\lesssim
	\|\nabla d_0\|_{\BB_{p,r}^{\frac{N}{p}-1}}+
	\|u^{n-1}\|_{L^{2r}_t L^\frac{Nr}{r-1}_x}
	\|\nabla^2 d^{n-1}\|_{L^{2r}_t L^\frac{Nr}{2r-1}_x}+
	\|\nabla u^{n-1}\|_{L^{2r}_t L^\frac{Nr}{2r-1}_x}
	\|\nabla d^{n-1}\|_{L^{2r}_t L^\frac{Nr}{r-1}_x}+\\&\quad\quad\quad\quad
	+
	\|\nabla^2 d^{n-1}\|_{L^{2r}_t L^\frac{Nr}{2r-1}_x} 
	\|\nabla d^{n}\|_{L^{2r}_t L^\frac{Nr}{r-1}_x}\|d^n\|_{L^\infty_{t,x}}+
	\|\nabla d^{n-1}\|_{L^{3r}_tL^{\frac{3Nr}{3r-2}}_x}^2
	\|\nabla d^{n}\|_{L^{3r}_tL^{\frac{3Nr}{3r-2}}_x}.
\end{split}
\end{equation}
Here, we have used Lemma \ref{Lemma3} with $r_1 =r $, $r_2 = 2r$ (respectively $r_2=2$), $q_1 = Nr/(3r-2)$ and $q_2= Nr/(r-1)$ (respectively $q_2=\infty$), which fulfil the conditions \eqref{condition_lemma3a} and \eqref{condition_lemma3}  (since $r<2$). Moreover the constants of Lemma \ref{Lemma2} are determined by $r_1=r$, 
$q_1 = Nr/(3r-2)$, $r_2=2r$ (respectively $r_2=r$) and $q_2=Nr/(2r-1)$ (respectively $q_2=Nr/(2r-2)$), which satisfy the conditions \eqref{condition_lemma2a} and \eqref{condition_lemma2}.

\noindent The hypotheses \eqref{ineq_proof_prop_smooth_date1} for $n-1$, allows us to absorb all the terms on the right-hand side with index $n$ by the left-hand side (for $\eta$ small enough). Hence,  
\eqref{ineq_proof_prop_smooth_date1} is true at least for the terms related to $d$.
Now, let us estimate the remaining terms. By the Mild formulation for the Stokes equation, we get
\begin{equation}\label{prop_smooth_dates_explicit_formula_u}
	u^n(t) =e^{t\Delta}u_0 +
	\int_0^t e^{(t-s)\Delta}\PP\big\{-u^{n-1} \cdot \nabla u^n(s) + F^n(s) \big\}\dd s
\end{equation}
where $\PP$ is the well known Leray projector.  Moreover, applying $\Div$ to the second equation of 
\eqref{prop_smooth_dates_system2}, we get $-\Delta \Pi^n =\{-u^{n-1} \cdot \nabla u^n + F^n\}$, which yields
\begin{equation}\label{prop1_Pin_explicit_formula}
	\nabla \Pi^n = RR\cdot \big\{ -u^{n-1} \cdot \nabla u^n + F^n\big\},
\end{equation}
where $R$ is the Riesz operator. Since $\PP$ and $R$ are bounded operators from $L^p$ to $L^p$, for every 
$1<p<\infty$, applying  Corollary \ref{Cor_Characterization_of_hom_Besov_spaces}, Lemma \ref{Lemma2}, Lemma \ref{Lemma3} (with the constants $r_1$, $r_2$, $q_1$ and $q_2$ as in \eqref{prop_smooth1_estimate_dn}) and Theorem \ref{Maximal_regularity_theorem}, we deduce that
\begin{equation}\label{prop_smooth1_estimate_un}
\begin{aligned}
		\| \nabla u^n &\|_{L^{2r}_t L^{ \frac{Nr}{2r-1}}_x}+ 
		\|\nabla u^n\|_{L^{r}_t L^{\frac{Nr}{2(r-1)}}_x} +
		\| u^n\|_{L^{2r}_t L^{\frac{Nr}{r-1}}_x}+ 
		\| u^n\|_{L^{2}_t L^{\infty}_x}+ 
		\|(\nabla^2 u^n,\, \nabla \Pi^n)\|_{L^{r}_t L^{\frac{Nr}{3r-2}}_x}\lesssim \\
		&\|u_0\|_{\BB^{\frac{N}{p}-1}_{p,r}}+
		\|u^{n-1}\|_{L^{2r}_t L^{\frac{Nr}{r-1}}_x}
		\|\nabla u^n\|_{L^{2r}_t L^{\frac{Nr}{2r-1}}_x}+
		\|F^n \|_{L^r_t L^{\frac{Nr}{3r-2}}_x}
		\lesssim
		 \eta \|\nabla u^n\|_{L^{2r}_t L^{\frac{Nr}{2r-1}}_x}+\eta.
\end{aligned}
\end{equation}
As in the previous estimates, the term of index $n$ in the right-hand side can be absorbed by the left-hand side, obtaining finally \eqref{ineq_proof_prop_smooth_date1}.

\noindent Now, let us observe that $\|\nabla d^n\|_{L^{3r}_t L^{\frac{3Nr}{(3-\ee)r-2}}_x}\lesssim \|\nabla d_0\|_{\BB_{p,r}^{\frac{N}{p}-1+\frac{2}{3}\ee}}$, for every $n\in\NN$. Indeed, by induction and recalling \eqref{prop_smooth_dates_explicit_formula_d}, we get
\begin{equation*}
	\|\nabla d^n\|_{L^{3r}_t L^{\frac{3Nr}{3r-2}}_x}\lesssim 
	\|\nabla d_0\|_{\BB_{p,r}^{\frac{N}{p}-1+\frac{2}{3}\ee}} + 
	\|(u^{n-1},\nabla d^{n-1})\|_{L^{2r}_t L^{\frac{Nr}{r-1}}_x} 
	\|\nabla d^{n-1}\|_{L^{3r}_t L^{\frac{3Nr}{3r-2}}_x}\lesssim
	\|\nabla d_0\|_{\BB_{p,r}^{\frac{N}{p}-1+\frac{2}{3}\ee}},
\end{equation*}
where we have used Corollary \ref{Cor_Characterization_of_hom_Besov_spaces} and Lemma \ref{Lemma2}, with $r_1=6r/5$, $r_2=3r$, $q_1=3Nr/((6-\ee)r-5)$ and $q_2=3Nr/((3-\ee)r-2)$ which verify \eqref{condition_lemma2a} and \eqref{condition_lemma2}. Hence, in the same line of \eqref{prop_smooth1_estimate_dn} we get
\begin{equation*}
\begin{split}
	\|&\nabla^2 d^n\|_{L^{\frac{2}{2-\ee}}_t L^\infty_x} + 
	\|\nabla^2 d^n \|_{L^{2r}_t L^{ \frac{Nr}{(2-\ee)r-1}}_x}+ 
	\|\nabla^3 d^n\|_{L^r_t L^{\frac{Nr}{(3-\ee)r-2}}_x}
	\lesssim
	\|\nabla d_0\|_{\BB_{p,r}^{\frac{N}{p}-1+\ee}}+\\&+
	\|u^{n-1}\|_{L^{2r}_t L^\frac{Nr}{r-1}_x}
	\|\nabla^2 d^{n-1}\|_{L^{2r}_t L^\frac{Nr}{(2-\ee)r-1}_x}+
	\|\nabla u^{n-1}\|_{L^{2r}_t L^\frac{Nr}{(2-\ee)r-1}_x}
	\|\nabla d^{n-1}\|_{L^{2r}_t L^\frac{Nr}{r-1}_x}+\\&\quad
	+
	\|\nabla^2 d^{n-1}\|_{L^{2r}_t L^\frac{Nr}{(2-\ee)r-1}_x} 
	\|\nabla d^{n}\|_{L^{2r}_t L^\frac{Nr}{r-1}_x}\|d^n\|_{L^\infty_{t,x}}+
	\|\nabla d^{n-1}\|_{L^{3r}_tL^{\frac{3Nr}{(3-\ee)r-2}}_x}^2
	\|\nabla d^{n}\|_{L^{3r}_tL^{\frac{3Nr}{(3-\ee)r-2}}_x}.
\end{split}
\end{equation*}
Here, we have Theorem \ref{Maximal_regularity_theorem} and used Lemma \ref{Lemma2} with $r_1=r$, 
$q_1 = Nr/((3-\ee)r-2)$, $r_2=2r$ (respectively $r_2=2/(2-\ee)$) and $q_2=Nr/((2-\ee)r-1)$ (respectively $q_2=\infty$), which satisfy the conditions \eqref{condition_lemma2a} and \eqref{condition_lemma2}. Furthermore, in the same line of \eqref{prop1_estimates_delta_un}, we get 
\begin{equation*}
\begin{split}
		\| &\nabla u^n \|_{L^{2r}_t L^{ \frac{Nr}{(2-\ee)r-1}}_x}+ 
		\|\nabla u^n\|_{L^{\frac{2}{2-\ee}}_t L^{\infty}_x}+ 
		\|(\nabla^2 u^n,\, \nabla \Pi^n)\|_{L^{r}_t L^{\frac{Nr}{(3-\ee)r-2}}_x}\lesssim 
		\|u_0\|_{\BB^{\frac{N}{p}-1+\ee}_{p,r}}+
		\|u^{n-1}\|_{L^{2r}_t L^{\frac{Nr}{r-1}}_x}\\ & {\scriptstyle \times}
		\|\nabla u^n\|_{L^{2r}_t L^{\frac{Nr}{(2-\ee)r-1}}_x}+		+
		\eta\|\nabla( u^n, \nabla d^{n-1})\|_{L^{2r}_t L^{\frac{Nr}{(2-\ee)r-1}}_x}+
		\eta\|(\nabla^2 u^{n-1},\,\nabla \Pi^{n-1})\|_{L^r_t L^{\frac{Nr}{(3-\ee)r-2}}_x}.
\end{split}
\end{equation*}
Summarizing the previous consideration, we get by induction
\begin{equation}\label{ineq_proof_prop_smooth_date4}
	\|(u^n,\,\nabla d^n,\,\nabla \Pi^n)\|_{\X_r^\ee}
	\lesssim 	
	\|(u_0,\,\nabla d_0)\|_{\BB_{p,r}^{\frac{N}{p}-1+\ee}}+
	\|\nabla d_0\|_{\BB_{p,r}^{\frac{N}{p}-1+\frac{2}{3}\ee}}^3,\quad\quad n\in \NN.
\end{equation}	

\noindent At last, arguing as in the proof of \eqref{ineq_proof_prop_smooth_date4}, we get also the following inequality
\begin{equation}\label{ineq_proof_prop_smooth_date5}
\begin{split}
	\|\nabla d^n\|_{L^{\frac{12r}{4-3\ee }}_t L^\frac{3Nr}{3r-2}_x}&+
	\|\nabla(  u^n,\, \nabla d^n) \|_{L^{\frac{4r}{2-\ee}}_t L^{ \frac{Nr}{2r-1}}_x}+
	\|(\nabla^2 u^n,\, \nabla \Pi^n)\|_{L^{\frac{4r}{4-\ee}}_t L^{\frac{Nr}{3r-2}}_x}\\& 
	\lesssim \|(u_0,\,\nabla d_0)\|_{\BB_{p,r}^{\frac{N}{p}-1+\frac{\ee}{2r}}}\lesssim 
	\eta + \|(u_0,\nabla d_0)\|_{\BB_{p,r}^{\frac{N}{p}-1+\ee}}.
\end{split}
\end{equation}
Here, we need Lemma \ref{Lemma3} with $r_1 =4r/(4-\ee) $, $r_2 = 4r/(2-\ee)$, $q_1 = Nr/(3r-2)$ and $q_2= Nr/(r-1)$, which fulfil the conditions \eqref{condition_lemma3a} and \eqref{condition_lemma3}. We need also Lemma \ref{Lemma2} with $r_1=4r/(4-\ee)$, 
$q_1 = Nr/(3r-2)$, $r_2=4r/(2-\ee)$ and $q_2=Nr/(2r-1)$, which satisfy the conditions \eqref{condition_lemma2a} and \eqref{condition_lemma2}.

\noindent Now we claim that, for every $T>0$, $d^n$ is a Cauchy sequence in $L^\infty_TL^\infty_x$,  $(u^n,\, \nabla d^n, \nabla \Pi_n)_\NN$ is a Cauchy sequence in $\X_{r,T}$ and   $\nabla( u^n,\, \nabla d^n)_\NN$ is a Cauchy sequence in  $L^{2}_T L^{\infty}_x$.

\noindent Denoting $\delta u^n:= u^{n+1}-u^{n}$, $\delta a^{n}:= a^{n+1}-a^{n}$, $\delta d^n:= d^{n+1}-d^{n}$, 
$\delta \Pi^{n}:= \Pi^{n+1}-\Pi^{n}$, we define
\begin{equation*}
	\delta U^n(T):= \|\delta d^n\|_{L^\infty_TL^\infty_x}+
	\|(\delta u^n,\,\nabla \delta d^n\,\nabla \delta \Pi^n)\|_{\X_{r,T}}+ 
	\| (\delta u^n,\, \nabla \delta d^n) \|_{L^{2}_T L^{\infty}_x}
\end{equation*}
We want to prove that $\sum_{n\in\NN}\delta U^n(T)$ is finite.
First, let us consider $\delta a^n$, which is solution of the following system
\begin{equation}\label{prop1_eq_delta_an}
\begin{cases}
	\;\partial_t\delta a^n +  u^n\cdot\nabla \delta a^n= -\delta u^{n-1} \nabla  a^n
																& \RR_+ \times\RR^N,\\
	\;\delta a^n_{|t=0} = 0										& \;\;\quad \quad\RR^N.\\
\end{cases}
\end{equation}
Using standard estimates for the transport equation, we obtain that
\begin{equation}\label{ineq_proof_prop_smooth_date3}
	\| \delta a^n(t) \|_{L^\infty_x} \leq 
	\int_0^t \| \delta u_{n-1}(s)\|_{L^\infty_x}\| \nabla  a^n(s)\|_{L^\infty_x}
	\leq \bar{C}(T)\delta U^{n-1}(t)\|\nabla a_0\|_{L^\infty_x},
\end{equation}
for every $t\in (0,T)$, where
\begin{equation*} 
	\bar{C}(t):=
	T^{\frac{1}{2}}\exp\big\{ T^{\frac{\ee}{2}} \|(u_0,\nabla d_0)\|_{\BB^{\frac{N}{p}-1+\ee}_{p,r}} \big\}.
\end{equation*}
Considering $\delta d^n$ we observe that it is solution of
\begin{equation}\label{prop1_eq_delta_dn}
\begin{cases}
	\; \partial_t \delta d^n  -\Delta \delta  d^n= 
	\delta H^n
	& \RR_+ \times\RR^N,\\
	\;\delta d^n_{|t=0} = 0									& \;\;\quad \quad\RR^N,\\
\end{cases}
\end{equation}
where
\begin{equation*}
\begin{split}
	\delta H^n:= - \delta u^{n-1}\cdot \nabla d^{n} 
	&- u^{n-1} \cdot \nabla \delta  d^{n-1}+
	\nabla\delta d^{n-1} \odot \nabla d^n d^n+\\
	&+\nabla d^{n-1} \odot  \nabla \delta  d^{n-1} d^n+
	\nabla d^{n-1} \odot \nabla d^{n-1} \delta d^{n}.
\end{split}
\end{equation*}
Thus, for every $t\in (0,T)$,
\begin{equation*}
	\|\delta d^n(t)\|_{L^\infty_x}\leq \| \delta H^n \|_{L^1_TL^\infty_x}\lesssim 
	\eta \|(\delta u^{n-1},  \nabla \delta  d^{n-1})\|_{L^2_TL^\infty_x} + 
	\eta^2 \|\delta d^n(t)\|_{L^\infty_TL^\infty_x},
\end{equation*}
which yields
\begin{equation}
	\|\delta d^n\|_{L^\infty_TL^\infty_x}\lesssim \eta \delta U^{n-1}(T).
\end{equation}

\noindent
Arguing exactly as in the proof of inequality \eqref{ineq_proof_prop_smooth_date1}, we obtain
\begin{equation*}
\begin{split}
	\| \nabla & \delta  d^n\|_{L^{2}_T L^{\infty}_x}+
	\| \nabla \delta  d^n\|_{L^{3r}_T L^{\frac{3Nr}{3r-2}}_x}+
	\| \nabla \delta  d^n\|_{L^{2r}_T L^{\frac{Nr}{r-1}}_x}+
	\|\nabla^2 \delta d^n \|_{L^{2r}_TL^{ \frac{Nr}{2r-1}}_x}+
	\|\nabla^2 \delta d^n\|_{L^{r}_TL^{\frac{Nr}{2(r-1)}}_x}+\\ 
	&+
	\|\nabla^3 \delta d^n\|_{L^{r}_TL^{\frac{Nr}{3r-2}}_x}
	\lesssim
	\|\nabla \delta H^n\|_{L^r_T L^{\frac{Nr}{3r-2}}_x}+
	\|\delta H^n\|_{L^{\frac{6}{5}r}_T L^{\frac{3Nr}{6r-2}}_x}\lesssim
	\eta \big(\delta U^n(T)+\delta U^{n-1}(T)\big).
\end{split}
\end{equation*}
In order to conclude our estimate we have to bound the terms related to $\delta u^n$, which is solution of
\begin{equation}\label{prop1_estimates_delta_un}
\begin{cases}
	\;\partial_t \delta u^n + u^{n}\cdot \nabla\delta u^{n} + \delta u^{n-1}\cdot \nabla u^n
	-\Delta \delta u^n + \nabla \delta \Pi^n =	\delta F^{n}					& \RR_+ \times\RR^N,\\
	\;\Div\, \delta u^n = 0														& \RR_+ \times\RR^N,\\
	\;\delta u^n_{|t=0} = 0														& \;\;\quad \quad\RR^N,\\
\end{cases}
\end{equation}
where $\delta F^n:= F^{n+1}-F^n$. First, let us observe that
\begin{equation*}
\begin{split}
	\| \nabla &\delta u^n\|_{L^{2r}_TL^{ \frac{Nr}{2r-1}}_x}+ 
	\|\nabla \delta u^n\|_{L^{r}_TL^{\frac{Nr}{2(r-1)}}_x} +
	\| \delta u^n \|_{L^{2r}_TL^{\frac{Nr}{r-1}}_x}+
	\| \delta u^n \|_{L^{2}_TL^{\infty}_x}+\\&+ 
	\|(\nabla^2 \delta u^n,\, \nabla \delta \Pi^n)\|_{L^{r}_TL^{\frac{Nr}{3r-2}}_x} 
	\lesssim
	\eta 
	\big(\| \delta u^{n-1} \|_{L^{2r}_TL^{\frac{Nr}{r-1}}_x}+
	\|\nabla  \delta u^{n} \|_{L^{2r}_TL^{\frac{Nr}{2r-1}}_x}
	\big)
	+
	\| \delta F^n \|_{L^{r}_TL^{\frac{Nr}{3r-2}}_x}
\end{split}	
\end{equation*}
Hence, denoting by $G^n:=\Div\big(\nabla d^n \odot \nabla d^n\big)$,
we have that
\begin{equation*}
	\begin{split}
		\| \delta F^n \|_{L^{r}_TL^{\frac{Nr}{3r-2}}_x}&\lesssim 
		\|\delta a^n \big(G^n,\,\Delta u^n,\, \nabla \Pi^n\big)\|_{L^{r}_TL^{\frac{Nr}{3r-2}}_x} 
		+\|\delta G^n\|_{L^{r}_TL^{\frac{Nr}{3r-2}}_x}+ 
		\eta\|(\Delta \delta u^{n-1},\, \nabla \delta \Pi^{n-1}) \|_{L^{r}_TL^{\frac{Nr}{3r-2}}_x},		
	\end{split}
\end{equation*}
which yields, recalling \eqref{ineq_proof_prop_smooth_date5}, 
\begin{equation*}
\begin{split}
	\|\delta F^n\|_{L^{r}_TL^{\frac{Nr}{3r-2}}_x} &\lesssim
	\|(G^n,\,\Delta u^n, \nabla \Pi^n) \|_{L^{\frac{4r}{4-\ee }}_TL^{\frac{Nr}{3r-2}}_x}
	\|\delta a^n\|_{L^{\frac{4r}{\ee}}_T L^\infty_x} + \delta U^n(T)\eta\\
	&\lesssim
	\bar{C}(T)\|\nabla a_0\|_{L^\infty_x}\|\delta U^{n-1}\|_{L^{\frac{4r}{\ee}}_T} +  \delta U^n(T)\eta.
\end{split}
\end{equation*}

\noindent
Summarizing the previous considerations and supposing $\eta$ small enough, we obtain
\begin{equation}\label{estimate_deltaUn}
	\delta U^n(T)\lesssim \eta \delta U^{n-1}(T) + 
	\bar{C}(T)\|\nabla a_0\|_{L^\infty_x}\|\delta U^{n-1}\|_{L^{\frac{e}{2}}(0,T)}.
\end{equation}
We claim that there exists $C(T)>0$ and $K(T)>0$ such that, for all $t\in [0,T]$, and 
for all $n\in\NN$
\begin{equation}\label{estimate_deltaUn2}
	\delta U^{n}(t)\leq C\eta^{\frac{n}{2}}\exp\big\{ K(T)\frac{t}{\sqrt{\eta}}\big\}.
\end{equation}
We are going to prove it by induction and the base case is trivial, since it is sufficient to find $C(T)>0$ such that, for all $t\in [0,T]$,
\begin{equation*}
	\delta U^{0}(t)\leq C(T).
\end{equation*}
Then, for all $K(T)>0$, it is fulfilled
	\begin{equation*}
		\delta U^{0}(t)\leq C(T)\exp\big\{ K(T)\frac{t}{\sqrt{\eta}}\big\}.
	\end{equation*}
	Passing trough the induction hypotheses, we have that, for all $t\in [0,T]$
	\begin{align*}
		\delta U^{n}(t)&\lesssim
		\eta\delta U_{n-1}(t)+ \bar{C}(T)\|\nabla a_0\|_{L^\infty_x}
		\|\delta U_{n-1}\|_{L^{\frac{4r}{\ee}}(0,t)}\\
		&\lesssim \sqrt{\eta}C(T) \eta^{\frac{n}{2}}+
					C(T)\bar{C}(T)\|\nabla a_0\|_{L^\infty_x}
					\Big(
						\int_0^t \exp\left\{\frac{4r}{\ee} K(T)\frac{s}{\sqrt{\eta}} \right\}\dd s
					\Big)^\frac{\ee}{2}\\
			&\leq 	\Big( 
						\sqrt{\eta}C(T) + \frac{\ee^{\frac{\ee}{4r}}}{(2K(T))^{\frac{\ee}{4r}}}C(T)\bar{C}(T)
						\|\nabla a_0\|_{L^\infty_x}
					\Big)\eta^{\frac{n}{2}}\exp\left\{ K(T)\frac{t}{\sqrt{\eta}}\right\}. 
	\end{align*}
	Chosen $K(T)>0$ big enough and supposing $\eta$ small enough, we finally obtain \eqref{estimate_deltaUn}.
	
	\noindent	
	It is now immediate to conclude that $(d^n)_\NN$, $(\nabla(u^n,\,\nabla d^n)\,)_\NN$ and $(u_n,\,
	\nabla d^n,\,\nabla \Pi^n)_\NN $ are Cauchy sequences in $L^\infty_TL^\infty_x$, $L^2_TL^\infty_x$ 
	and $\X_{r,T}$, respectively . Moreover, resuming \eqref{ineq_proof_prop_smooth_date3}, we deduce 
	that $(a^n)_\NN$ is a Cauchy sequence in $L^\infty_TL^\infty_x$.  Granted with these convergence 
	results and recalling the inequalities \eqref{ineq_proof_prop_smooth_date2}, 
	\eqref{ineq_proof_prop_smooth_date1} and \ref{ineq_proof_prop_smooth_date4} we conclude that the 
	limit $(a, \,u,\, \nabla d,\, \nabla \Pi)_\NN$ fulfils the property of the Proposition. 
	
	\noindent Finally, recalling that, for every natural $n$, $(a^n,\, u^n, \, \nabla d^n, \nabla \Pi^n)$ is 
	solution of \eqref{prop_smooth_dates_system1} and \eqref{prop_smooth_dates_system2}, passing 
	through the limit, we deduce that $(a, \,u,\, \nabla d,\, \nabla \Pi)_\NN$ is solution of 
	\eqref{Navier_Stokes_system} with $(a_0,\,u_0,\,d_0)$ as initial data, and this completes the proof 
	of Proposition \ref{Theorem_solutions_smooth_dates}.
\end{proof}

\begin{prop}\label{Theorem_solutions_smooth_dates_2} 
	Let $r\in(1,\infty)$, $p\in (1,N)$. Suppose that the initial data fulfil \eqref{initial_data} and moreover 
	$\nabla a_0\in L^\infty_x$ and $(u_0,\,\nabla d_0)\in\BB_{p,r}^{N/p-1+\ee}$ with 
	$\ee<\min\{1/r,1-1/r, N/p-1\}$.	
	If the smallness condition \eqref{smallness_condition} holds, then 
	there exists a global weak-solution $(a,\, u,\, d)$  with the same property of Theorem
	\ref{Main_Theorem_2}. Moreover $\nabla a\in L^\infty_{t,loc}L^\infty_x$, $(u,\, \nabla d, \nabla \Pi)$ 
	belongs to $\Y_r^\ee$ and
	\begin{equation}\label{condition_theorem_smooth_data_2}
		\| (u,\, \nabla d, \,\nabla \Pi)\|_{\Y_r^\ee} 
		\lesssim 
		\|(u_0,\,\nabla d_0)\|_{\BB_{p,1}^{\frac{N}{p}-1+\ee}}.
	\end{equation}
\end{prop}
\begin{remark}\label{remark1}
	The condition \eqref{condition_theorem_smooth_data_2} ensures the velocity field to be in 
	$L^1_{t,loc}\mathcal{L}ip_x$. Indeed, a classical Gagliardo-Niremberg interpolation inequality
	\begin{equation*}
		\|\nabla f\|_{L^\infty_x}\lesssim 
		\|f\|_{L^{q_3}_x}^{1-\theta}\|\nabla^2 f\|_{L^{q_1}_x}^\theta\quad\quad
		\theta =\frac{Nq_1+q_1q_3}{Nq_1+2q_1q_3-Nq_3}\in \left[\frac{1}{2},\,1\right],
	\end{equation*}
	allows us to obtain the following estimate for every positive $T$:
	\begin{equation*}
	 \|\nabla (u,\,\nabla d)\|_{L^1_T L^\infty_x}\lesssim 
	 \|(t^{-\gamma_1^\ee},\,t^{-\alpha_1^\ee})\|_{L^{(2r)'}_T}
	\Big(\, 
	\|t^{\gamma_1^\ee} (u,\,\nabla d)\|_{L^{2r}_TL^{q_3}_x} + 
	\|t^{\alpha_1^\ee} \nabla^2 (u,\nabla d)\|_{L^{2r}_TL^{q_1}_x}
	\Big)<\infty.
	\end{equation*}
	As already mentioned, such condition permits the existence of the flow for the velocity field, hence 
	we can reformulate system \eqref{Navier_Stokes_system_1} trough Lagrangian coordinates (see section $6$).
 	Adding a weigh in time, we can increase the time integrability by
 	\begin{equation*}
 	\|t^{\alpha_1^{\ee}}(\nabla u,\,\nabla^2 d)\|_{L^{2r}_T L^\infty_x} \lesssim 
 	T^{\alpha_1^\ee-\gamma_1^\ee}\|t^{\gamma_1^\ee} u\|_{L^{2r}_TL^{q_3}_x} + 
	\|t^{\alpha_1^\ee} \nabla^2 u\|_{L^{2r}_TL^{q_1}_x}<\infty,
	\end{equation*}
	observing that $\alpha_1^\ee-\gamma_1^\ee$ is positive.	These estimates are going to be useful to 
	prove the uniqueness for the solution of \eqref{Navier_Stokes_system_1} in the $\Y_r\cap \Y_r^\ee$ 
	functional framework. 	 
\end{remark}
\begin{proof}[Proof of Proposition \ref{Theorem_solutions_smooth_dates_2}]
	Proceeding with the same strategy of Proposition \ref{Theorem_solutions_smooth_dates}, we consider the 
	sequence of solutions for the systems \eqref{prop_smooth_dates_system1} and 
	\eqref{prop_smooth_dates_system2}. 
	We claim by induction that such solutions belong to the same space of Theorem \ref{Main_Theorem_2} and 
	moreover that 
	\begin{equation}\label{prop2_inductionclaimestimate}
	\begin{aligned}
		\|(\delta u^n,\,\nabla d^n,\,\nabla \delta \Pi^n)\|_{\Y_r}		
		\lesssim\eta.
	\end{aligned}
	\end{equation}
	
	\noindent
	At first, let us observe that $\| e^{t\Delta}d_0 \|_{L^\infty_x}\leq \|d_0\|_{L^\infty_x}\leq 1$.
	Recalling that
	\begin{equation}\label{prop2_dn}
		d^n(t)= e^{t\Delta}d_0
		+
		\int_0^t e^{(t-s)\Delta}\big\{ |\nabla d^{n-1}|^2 d^n - u^{n-1}\cdot \nabla d^{n-1} \big\}(s)\dd s,
	\end{equation}
	by Lemma \ref{Lemma6}, with $\bar{r}=r$, $q=p_3/2$ and $\sigma=2\gamma_1=1-N/p_3-1/r$, we have
	\begin{equation*}
		\|d^n(t)\|_{L^\infty_{t,x}} \leq 1 + C_r\|s^{2\gamma_1} u^{n-1}\cdot 
		\nabla d^{n-1}(s)\|_{L^{r}_s L^{\frac{p_3}{2}}_x} + 
		C_r\Big(
		\int_0^ts^{2r\gamma_1}
			\|\nabla d^{n-1}(s)\|_{L^{p_3}_x}^{2r}\|d^n(s)\|_{L^\infty_x}
		\dd s
		\Big)^{\frac{1}{2r}},
	\end{equation*}
	for every $t\in \RR_+$, where $C_r$ is a positive constant. Thus, by the induction hypotheses and 
	Gronwall inequality, we deduce
	\begin{equation}\label{estimate_dn_prop2}
		\|d^n\|_{L^\infty_{x}}^{2r}\leq 2r(1 + C_{r}\eta^2)^{2r}\exp\Big\{C_{r}\eta^{2r}\Big\}.
	\end{equation}
	Then we get $\|d^n(t)\|_{L^\infty_{t,x}}\leq \bar{C_r}$, with $\bar{C_r}$ positive constant 
	dependent only by $r$.
	
	\noindent
	Furthermore, using standard estimates for the transport equation, we have 
	$\|a^n\|_{L^\infty_{t,x}}\leq \|a_0\|_{L^\infty_{x}}$ and for all positive $t$,
	\begin{equation}\label{prop2_estimate_nabla_an}
		\|\nabla a^n(t)\|_{L^\infty_{t,x}}\leq \|\nabla a_0\|_{L^\infty_x}
		\exp\Big\{\int_0^t \|\nabla u^{n-1}(s)\|_{L^\infty_x}\dd s\Big\}.
	\end{equation}
	
	\noindent
	We claim now that $t^{\gamma_1}\nabla d^n$ belongs to $L^{2r}_t L^{p_3}_x$, $t^{\gamma_2}\nabla d^n$ 
	to $L^{\infty}_t L^{p_3}_x$, $t^{\gamma_3}\nabla d^n$ to $L^{2r}_t L^{3p_1}_x$ and 	$t^{\gamma_4}\nabla 
	d^n$ to $L^{\infty}_t L^{3p_1}_x$. By Theorem \ref{Characterization_of_hom_Besov_spaces}, 
	\eqref{smallness_condition} and	since 
	\begin{equation}\label{prop2_emb_B}
	\BB_{p,r}^{N/p-1}\hookrightarrow 
	\BB_{p_3,2r}^{N/{p_3}-1}\cap \BB_{3p_1,2r}^{N/{3p_1}-1}\cap
	\BB_{p_3,\infty}^{N/{p_3}-1}\cap \BB_{3p_1,\infty}^{N/{3p_1}-1},
	\end{equation} 
	we obtain that 
	\begin{equation*}
		\|t^{\gamma_1}  \nabla e^{t\Delta} d_0 \|_{L^{2r}_t L^{p_3}_x}+
		\|t^{\gamma_2}  \nabla e^{t\Delta} d_0 \|_{L^{\infty}_t L^{p_3}_x}+
		\|t^{\gamma_3}  \nabla e^{t\Delta} d_0 \|_{L^{2r}_t L^{3p_1}_x}+
		\|t^{\gamma_4}  \nabla e^{t\Delta} d_0 \|_{L^{\infty}_t L^{3p_1}_x}\lesssim
		\eta.
	\end{equation*}
	Furthermore, from the induction hypotheses and \eqref{estimate_dn_prop2}, we have 
	$t^{\alpha_1}\nabla ( |\nabla d^{n-1}|^2 d^n - u^{n-1}\cdot \nabla d^{n-1})\in L^{2r}_tL^{p_1}_x$. 
	Thus, applying Lemma  \ref{Lemma4b} with $q=p_1$, $\tilde{q}=p_3$ and moreover 
	the same Lemma with $\tilde{q}=3p_1$, we finally obtain
	\begin{equation}\label{recal_1}
	\begin{split}
		\|t^{\gamma_1}  \nabla d^n \|_{L^{2r}_t L^{p_3}_x}+
		&\|t^{\gamma_2}  \nabla d^n \|_{L^{\infty}_t L^{p_3}_x}+
		\|t^{\gamma_3}  \nabla d^n \|_{L^{2r}_t L^{3p_1}_x}+
		\|t^{\gamma_4}  \nabla d^n \|_{L^{\infty}_t L^{3p_1}_x}\lesssim\\
		&\lesssim 
		\|\nabla d_0\|_{\BB_{p,r}^{\frac{N}{p}-1}} + 
		\|t^{\alpha_1}\nabla ( |\nabla d^{n-1}|^2 d^n - u^{n-1}\cdot \nabla d^{n-1})
		\|_{L^{2r}_t L^{p_1}_x}.
	\end{split}	
	\end{equation}
	Developing the right-hand side and absorbing the terms with index $n$ by the left-hand side, 
	it results
	\begin{equation}\label{prop2_estimate_nabladn}
		\|t^{\gamma_1}  \nabla d^n \|_{L^{2r}_t L^{p_3}_x}+
		\|t^{\gamma_2}  \nabla d^n \|_{L^{\infty}_t L^{p_3}_x}+
		\|t^{\gamma_3}  \nabla d^n \|_{L^{2r}_t L^{3p_1}_x}+
		\|t^{\gamma_4}  \nabla d^n \|_{L^{\infty}_t L^{3p_1}_x}\lesssim \eta.
	\end{equation}
	
	\noindent
	In a similar way, by Theorem \ref{Maximal_regularity_Thm_weight_time} and by Lemma \ref{Lemma5b} with  
	$\bar{q}=p_2$ and $q=p_1$, we obtain
	\begin{equation}\label{prop2_estimate_nabla2dn}
		\|t^{\beta_1}  \nabla^2 d^n \|_{L^{2r}_t L^{p_2}_x}+ 
		\|t^{\beta_2}  \nabla^2 d^n\|_{L^{\infty}_t L^{p_2}_x}+
		\|t^{\alpha_1}\nabla^3 d\|_{L^{2r}_t L^{p_1}_x}+
		\|t^{\alpha_2}\nabla^3 d\|_{L^{r}_t L^{p_1}_x}\lesssim \eta.
	\end{equation}

	\noindent
	Now, let us take the velocity field into account. At first, we recall that $u^n$ fulfils 
	\begin{equation}\label{prop2_un}
		u^n(t) = e^{t\Delta}u_0 + \int_0^t e^{(t-s)\Delta}
		\PP\big\{-u^{n-1} \cdot \nabla u^n(s) + F^n(s) \big\}\dd s\dd s,
	\end{equation}
	with $F^n= (1+a^n)\Div\big(\nabla d^n \odot \nabla d^n \big)+a^n(\Delta u^{n-1} - \nabla \Pi^{n-1})$.
	By Theorem \ref{Characterization_of_hom_Besov_spaces},	$t^{\gamma_1} e^{t\Delta}u_0$ belongs to 
	$L^{2r}_tL^{p_3}_x$ and $t^{\gamma_2} e^{t\Delta}u_0\in L^{\infty}_tL^{p_3}_x$. 
	By Theorem 
	\ref{Maximal_regularity_Thm_weight_time} with $\alpha = \alpha_1$ (respectively  $\alpha=\alpha_2$), 
	$\bar{r}=2r$ (respectively $\bar{r}=r$) and $q=p_1$ we obtain
	\begin{equation}\label{prop2_estimate_nabla2un}
	\begin{split}
		\|t^{\alpha_1} \nabla^2 u^n \|_{L^{2r}_t L^{p_1}_x}&+
		\|t^{\alpha_2} \nabla^2 u^n \|_{L^{r}_t L^{p_1}_x}\lesssim 
		\|t^{\alpha_1} F^n\|_{L^{2r}_t L^{p_1}_x}+\\
		&+
		\|t^{\alpha_1} u^{n-1}\cdot\nabla u^n\|_{L^{2r}_t L^{p_1}_x}+
		\|t^{\alpha_2} F^n\|_{L^{r}_t L^{p_1}_x}+
		\|t^{\alpha_2} u^{n-1}\cdot\nabla u^n\|_{L^{r}_t L^{p_1}_x}
	\end{split}
	\end{equation}
	Arguing exactly as in the proof of \eqref{prop2_estimate_nabladn} and \eqref{prop2_estimate_nabla2dn}
	we get also
	\begin{equation}\label{prop2_estimate_un_nablaun}
	\begin{split}	
		\|t^{1-\frac{1}{2r}}u^n\|_{L^{2r}_t L^{\infty}_x}&+
		\|t^{\gamma_1}u^n\|_{L^{2r}_t L^{p_3}_x} + 
		\|t^{\gamma_2}u^n\|_{L^{\infty}_t L^{p_3}_x} +\\&+
		\|t^{\beta_1}\nabla u^n\|_{L^{2r}_t L^{p_2}_x}+
		\|t^{\beta_2}\nabla u^n\|_{L^{\infty}_t L^{p_2}_x}\leq 
		\|t^{\alpha_1} F^n\|_{L^{2r}_t L^{p_1}_x}.
	\end{split}
	\end{equation}
	Since 
	\begin{equation*}
	\begin{split}	
		\|t^{\alpha_1} F^n\|_{L^{2r}_t L^{p_1}_x}&+
		\|t^{\alpha_2} F^n\|_{L^{r}_t L^{p_1}_x}\lesssim 
		(1+\|a_0\|_{L^\infty_x})
		\|t^{\beta_1}\nabla^2d^n\|_{L^{2r}L^{p_2}_x}\|t^{\beta_3}\nabla^2d^n\|_{L^{\infty}L^{p_3}_x}+ \\
		&+\|a_0\|_{L^\infty_x}
		\Big(\|t^{\alpha_1}(\Delta u^{n-1}, \nabla	\Pi^{n-1} ) \|_{L^{2r}_tL^{p_1}_x} +
		\|t^{\alpha_2}(\Delta u^{n-1}, \nabla	\Pi^{n-1} ) \|_{L^{r}_tL^{p_1}_x}\Big),
	\end{split}
	\end{equation*}
	summarizing \eqref{prop2_estimate_nabladn}, \eqref{prop2_estimate_nabla2dn}, 
	\eqref{prop2_estimate_nabla2un} and \eqref{prop2_estimate_un_nablaun}, there holds 
	\eqref{prop2_inductionclaimestimate} (by \eqref{prop1_Pin_explicit_formula} $\nabla \Pi^n$ has the same 
	regularity of $\Delta u^n$).  Moreover, from the arbitrary of $p_3$ we get also
	\begin{equation*}
		\|t^{\frac{1}{2}-\frac{1}{2r}}(u^n,\,\nabla d^n)\|_{L^{2r}_t L^\infty_x}+
		\|t^{\frac{1}{2}}(u^n,\,\nabla d^n)\|_{L^{\infty}_t L^\infty_x}\lesssim \eta.
	\end{equation*}
		
	\noindent Now, we claim by induction that $(u^n,\,\nabla d^n,\,\nabla \Pi^n)$ belongs to $\Y_r^\ee$ and 
	moreover 
	\begin{equation}\label{smooth_prop3_induction}
		\|(u^n,\,\nabla d^n,\,\nabla \Pi^n)\|_{\Y_r^\ee}\lesssim \|(u_0,\,\nabla d_0)\|_{\BB_{p,r}^{
		\frac{N}{p}-1+\ee}},
	\end{equation}	 
	for every $n\in\NN$ (uniformly). 
	Recalling \eqref{prop2_dn} and \eqref{prop2_un}, by Lemma \ref{Lemma4b} with $\tilde{q}=q_3$ or 
	$\tilde{q}=3q_1$ and $q=p_1$, we get that
	\begin{equation*}
	\begin{split}
		\|t^{\gamma_1^\ee-\ee}(u^n,\,\nabla d^n)&\|_{L^{2r}_t L^{q_3}_x} + 
		\|t^{\gamma_3^\ee-\ee}\nabla d^n\|_{L^{2r}_t L^{3q_1}_x} +
		\|t^{\gamma_4^\ee-\ee}\nabla d^n\|_{L^{\infty}_t L^{3q_1}_x}\lesssim 
		\|(u_0,\,\nabla d_0)\|_{\BB_{p,r}^{\frac{N}{p}-1}}+\\&+
		\|t^{\alpha_1}\big(\nabla \{ |\nabla d^{n-1}|^2 d^n - u^{n-1}\cdot \nabla d^{n-1} \},\,
		\{-u^{n-1} \cdot \nabla u^n+ F_n\})\|_{L^{2r}_TL^{p_1}_x},
	\end{split}
	\end{equation*}
	for every $n\in\NN$. Hence, since $\|(u^n,\,\nabla d^n,\,\nabla\Pi^n)\|_{\Y_r}\lesssim \eta$, we deduce 
	the following uniformly estimate:
	\begin{equation}\label{prop3smooth_estimate_no_eps}
		\|t^{\gamma_1^\ee-\ee}(u^n,\,\nabla d^n)\|_{L^{2r}_t L^{q_3}_x} + 
		\|t^{\gamma_3^\ee-\ee}\nabla d^n\|_{L^{2r}_t L^{3q_1}_x}+
		\|t^{\gamma_4^\ee-\ee}\nabla d^n\|_{L^{\infty}_t L^{3q_1}_x} \lesssim \eta.
	\end{equation}
	We still proceed by induction and the base case $(0,\,0,\,0)\in \Y_r^\ee$ is trivial. 
	Now, let us assume that $(u^{n-1},\,\nabla d^{n-1},\,\nabla \Pi^{n-1})$ belongs to $\Y_r^\ee$. At first,
	since $(u_0,\,\nabla d_0)\in \BB_{q_3,2r}^{N/q_3-1+\ee}$, $\nabla (u_0,\,\nabla d_0)
	\in \BB_{q_2,2r}^{N/q_2-2+\ee}$ and $\nabla^2(u_0,\,\nabla d_0)\in \BB_{q_1,2r}^{N/q_1-3+\ee}$, 
	by Theorem \ref{Characterization_of_hom_Besov_spaces} we get that the conditions for $u$ and $\nabla d$ 
	determined by $\Y_r^\ee$ are satisfied by $e^{t\Delta}u_0$ and $e^{t\Delta}\nabla d_0$.
	Now, arguing as for proving \eqref{recal_1}, we get
	\begin{equation*}
	\begin{split}
		&\|t^{\alpha_1^\ee}\nabla^3 d^n\|_{L^{2r}_t L^{q_1}_x}+
		\|t^{\alpha_2^\ee}\nabla^3 d^n\|_{L^{r}_t L^{q_1}_x}+
		\|t^{\gamma_1^\ee}  \nabla d^n \|_{L^{2r}_t L^{q_3}_x}+
		\|t^{\gamma_2^\ee}  \nabla d^n \|_{L^{\infty}_t L^{q_3}_x}+
		\|t^{\gamma_3^\ee}  \nabla d^n \|_{L^{2r}_t L^{3q_1}_x}+\\&+
		\|t^{\gamma_4^\ee}  \nabla d^n \|_{L^{\infty}_t L^{3q_1}_x}
		\lesssim 
		\|\nabla d_0\|_{\BB_{p,r}^{\frac{N}{p}-1+\ee}} + 
		\|t^{\alpha_1^\ee}\nabla \{ |\nabla d^{n-1}|^2 d^n - u^{n-1}\cdot \nabla d^{n-1}\}
		\|_{L^{2r}_t L^{q_1}_x}+\\&\quad\quad\quad\quad\quad\quad\quad\quad\quad\quad\quad
		\quad\quad\quad\quad\quad\quad\quad\quad\quad+
		\|t^{\alpha_2^\ee}\nabla \{ |\nabla d^{n-1}|^2 d^n - u^{n-1}\cdot \nabla d^{n-1}\}
		\|_{L^{r}_t L^{q_1}_x}.
	\end{split}	
	\end{equation*}
	Thus, recalling \eqref{prop3smooth_estimate_no_eps}, we get
	\begin{equation*}
	\begin{split}
		&\|t^{\alpha_1^\ee}\nabla^3 d^n\|_{L^{2r}_t L^{q_1}_x}+
		\|t^{\alpha_2^\ee}\nabla^3 d^n\|_{L^{r}_t L^{q_1}_x}+
		\|t^{\gamma_1^\ee}  \nabla d^n \|_{L^{2r}_t L^{q_3}_x}+
		\|t^{\gamma_2^\ee}  \nabla d^n \|_{L^{\infty}_t L^{q_3}_x}+\\&+
		\|t^{\gamma_3^\ee}  \nabla d^n \|_{L^{2r}_t L^{3q_1}_x}+
		\|t^{\gamma_4^\ee}  \nabla d^n \|_{L^{\infty}_t L^{3q_1}_x}
		\lesssim  
		\|\nabla d_0\|_{\BB_{p,r}^{\frac{N}{p}-1+\ee}}+
		\eta  
		\big(
		\|t^{\gamma_4^\ee} \nabla d^n\|_{L^{\infty}_tL^{3q_1}_x}+ \\&\quad+
		\|t^{\gamma_2^\ee} \nabla d^{n-1} \|_{L^{\infty}_t L^{q_3}_x}+
		\|t^{\beta_2^\ee}  \nabla^2 d^{n-1} \|_{L^{\infty}_t L^{q_2}_x}+
		\|t^{\gamma_1^\ee} \nabla d^{n-1} \|_{L^{2r}_t L^{q_3}_x}+
		\|t^{\beta_1^\ee}  \nabla^2 d^{n-1} \|_{L^{2r}_t L^{q_2}_x}
		\big).
	\end{split}	
	\end{equation*}
	Applying the induction hypotheses we obtain \eqref{smooth_prop3_induction} at least for the terms 
	concerning $\nabla d^n$. Moreover, arguing as for proving \eqref{prop2_estimate_nabla2un} we get
	\begin{equation*}
	\begin{split}
		\|t^{\alpha_1^\ee} \nabla^2 u^n \|_{L^{2r}_t L^{q_1}_x}&+
		\|t^{\alpha_2^\ee} \nabla^2 u^n \|_{L^{r}_t L^{q_1}_x}\lesssim 
		\|t^{\alpha_1^\ee} F^n\|_{L^{2r}_t L^{q_1}_x}+\\
		&+
		\|t^{\alpha_1^\ee} u^{n-1}\cdot\nabla u^n\|_{L^{2r}_t L^{q_1}_x}+
		\|t^{\alpha_2^\ee} F^n\|_{L^{r}_t L^{p_1}_x}+
		\|t^{\alpha_2^\ee} u^{n-1}\cdot\nabla u^n\|_{L^{r}_t L^{q_1}_x}, 
	\end{split}
	\end{equation*}
	and as for proving \eqref{prop2_estimate_un_nablaun} 
	\begin{equation*}
		\|t^{\gamma_1^\ee}u^n\|_{L^{2r}_t L^{q_3}_x} + 
		\|t^{\gamma_2\ee}u^n\|_{L^{\infty}_t L^{q_3}_x} +
		\|t^{\beta_1^\ee}\nabla u^n\|_{L^{2r}_t L^{q_2}_x}+
		\|t^{\beta_2^\ee}\nabla u^n\|_{L^{\infty}_t L^{q_2}_x}\lesssim
		\|t^{\alpha_1^\ee} F^n\|_{L^{2r}_t L^{q_1}_x}.
	\end{equation*}
	From the definition of $F^n$ and by the induction hypotheses, we get \eqref{smooth_prop3_induction} also 
	for the terms concerning $u^n$. Since by \eqref{prop1_Pin_explicit_formula} $\nabla \Pi^n$ has the same 
	regularity of $\nabla^2 u^n$, we finally obtain \eqref{smooth_prop3_induction}.

	\noindent Let us observe that, by Remark \ref{remark1}, for every $T>0$ there exists a positive 
	$\hat{C}(T)>0$ such that
	\begin{equation*}
		\|\nabla u^n\|_{L^1_T L^\infty_x}\lesssim 
		\hat{C}(T)\|(u_0,\,\nabla d_0)\|_{\BB_{p,r}^{\frac{N}{p}-1+\ee}}.
	\end{equation*}
	
	\noindent To conclude the proof we want to show that $(a^n,\,d^n,\,u^n)_\NN$ is a Cauchy sequence in 
	the considered spaces. The strategy is similar to the last part of Theorem 
	\ref{Theorem_solutions_smooth_dates}. Denoting $\delta u^n:=u^{n+1}-u^n$ and so on for 
	$\delta d^n$, $\delta a^n$ and $\delta \Pi^n$, for all $T>0$ we define
	\begin{equation*}
		\delta U^n(T):=\|(\delta u^n,\,\nabla \delta d^n,\,\nabla \Pi^n)\|_{\Y_{r,T}}+
		\|\delta d^n\|_{L^\infty_TL^\infty_x}+
		\|t^{\frac{1}{2}}\delta u^n\|_{L^{\infty}_T L^\infty_x}.
	\end{equation*} 
	We want to prove that $\sum_{n\in \NN}\delta U^n(T)$ is finite. Let us consider $\delta a^n$ which 
	is solution of \eqref{prop1_eq_delta_an}. By standard estimates for the transport equation and 
	by \eqref{prop2_estimate_nabla_an}, we obtain 
	\begin{equation}\label{prop2_estimate_delta_an}
		\| \delta a^n(t)\|_{L^\infty_x}\leq 
		\int_0^t \|\delta u^{n-1}(s)\|_{L^\infty_x} \|\nabla a^n(s)\|_{L^\infty_x}\dd s\leq 
		\bar{C}(T)\Big(\int_0^T \frac{1}{s^{\frac{3}{4}}}\dd s\Big)^{\frac{2}{3}}
		\|\delta U^{n-1}(t)\|_{L^3_T}\|\nabla a_0\|_{L^\infty_x},
	\end{equation}
	where $\bar{C}(T)=\exp\{
	\hat{C}(T)\|(u_0,\,\nabla d_0)\|_{\BB_{p,r}^{N/p-1+\ee}} \}$.
	Considering $\delta d^n$, we recall that is solution of \eqref{prop1_eq_delta_dn}. Hence, by Lemma 
	\ref{Lemma6} with $\sigma = 2\gamma_1$ and $q=p_3/2$, we get
	\begin{equation*}
		\| \delta d^n(t) \|_{L^\infty_x}\leq \|s^{2\gamma_1} \delta H^n \|_{L^{r}(0,t;L^{\frac{p_3}{2}}_x)}
		\lesssim \eta \delta U^{n-1}(T) + \eta^2 \|\delta d^n\|_{L^\infty_T L^\infty_x},
	\end{equation*}
	for every $t\in (0,T)$. Taking the sup on $t\in (0,T)$ we deduce 
	\begin{equation*}
		\|\delta d^n\|_{L^\infty_TL^\infty_x}\lesssim \eta \delta U^{n-1}(T).
	\end{equation*}
	Moreover, by Lemma \ref{Lemma4b}, Lemma \ref{Lemma5b}, Lemma \ref{Lemma6} and Lemma \ref{Lemma7}, we 
	obtain
	\begin{equation*}
	\begin{split}
		\|t^{\beta_1} &\nabla^2 \delta d^n \|_{L^{2r}_T L^{p_2}_x}+ 
		\|t^{\beta_2}  \nabla^2 \delta d^n\|_{L^{\infty}_T L^{p_2}_x} +
		\|t^{\gamma_1} \nabla 	\delta d^n \|_{L^{2r}_T L^{p_3}_x}+
		\|t^{\gamma_2} \nabla \delta d^n \|_{L^{\infty}_T L^{p_3}_x}+\\ 
		&+
		\|t^{\gamma_3} \nabla \delta d^n \|_{L^{2r}_T L^{3p_1}_x}+
		\|t^{\gamma_4} \nabla \delta d^n \|_{L^{\infty}_T L^{3p_1}_x}+
		\| \nabla \delta d^n\|_{L^{\infty}_T W^{1,\infty}_x}\lesssim\\
		&\quad\lesssim
		\|t^{\alpha_1}\nabla \delta H^n\|_{L^{2r}_TL^{p_1}_x} + 
		\|t^{\gamma_1}\nabla \delta H^n\|_{L^{2r}_TL^{p_3}_x}
		\lesssim \eta\big(\delta U^{n-1}(T) + \delta U^n(T)\big)
	\end{split}
	\end{equation*}
	Recalling \eqref{prop1_estimates_delta_un}, \eqref{prop1_Pin_explicit_formula} and still using 
	Lemmas \ref{Lemma4b}, \ref{Lemma5b}, \ref{Lemma6} and \ref{Lemma7}, we finally obtain
	\begin{equation*}
		\delta U^n(T)\lesssim \eta \delta U^{n-1}(T) + \eta\| \delta a^n\|_{L^\infty_TL^\infty_x},
	\end{equation*}
	and by \eqref{prop2_estimate_delta_an} we finally obtain. 
	\begin{equation*}
		\delta U^n(T)\lesssim \eta \delta U^{n-1}(T) + 
		\bar{C}(T)T^{\frac{1}{6}}\|\delta U^{n-1}\|_{L^3_T}\|\nabla a_0\|_{L^\infty_x}.
	\end{equation*}
	Such inequality is strictly similar to \eqref{estimate_deltaUn}, hence we can conclude 
	the proof of the proposition arguing exactly as in the last part of the proof of proposition 
	\ref{Theorem_solutions_smooth_dates}.
\end{proof}
\begin{section}{Existence of a Global Solution}

\noindent Let us now tackle the proof to the existence part of our main results, namely Theorem \ref{Main_Theorem} and Theorem \ref{Main_Theorem_2}.
Thanks to the dyadic partition we regularize the initial velocity $u_0$ and the initial molecular orientation $d_0$, while we regularize the initial density $a_0$ by a family of mollificators. 
The key is to use the existence results and the estimates of the previous section, constructing a family of solutions for \eqref{Navier_Stokes_system} with the regularized initial data. Due to the low regularity of $a_0$, it is not possible to prove the strong convergence of such approximate solutions. Hence, we shall focus on a compactness method, along the same line of 
\cite{MR3250369} and \cite{MR3056619}.

\vspace{0.2cm}
\noindent Let $(\chi_n)_\NN$ be a family of mollifiers, we define $a_{0,n}:= \chi_n*a_0$, for every $n\in\NN$. $a_{0,n}$ belongs to $W^{1,\infty}_x$ and its $L^\infty_x$-norm is bounded by $\|a_0\|_{L^\infty_x}$. Moreover, $(a_{0,n})_\NN$ weak* converges to $a_0$ up to a subsequence (which we still denote by $(a_{0,n})_\NN$). Since $d_0$ belongs to $L^\infty_x$, which is a subset of 
$\BB_{\infty,\infty}^{0}$, and $u_0$ belongs to $\BB_{p,r}^{N/p-1}$, we cut the low and the high frequencies in the following way:
\begin{equation*}
	u_{0,n}:=\sum_{|k|\leq n}\dot{\Delta}_k u_0,\quad\quad
	d_{0,n}:=\sum_{|k|\leq n}\dot{\Delta}_k d_0.
\end{equation*}
Each term $d_{0,n}$ belongs to $L^\infty_x$ with norm bounded by 1. Moreover $u_{0,n}$ and $\nabla d_{0,n}$ belong to $\BB^{s}_{p,1}$ for every real number $s$. In addition, the smallness condition \eqref{smallness_condition} is still valid for $(a_{0,n},\,u_{0,n},\,\nabla d_{0,n})$. 

\begin{proof}[Proof of Theorem \ref{Main_Theorem}]
	As already pointed out, $u_{0,n}$ and $\nabla d_{0,n}$ belong to $\BB^{s}_{p,1}$, 
	for every real number $s$, in particular for $s=N/p-1$ and $s=N/p+1$. The hypotheses of 
	Proposition \ref{Theorem_solutions_smooth_dates} are fulfilled, hence it determines  
	$(u^n,\,d^n,\,a^n)$ solution of \eqref{Navier_Stokes_system} with $u_{0,n}$, $d_{0,n}$ and 
	$a_{0,n}$ as initial data. Furthermore we get the following uniform estimates for the norms 
	of such solutions: 
	\begin{equation*}
		\|(u^n,\,\nabla d^n,\,\nabla \Pi^n)\|_{\X_r} 
		\lesssim \eta
	\end{equation*}
	for every $n\in\NN$. By these inequalities and the momentum equation of 
	\eqref{Navier_Stokes_system}, $(\partial_t u^n)_\NN$ is a bounded sequence in 
	$L^r_t L^{Nr/(3r-2)}_x$ and $(\partial_t d^n)_\NN$ is a bounded sequence in 
	$L^r_t L^{Nr/(2r-2)}_x$. Thus, applying Ascoli-Arzela Theorem, we conclude that there exists 
	a subsequence of $(u^n,\,d^n,\,a^n)_\NN$ (which we still denote by $(u^n,\,d^n,\,a^n)_\NN$) 
	and some $(u,\,d,\,a,\,\nabla \Pi)$ with $a,\,d\in L^\infty_{t,x}$ and $(u,\,\nabla d,\,\nabla\Pi)\in \X_r$ 
	such that
	\begin{equation*}
	\begin{alignedat}{3}
		(a^n,\, d^n) &\rightharpoonup (a,\,d)  &&\quad\text{weak}\;*\text{in}\; L^{\infty}_{t,x},\\
		\nabla d^n &\rightharpoonup \nabla d  &&\quad\text{weakly}\; \text{in}\; 
		L^{3r}_t L^{\frac{3Nr}{3r-2}}_x,\,L^{2r}_tL^{\frac{Nr}{r-1}}_x,\, L^2_t L^\infty_x,\\
		\nabla( u^n,\,\nabla d^n) & \rightharpoonup \nabla(u,\,\nabla d)&& \quad
		\text{weakly}\; \text{in}\; 
		L^{2r}_tL^{\frac{Nr}{2r-1}}_x,\, L^r_t L^{\frac{Nr}{r-1}},\\ 
	\end{alignedat}
	\end{equation*}
	with in addition
	\begin{equation*}
		\nabla( \nabla u^n,\, \Pi^n) \rightharpoonup \nabla(\nabla u,\, \Pi)
		\quad\text{weakly}\;\text{in}\; L^{r}_{t}L^\frac{Nr}{3r-2}_x
	\end{equation*}
	and
	\begin{equation*}
		u^n \rightarrow u \quad\text{strongly in}\; 
		L^r_{t,loc} L^{\frac{Nr}{r-1}-\bar{\ee}}_{x,loc},
	\end{equation*}
	for all positive $\bar{\ee}$ small enough. 
	The last strongly convergence is due to an interpolation result, observing that, for every 
	$T>0$, the sequence $(u_n-e^{t\Delta}u_{0,n})_\NN$ is uniformly 
	bounded and equicontinuous in $C([0,T],L^{Nr/(3r-2)}_x)$ and moreover 
	$(e^{t\Delta}u_{0,n})_\NN$ converges to $e^{t\Delta}u_0$ strongly in $L^r_t L^{Nr/(r-1)}_x$ 
	(since $(u_{0,n})_\NN$ converges to $u_0$ strongly in $\BB_{p,r}^{N/r-1}$).
	We deduce that $u^n\cdot \nabla d^n$ and $u^n\cdot \nabla u^n$ converge to $u\cdot\nabla 
	d$ and $u\cdot \nabla u$ respectively. Then, it is sufficient to prove that $a^n(\Delta u^n + 
	\nabla \Pi^n) $ 
	converge to $a(\Delta u + \nabla\Pi) $ in the distributional sense, in order to conclude that 
	$(u,\,d,\,a)$ is a solution for \eqref{Navier_Stokes_system} with initial data $(u_0,\,d_0,
	\,a_0)$. Toward this, we shall follow \cite{MR3250369} and \cite{MR3056619}, proving that 
	$(a^n)_\NN$ strongly converges to $a$ in $L^m_{loc}(\RR_+\times\RR^N)$ for any $m<\infty$. 
	thanks to the transport equation of \eqref{Navier_Stokes_system} we have
	\begin{equation*}
		\partial_t (a^n)^2 + u^n\cdot \nabla (a^n)^2 = 0,
	\end{equation*}
	which yields
	\begin{equation*}
		\partial_t \omega + u \cdot \omega = 0,
	\end{equation*}
	where $\omega$ is the weak $*$ limit of $((a^n)^2)_\NN$ (up to a subsequence). Moreover, 
	by a mollifying method as that in \cite{MR1022305}, we infer that
	\begin{equation*}
		\partial_t a^2 + \Div (u\,a^2) = 0.
	\end{equation*}
	Thus
	\begin{equation*}
	\begin{cases}
		\partial_t (a^2-\omega) + \Div \{u(a^2-\omega)\}=0	&\RR_+\times\RR^N\\
		(a^2-\omega)|_{t=0}=0,
	\end{cases}
	\end{equation*}
	and from the uniqueness of the Transport equation (see \cite{MR1022305}) we conclude that 
	$a^2-\omega=0$ almost everywhere. We deduce that $(\|a^n\|_{L^2(\Omega)})_\NN$ converges to 
	$\|a\|_{L^2(\Omega)}$, for every $\Omega$ bounded subset of $\RR_+\times\RR^N$, hence 
	$(a^n)_\NN$ strongly converges to $a$ in $L^2(\Omega)$. By interpolation, we deduce that 
	$(a^n)_\NN$ strongly converges to $a$ in $L^m_{loc}(\RR_+\times\RR^N)$ for any $m<\infty$ 
	and this completes the proof to the existence part of Theorem \ref{Main_Theorem}.
\end{proof}

\begin{proof}[Proof of Theorem \ref{Main_Theorem_2}]
	We proceed along the same line of the previous prove, using Proposition 
	\ref{Theorem_solutions_smooth_dates_2} instead of Proposition 
	\ref{Theorem_solutions_smooth_dates}. We get the following uniform estimates for the 
	sequence of the approximate solutions:
	\begin{equation*}
		\|(u^n,\,\nabla d^n,\,\nabla\Pi^n)\|_{\Y_r}		
		\lesssim\eta.
	\end{equation*}
	Since $\alpha_2 r'<1$, $(\nabla^2 u^n,\nabla\Pi^n )_\NN = (t^{-\alpha_2 }t^{\alpha_2}
	(\nabla^2 u^n,\,\nabla \Pi^n)_\NN$ is 
	uniformly bounded in $L^{\tau_1}_TL^{p_1}_x$, where $\tau_1$ belongs to $(1,r/(1+\alpha_2r)\,)$ 
	and $T>0$. 
	Similarly $(\nabla u^n, \nabla^2 d^n)_\NN$ and $(u^n,\,\nabla d^n)_\NN$ are uniformly bounded 
	in $L^{\tau_2}_TL^{p_2}_x$ and $L^{\tau_3}_TL^{p_3}_x$ respectively, where $\tau_2\in 
	(1,2r/(1+\beta_1 2r)\,)$ and $\tau_3\in (1,\,2r/(1+\gamma_12r)\,)$. It is not restrictive to 
	choose $\tau_2$ and $\tau_3$ such that $1/\tau_4:=1/\tau_2+1/\tau_3$ is less than $1$. 
	Hence $(\partial_t u^n)_\NN$ is uniformly bounded in $L^{\tau_1}(0,T; L^p_1)$ which yields that 
	$(u^n-e^{t\Delta}u_{0,n})_\NN$ is uniformly bounded and equicontinuous in $C([0,T],L^{p_1}_x)$.
	Moreover $(e^{t\Delta}u_{0,n})_\NN$ converges to $e^{t\Delta}u_0$ in $L^{\tau_3}(0,T;L^{p_3})$.
	Hence, by Ascola-Arzela Theorem, we conclude that, up to extraction, the sequence $(u^n,\,d^n,\,
	a^n,\,\nabla \Pi^n)_\NN$ converges to some $(u,d,\,a,\,\nabla \Pi)_\NN$ such that $a,\,d$ belong 
	to $L^\infty_{t,x}$ and $(u,\,\nabla d,\,\nabla \Pi)\in \Y_r$. The convergence is in the following sense:
	\begin{equation*}
	\begin{alignedat}{3}
		(a^n,\, d^n) &\rightharpoonup (a,\,d)  &&\quad\text{weak}\;*\text{in}\; 
		L^{\infty}_{t,loc},\\
		\nabla d^n &\rightharpoonup \nabla d  &&\quad\text{weakly}\; \text{in}\; 
		L^{\tau_3}_{t,loc} L^{p_3}_x,\\
		\nabla( u^n,\,\nabla d^n) & \rightharpoonup \nabla(u,\,\nabla d)&& \quad
		\text{weakly}\; \text{in}\; 
		L^{\tau_2}_{t,loc}L^{p_2}_x, \\ 
	\end{alignedat}	
	\end{equation*}
	with in addition
	\begin{equation*}
		\nabla( \nabla u^n,\, \Pi^n) \rightharpoonup \nabla(\nabla u,\, \Pi)
		\quad\text{weakly}\;\text{in}\; L^{\tau_1}_{t,loc}L^{p_1}_x
	\end{equation*}
	and
	\begin{equation*}
		u^n \rightarrow u \quad\text{strongly in}\; 
		L^{\tau_3}_{t,loc} L^{p_3-\ee}_{x,loc},
	\end{equation*}
	for all positive $\ee$ small enough. 
	Finally we can repeat the argument at the end of the proof of Theorem \ref{Main_Theorem}, 
	concluding the existence part of Theorem \ref{Main_Theorem_2}.
\end{proof}
\end{section}
\section{Lagrangian Coordinates}

\noindent The uniqueness result is basically based on the Lagrangian coordinates concept. The key is to rewrite system \eqref{Navier_Stokes_system} under such coordinates, obtaining a new formulation which allows the uniqueness in the functional framework of the main Theorems. This strategy has already been treated in \cite{MR2957705} on a restrictive family of \eqref{Navier_Stokes_system}, namely the incompressible Navier-Stokes equations with variable density in the whole space. We claim to extend it to the general simplified Ericksen-Leslie system. Before going on, in this section we recall some mainly results concerning the Lagrangian coordinates.

\vspace{0.2cm}
\noindent Let $T\in (0,\infty]$, we consider a vector field $u$ in $L^1_T\mathcal{L}ip_x$. The flow $X$ of $u$ is defined as the solution of the following ordinary differential equation:
\begin{equation*}
\begin{cases}
	\partial_t X(t,y) = u(t,X(t,y))	&(t,y)\in \RR_+\times \RR^n,\\
	X(0,y) = y							&y\in \RR^n.
\end{cases}
\end{equation*}
The unique solution is granted by Cauchy-Lipschitz Theorem. Defining $v(t,y):=u(t,X(t,y))$ we get the following relation between the Eulerian coordinates $x$ and the Lagrangian coordinates $y$: 
\begin{equation*}
	x = X(t,y) = y+\int_0^tv(s,y)\dd s.
\end{equation*}
Furthermore, fixing $t'\in\RR$, let $\tilde{X}=\tilde{X}(t',t,x)$ be the unique solution of
\begin{equation*}
\begin{cases}
	\partial_t \tilde{X}(t',t,x) = u(t,X(t',t,x))	&(t,x)\in (t',\infty)\times \RR^n,\\
	\tilde{X}(t',t', x) = x							&y\in \RR^n.
\end{cases}
\end{equation*}
Then $Y=Y(t,x)=\tilde{X}(-t,0,x)$ is the inverse map of $X$. Setting $D:=\,^t\nabla$, we get 
$A(t,y):=(D_yX)^{-1}(t,y) = D_xY(t,X(t,y))$ and moreover
\begin{equation*}
	\|A(t)-\Id\|_{L^\infty_x}\leq \int_0^t \|D_y v(s)\|_{L^\infty_x}\dd s.
\end{equation*} 
Assuming that $u$ has $L^1_T\mathcal{L}ip_x$-norm small enough, we obtain that the right-hand side of the previous inequality is less than 1. Thus $A(t,y)$ is determined by
\begin{equation*}
	A(t,y)=D_xY(t,X(t,y))=(\Id + ( D_yX(t,y)-\Id))^{-1} = \sum_{k=-\infty}^{\infty}(-1)^k
	\Big( \int_0^tD_yv(s,y)\dd s \Big)^k.
\end{equation*}
Furthermore 
\begin{equation*}
\begin{array}{cc}
	(\nabla_x u)(t,X(t,y)) = \,^{t}A(t,y)\nabla_yv(t,y),		
	&(\Div_x u)(t,X(t,y)) = \Div_y\{A(t,y)v(t,y)\},\\
\end{array}
\end{equation*}
Setting $b(t,y):=a(t,X(t,y))$, $\omega(t,y):=d(t,X(t,y)$, $P(t,y):=\Pi(t,X(t,y))$ and moreover
$h(t,y) := (\nabla_x d )(t, X(t,y))=\,^tA(t,y)\nabla_y\omega(t,y) $, system \eqref{Navier_Stokes_system} becomes 
\begin{equation*}
\begin{cases}
	\,\partial_t b =	 0 											&\RR_+\times\RR^N,\\
	\,\partial_t v -\Div_y\{\,A\,^tA\nabla_y v\} + \,^tA\nabla_yP
	=-\Div_y\{\,A h\odot h\}										&\RR_+\times\RR^N,\\
	\,\partial_t \omega - \,^tA : \nabla h = |h|^2\omega			&\RR_+\times\RR^N,\\								
	\,\Div (\,^tA\, v ) = 0 										&\RR_+\times\RR^N,\\
	h = \,^t A \nabla_y \omega										&\RR_+\times\RR^N,\\
	\;(v,b, \omega)_{|t=0} = (v_0,\,b_0, \, \omega_0)				& \;\;\quad \quad\RR^N,\\
\end{cases}
\end{equation*}
which is the Lagrangian formulation. Moreover, taking the derivative in $x$ to the third equation of \eqref{Navier_Stokes_system}
\begin{equation*}
	\partial_t \nabla d +u\cdot \nabla^2 d+ \nabla u\cdot \nabla d -\Delta \nabla d = 
	2\nabla d \cdot \nabla ^2 d + |\nabla d|^2\nabla d,
\end{equation*}
thus, $h$ is solution of
\begin{equation*}
	\partial_t h + (\,^tA\nabla_y v)\cdot \nabla_y h - \Div_y (A \nabla_y h) = 2h\cdot \nabla_y h\, \omega + |h|^2h.
\end{equation*}

\section{Uniqueness}

\noindent This section is devoted to the proof of Theorem \ref{Main_Theorem_3} and Theorem \ref{Main_Theorem_4}. For $i=1,2$, let $(u_i,\,d_i,\,a_i)$ be two solutions of \eqref{Navier_Stokes_system} satisfying the condition of Theorem \ref{Main_Theorem_2}. Let $X_i$ be the flow generated by $u_i$, for $i=1,2$, and $(v_i,\,\omega_i, b_i)$ the Lagrangian formulations of the solutions. At fist, let us observe that $b_1\equiv b_2\equiv b_0$, thus setting $\delta v := v_1-v_2$, $\delta \omega := \omega_1 - \omega_2$ and $\delta P:= P_1-P_2$, we observe that $(\delta v,\,\delta \omega\,\delta h,\,\delta P)$ is solution for
\begin{equation}\label{Navier_Stokes_Lagrangian_coordinates_difference}
\begin{cases}
	\,\partial_t \delta v - \Delta \delta v + \nabla \delta P =  b_0(\Delta \delta v - \nabla \delta P)
	+\delta f_1+\delta f_2+ \delta f_3											&\RR_+\times\RR^N,\\
	\partial_t \delta \omega = \delta f_4 +\delta f_5							&\RR_+\times\RR^N,\\
	\partial_t \delta h - \Delta \delta h = \delta f_6+\delta f_7+\delta f_8	&\RR_+\times\RR^N,\\
	\Div\{\delta v\} =\delta g													&\RR_+\times\RR^N,\\
	\partial_t \delta g = \delta R												&\RR_+\times\RR^N,\\
	(\delta v,\,\delta \omega, \,\delta h)_{|t=0}=(0,0,0)						& \;\;\quad \quad\RR^N,\\
	\end{cases}
\end{equation} 
where
\begin{equation*}
\begin{aligned}
	\delta f_1&:=(1+a_0)[(\Id-\,^tA_2)\nabla\delta P -\delta A \nabla P_1],\\
	\delta f_2&:= (1+a_0)\Div\{ (A_2\,^tA_2 -\Id)\nabla \delta v +(A_2\,^tA_2 -A_1\,^tA_1)\nabla v_1\},\\
	\delta f_3&:=\Div\big\{
	\delta A (h_2 \odot \,h_2)+A_1 (\delta h \odot \,h_2)+	A_1\, (h_1 \odot \,\delta h)\big\},\\
	\delta f_4&:=\delta h\cdot h_2 \omega_2 + h_1\cdot \delta h\,\omega_2 + |h_1|^2\delta \omega,\\
	\delta f_5&:= \,^t\delta A : \nabla h_2 + \,^t A_1 : \nabla \delta h,
\end{aligned}
\end{equation*}
and
\begin{equation*}
\begin{aligned}
	\delta f_6&:=-\,^t\delta A \nabla v_2\cdot h_2 -\,^tA_1\nabla \delta v\cdot h_2-\,^tA_1
	\nabla v_1\cdot \delta h,\\
	\delta f_7&:= \Div\big\{(\,^tA_2-\Id)\nabla\delta h +\,^t\delta A \nabla h_1\big\},\\
	\delta f_8&:= 2\delta h\cdot \nabla h_2\, \omega_2 + 2h_1 \cdot \nabla \delta h\, \omega_2 + 
	h_1 \cdot \nabla h_1\delta \omega,\\
	\delta g&:= (\Id-\,^tA_2): \nabla\delta v + \,^t\delta A : \nabla v_1,\\
	\delta R&: =  \partial_t \big[ (\Id -A_2)\delta v \big] - \partial_t \big[\delta A v_1 \big].
\end{aligned}
\end{equation*}
\noindent In what follows, we will use repeatedly the following identity:
\begin{equation}\label{uniqueness_formula_deltaA}
	\delta A(t) =\Big( \int_0^t D\delta v(\tau)\dd \tau\Big)\Big(\sum_{k\geq 1} \sum_{0\leq j < k}
	C_1^j(t) C_2^{k-1-j}(t)\Big),
\end{equation}
with
\begin{equation*}
	C_i(t):=\int_0^t Dv_i(\tau),\quad \text{for}\quad
	i=1,\,2.
\end{equation*}

\subsection{Uniqueness: the smooth case}

\noindent Let us assume that $1<p<Nr/(3r-2)$, $\ee\in (0,1]$ and $r\in (1, 2/(2-\ee)$. 
We suppose our initial data $(u_0,\,\nabla d_0)$ to be in $\BB_{p,r}^{N/p-1+\ee}\cap \BB_{p,r}^{N/p-1}$ and we want to prove that the solution for \eqref{Navier_Stokes_system_1}, given by Theorem \ref{Main_Theorem}, is unique. First, let us observe that our solution belongs to the functional framework of Theorem \ref{Main_Theorem_3}, thanks to proposition \ref{Theorem_solutions_smooth_dates}.  
Now, let us tackle the proof of the uniqueness. We need the following Lemma
\begin{lemma}\label{Lemma2_uniq_estimate}
	Let $T>0$ and let us assume that $f$, $\nabla g$, and $R$ belong to $L^r_T L^{Nr/(3r-2)}_x$. Then
	\begin{equation*}
	\begin{cases}
		\,\partial_t v - \Delta v + \nabla P = f		&(0,T)\times\RR^N,\\
		\,\Div v = g									&(0,T)\times\RR^N,\\
		\,\partial_t g = \Div\, R							&(0,T)\times\RR^N,\\
		v_{|t=0}=0										&\quad\quad\quad\;\;\RR^N,
	\end{cases}
	\end{equation*} 
	admits a unique solution such that
	\begin{equation*}
	\begin{split}
		\|v							\|_{L^{2r}_T L^{\frac{Nr}{r-1}}_x} 		+ 
		\|v							\|_{L^{2}_T L^{\infty}_x}				&+
		\|\nabla v					\|_{L^{2r}_T L^{\frac{Nr}{2r-1}}_x} 	+
		\|\nabla v					\|_{L^{2r}_T L^{\frac{Nr}{2(r-1)}}_x} 	+\\&+
		\|(\partial_t v,\, \nabla^2 v,\,\nabla \Pi)	\|_{L^r_T	 L^{\frac{Nr}{3r-2}}_x}
		\lesssim 
		\|(f,\,\nabla g,\,R)\|_{L^r_T	 L^{\frac{Nr}{3r-2}}_x}.
	\end{split}
	\end{equation*}
\end{lemma}
\begin{proof}
	Applying $-\Div$ to the first equation, we get
	\begin{equation*}
		-\Delta P = \Div\, \{R-f-\nabla g\},
	\end{equation*}
	which yields $\nabla P = RR\cdot \{R-f-\nabla g\}$.  Hence, $\|\nabla P\|_{L^q_x}\leq \|R,\,f,\,\nabla g 
	\|_{L^q_x}$, for every $q\in (1,\infty)$. Moreover, $v$ is determined by
	\begin{equation*}
		v(t)= \int_0^t e^{(t-s)\Delta}(f-\nabla P)(s)\dd s,
	\end{equation*}
	thus by Lemmas \ref{Lemma2}, \ref{Lemma3} and Theorem \ref{Maximal_regularity_theorem}, we 
	obtain the required estimate.
\end{proof}
\noindent Thus, recalling system \eqref{Navier_Stokes_Lagrangian_coordinates_difference}, we get
\begin{equation*}
\begin{split}
	\|\delta v							\|_{L^{2r}_T L^{\frac{Nr}{r-1}}_x} 	    + 
	\|\delta v							\|_{L^{2}_T L^{\infty}_x}		 	   &+	
	\|\nabla \delta v					\|_{L^{2r}_T L^{\frac{Nr}{2r-1}}_x} 	+
	\|\nabla \delta v					\|_{L^{2r}_T L^{\frac{Nr}{2(r-1)}}_x} 	+
	\|\delta v					\|_{L^{2}_T L^{\infty}_x} 	+\\&+
	\|(\partial_t\delta  v,\, \nabla^2\delta  v,\,\nabla\delta  \Pi)	\|_{L^r_T	 L^{\frac{Nr}{3r-2}}_x}
	\lesssim 
	\|(\delta f_1,\,\delta f_2,\,\delta f_3,\,\nabla \delta g, \delta R)\|_{L^r_T L^{\frac{Nr}{3r-2}}_x},
\end{split}
\end{equation*}
where we have also used that  $\|b_0\|_{L^\infty_x}=\|a_0\|_{L^\infty_x}\leq \eta$. Furthermore by the second equation of \eqref{Navier_Stokes_Lagrangian_coordinates_difference}, we get $\delta \omega
\in L^\infty_T L^{Nr/2(r-1)}_x$ and
\begin{equation*}
	\|\delta \omega\|_{L^\infty_T L^{\frac{Nr}{2(r-1)}}_x}\lesssim 
	\|(\delta f_4,\,\delta f_5)\|_{L^1_T L^\frac{Nr}{2(r-1)}_x},
\end{equation*}
and by Theorem \ref{Maximal_regularity_theorem}, Lemmas \ref{Lemma2} and \ref{Lemma3}, we get
\begin{equation*}
\begin{split}
	\|\delta h								\|_{L^{ 2r }_T L^{ \frac{Nr}{r-1}   }_x}  &+
	\|\delta h								\|_{L^{ 3r }_T L^{ \frac{3Nr}{3r-2} }_x}   +
	\|\delta h								\|_{L^{ 2  }_T L^{ \infty           }_x}   +
	\|\nabla \delta h						\|_{L^{ 2r }_T L^{ \frac{Nr}{2r-1}  }_x}   +\\&+
	\|\nabla \delta h						\|_{L^{ 2r }_T L^{ \frac{Nr}{2(r-1)}}_x}   +
	\|\nabla^2 \delta h						\|_{L^{ r  }_T L^{ \frac{Nr}{3r-2}  }_x}
	\lesssim
	\|(\delta f_6,\,\delta f_7,\,\delta f_8)\|_{L^{ r  }_T L^{ \frac{Nr}{3r-2}  }_x}.
\end{split}
\end{equation*}

\noindent Summarizing the previous inequality, we need to control the right-hand side of
\begin{equation*}
\begin{split}
	\|(\delta& v,\,\delta h,\,\delta P)\|_{\X_{r,T}} + \|\partial_t v\|_{L^r_T L^{\frac{Nr}{3r-2}_x}}
	+ \|\delta \omega\|_{L^\infty_T L^{\frac{Nr}{2(r-1)}}_x}
	\lesssim \\&
	\lesssim
	\|(\delta f_1,\,\delta f_2,\,\delta f_3,\,\delta f_6,\, \delta f_6\, \delta f_7,\, \delta f_8,\,
	\nabla \delta g, \delta R)\|_{L^r_T L^{\frac{Nr}{3r-2}}_x}+
	\|(\delta f_4,\,\nabla \delta f_5,\,\delta f_6)\|_{L^1_T L^\frac{Nr}{2(r-1)}_x}.
\end{split}
\end{equation*}

\noindent We are going to estimate each of these terms step by step. Moreover, in what follows we will use that
\begin{equation*}
\begin{aligned}
	\|(\nabla^2 v_i,\,\nabla^2 h_i,\nabla P_i)\|_{L^r	  _T  L^N                         _x}&\lesssim 
	\|(\nabla^2 v_i,\,\nabla^2 h_i,\nabla P_i)\|_{L^r     _T  L^{\frac{Nr}{(3-\ee)r-2}}   _x}+
	\|(\nabla^2 v_i,\,\nabla^2 h_i,\nabla P_i)\|_{L^r     _T  L^{\frac{Nr}{3r-2}}         _x}<\infty,\\		
	\|             \nabla (v_i,\,h_i)         \|_{L^{2r}  _T  L^{Nr}                      _x}&\lesssim
	\|             \nabla (v_i,\,h_i)         \|_{L^{2r}  _T  L^{\frac{Nr}{2r-1}}         _x}+			
	\|             \nabla (v_i,\,h_i)         \|_{L^{2r}  _T  L^{\frac{Nr}{(2-\ee)r-1}}   _x}<\infty,\\
	\|					\nabla A_i			  \|_{L^\infty_T  L^{\frac{Nr}{(3-\ee)r-2}}   _x}&\lesssim
	\|					\nabla A_i			  \|_{L^\infty_T  L^N						  _x}+
	\|					\nabla A_i			  \|_{L^\infty_T  L^{\frac{Nr}{3r-2}}         _x}<\infty,
\end{aligned}
\end{equation*}
for $i=1,2$.

\noindent	
\textit{Bounds for $\delta f_1$}. From the definition of $\delta f_1$, we readily get
\begin{equation*}
	\|\delta f_1\|_{L^{r}_TL^{\frac{Nr}{3r-2}}_x}\lesssim 
	\|\Id-\,^tA_2\|_{L^{\infty}_TL^{\infty}_x}
	\|\nabla \delta P\|_{L^{r}_TL^{\frac{Nr}{3r-2}}_x} + 
	\|\delta A\|_{L^\infty_T L^{\frac{Nr}{2(r-1)}}_x}\|\nabla P\|_{L^{2r}_TL^{N}_x}, 
\end{equation*}
where $Nr/(2r-2)$ is the Lebesgue exponent in the critical Sobolev embedding
\begin{equation*}
	W^{1,\frac{Nr}{3r-2}}_x\hookrightarrow L^{\frac{Nr}{2(r-1)}}_x.
\end{equation*}
Consequently, because $T<1$, recalling \eqref{uniqueness_formula_deltaA}, we obtain
\begin{equation*}
	\|\delta f_1\|_{L^{r}_TL^{\frac{Nr}{3r-2}}_x}\lesssim 
	\|\nabla v_2\|_{L^{\frac{2}{2-\ee}}_T L^{\infty}_x}
	\|\nabla \delta P\|_{L^{r}_TL^{\frac{Nr}{3r-2}}_x} + 
	\|\nabla^2 \delta v\|_{L^r_TL^{\frac{Nr}{3r-2}}_x}
	\|\nabla P\|_{L^{r}_TL^{\frac{Nr}{3r-2}}_x}.
\end{equation*}
Thus there exists a continuous function $t\rightarrow \chi_1(t)$, which goes to $0$ for $t\rightarrow 0$ and
\begin{equation}\label{uniqueness1_delta_f_1}
	\|\delta f_1\|_{L{r}_TL^{\frac{Nr}{3r-2}}_x}\lesssim \chi_1(T)
	\|(\delta v,\,\delta h,\,\nabla \delta P)\|_{\X_{r,T}}.
\end{equation}

\noindent	
\textit{Bounds for $\delta f_2$.}  From the definition of $\delta f_2$  and observing that 
$A_2\,\,^tA_2-A_1\,\,^tA_1=\delta A\,^t A_2 +  A_1\,^t \delta A$, we deduce that
\begin{equation*}
\begin{split}
	\|&\delta f_2\|_{L^{r}_TL^{\frac{Nr}{3r-2}}_x}\lesssim
	\|\nabla(A_2\,^tA_2)\|_{L^{\infty}_T L^{N}_x} 
	\|\nabla \delta v\|_{L^{r}_T L^{\frac{Nr}{2(r-1)}}_x} + 
	\|A_2\,^tA_2\|_{L^{\infty}_TL^\infty_x} \|\nabla^2 \delta v\|_{L^{r}_TL^{\frac{Nr}{3r-2}}_x} +\\&+
	\|\delta A\|_{L^\infty_TL^{\frac{Nr}{2(r-1)}}_x}  
	(\|\nabla A_2\|_{L^{\infty}_T L^{N}_x} + \|\nabla A_1 \|_{L^{\infty}_T L^{N}_x})
	\|\nabla v_1\|_{L^{\frac{2}{2-\ee}}_TL^\infty_x} +
	\|\nabla \delta A\|_{L^\infty_TL^{\frac{Nr}{3r-2}}_x}  
	(\| A_2\|_{L^{\infty}_T L^{\infty}_x} +\\&+ \| A_1 \|_{L^{\infty}_T L^{\infty}_x})
	\|\nabla v_1\|_{L^{\frac{2}{2-\ee}}_TL^\infty_x}+
	\|\delta A\|_{L^\infty_TL^{\frac{Nr}{2(r-1)}}_x}  
	(\| A_2\|_{L^{\infty}_T L^{\infty}_x} + \| A_1 \|_{L^{\infty}_T L^{\infty}_x})
	\| \nabla^2 v_1\|_{L^{r}_TL^{N}_x}
\end{split}
\end{equation*}
Hence, there exists a continuous function $t\rightarrow \chi_2(t)$ which goes to $0$ for $t\rightarrow 0$, 
such that 
\begin{equation}\label{uniqueness1_delta_f_2}
	\|\delta f_2\|_{L^{r}_TL^{\frac{Nr}{3r-2}}_x}\lesssim 
	 \chi_2(T)\|(\delta v,\, \delta  h,\,\nabla \delta P)\|_{\X_{r,T}}.
\end{equation} 

\noindent
\textit{Bounds for $\delta f_3$} From the definition of $\delta f_3$ we get
\begin{equation*}
\begin{split}
	\|\delta f_3&\|_{L^{r}_T L^{\frac{Nr}{3r-2}}_x}\lesssim \\
	&\quad\;
	\|\nabla \delta A \|_{L^\infty_T L^{\frac{Nr}{3r-2}}_x} 
	\|h_1\|_{L^{\frac{2}{2-\ee}}_T L^{\infty}_x}^2+
	\|\delta A \|_{L^\infty_T L^{\frac{Nr}{2r-1}}_x} 
	\|\nabla h_1\|_{L^{\frac{2}{2-\ee}}_T L^{\infty}_x}
	\|h_1\|_{L^{2r}_T L^{\frac{Nr}{r-1}}_x}+\\&+
	\|\nabla A_1\|_{L^\infty_TL^N_x}
	\|\nabla \delta h\|_{L^{2r}_TL^{\frac{Nr}{2r-1}}_x} 
	\|h_2\|_{L^{2r}_T L^{\frac{Nr}{r-1}}_x}+
	\|\nabla A_1\|_{L^\infty_TL^N_x}
	\|\delta h\|_{L^{2r}_TL^{\frac{Nr}{r-1}}_x} 
	\|h_2\|_{L^{2r}_T L^{\frac{Nr}{r-1}}_x}+
	\\&+
	\|A_1\|_{L^\infty_TL^\infty_x}
	\|\nabla \delta h\|_{L^{2r}_TL^{\frac{Nr}{2r-1}}_x} 
	\|h_2\|_{L^{2r}_T L^{\frac{Nr}{r-1}}_x} +
	\|A_1\|_{L^\infty_TL^\infty_x}
	\|\delta h\|_{L^{2r}_TL^{\frac{Nr}{r-1}}_x} 
	\|\nabla h_2\|_{L^{2r}_T L^{\frac{Nr}{2r-1}}_x} +\\&+	
	\|\nabla A_1\|_{L^\infty_TL^N_x}
	\|\nabla \delta h\|_{L^{2r}_TL^{\frac{Nr}{2r-1}}_x} 
	\|h_1\|_{L^{2r}_T L^{\frac{Nr}{r-1}}_x}+
	\|\nabla A_1\|_{L^\infty_TL^N_x}
	\|\delta h\|_{L^{2r}_TL^{\frac{Nr}{2r-1}}_x} 
	\|h_1\|_{L^{2r}_T L^{\frac{Nr}{r-1}}_x}+
	\\&+
	\|A_1\|_{L^\infty_TL^\infty_x}
	\|\nabla \delta h\|_{L^{2r}_TL^{\frac{Nr}{2r-1}}_x} 
	\|h_1\|_{L^{2r}_T L^{\frac{Nr}{r-1}}_x} +
	\|A_1\|_{L^\infty_TL^\infty_x}
	\|\delta h\|_{L^{2r}_TL^{\frac{Nr}{r-1}}_x} 
	\|\nabla h_1\|_{L^{2r}_T L^{\frac{Nr}{r-2}}_x}\end{split}
\end{equation*}
Hence, there exists $\chi_3(t)$ such that
\begin{equation}\label{uniqueness1_delta_f_3}
	\|\delta f_3\|_{L^{r}_T L^{\frac{Nr}{3r-2}}_x}\lesssim 
	\chi_3(T)\|(\delta v,\, \delta h,\,\nabla \delta P)\|_{\X_{r,T}}.
\end{equation}

\noindent
\textit{Bounds for $\delta f_6$} From the definition of $\delta f_6$ we get
\begin{equation*}
\begin{split}
	\|&\delta f_6\|_{L^{r}_TL^{\frac{Nr}{3r-2}}_x}\lesssim 
	\|\delta A \|_{L^\infty_TL^{\frac{Nr}{2r-1}}_x}
	\|\nabla v_2\|_{L^{\frac{2}{2-\ee}}_TL^\infty_x}
	\|h_2\|_{L^{2r}_TL^{\frac{Nr}{r-1}}_x} +\\& +
	\|A_1\|_{L^\infty_TL^\infty_x} 
	\|\nabla \delta v\|_{L^{2r}_TL^{\frac{Nr}{2r-1}}_x}
	\|h_2\|_{L^{2r}_TL^{\frac{Nr}{r-1}}_x}
	+
	\|A_1\|_{L^\infty_TL^\infty_x} 
	\|\nabla  v_1\|_{L^{2r}_TL^{\frac{Nr}{r-1}}_x}
	\|\delta  h\|_{L^{2r}_TL^{\frac{Nr}{r-1}}_x},
\end{split}
\end{equation*}
Thus
\begin{equation}\label{uniqueness1_delta_f_6}
	\|\delta f_6\|_{L^{r}_TL^{\frac{Nr}{3r-2}}_x}\lesssim
	\chi_6(T)\|(\delta v,\, \delta h,\,\nabla \delta P)\|_{\X_{r,T}},
\end{equation}
for an opportune continuous function $\chi_7(t)$ which goes to $0$ when $t\rightarrow 0$.

\noindent	
\textit{Bounds for $\delta f_7$}. From the definition of $\delta f_7$ we have
\begin{equation*}
\begin{split}
	\|\delta f_7  \|_{L^{r}_TL^{\frac{Nr}{3r-2}}_x}\lesssim 
	\|\nabla A_2\|_{L^\infty_T L^N_x}&
	\|\nabla^2 \delta h\|_{L^{r}_T L^{\frac{Nr}{3r-2}}_x}+
	\|A_2-\Id\|_{L^\infty_T L^\infty_x}
	\|\nabla^2 \delta h \|_{L^{r}_T L^{\frac{Nr}{3r-2}}_x}+ \\&+
	\|\nabla \delta A \|_{L^\infty_TL^{\frac{Nr}{3r-2}}_x} \|\nabla h_1\|_{L^{\frac{2}{2-\ee}}_TL^{\infty}_x} + 
	\|\delta A\|_{L^\infty_T L^{\frac{Nr}{2(r-1)}}_x}\|\nabla^2 h_1\|_{L^{2r}_T L^N_x},
\end{split}
\end{equation*}
which yields that there exists a continuous function $\chi_7(t)\geq 0$, with $\chi_7(0)=0$, such that
\begin{equation}\label{uniqueness1_delta_f_7}
	\|\delta f_7  \|_{L^{r}_TL^{\frac{Nr}{3r-2}}_x}\lesssim \chi_7 (T)
	\|(\delta v,\, \delta h,\,\nabla \delta P)\|_{\X_{r,T}}.
\end{equation}

\noindent \textit{Bounds for $\delta f_8$}. From the definition of $\delta f_8$ we get 
\begin{equation*}
\begin{split}
	\|\delta f_8&\|_{L^{r}_TL^{\frac{Nr}{3r-2}}_x}\lesssim  
	\|\delta h\|_{L^{2r}_TL^{\frac{Nr}{r-1}}_x}
	\|\nabla h_2 \|_{L^{2r}_T L^{\frac{Nr}{2r-1}}_x}
	\|\omega_2\|_{L^\infty_TL^\infty_x} +\\&+ 
	\|h_1\|_{L^{2r}_TL^{\frac{Nr}{r-1}}_x}
	\|\nabla \delta h \|_{L^{2r}_T L^{\frac{Nr}{2r-1}}_x} 
	\|\omega_2\|_{L^\infty_TL^\infty_x} +
	\|h_1\|_{L^{2r}_T L^{\frac{Nr}{r-1}}_x}\| \nabla h_1 \|_{L^{2r}_T L^{Nr}_x}
	\|\delta \omega\|_{L^\infty_TL^\frac{Nr}{2(r-1)}_x} 	 
\end{split}
\end{equation*}
which yields that there exists a continuous function $\chi_8(t)\geq 0$, with $\chi_8(0)=0$, such that
\begin{equation}\label{uniqueness1_delta_f_8}
	\|\delta f_8  \|_{L^{r}_TL^{\frac{Nr}{3r-2}}_x}\lesssim \chi_9 (T)\Big\{
	\|(\delta v,\, \delta h,\,\nabla \delta P)\|_{\X_{r,T}} +
	\|\delta \omega\|_{L^\infty_TL^\frac{Nr}{2(r-1)}_x}
	\Big\}.
\end{equation}

\noindent	
\textit{Bounds for $\nabla \delta g $}. By the definition of $\delta g$ we get
\begin{equation*}
\begin{split}
	\|\nabla \delta g\|_{L^{r}_TL^{\frac{Nr}{3r-2}}_x} \lesssim 
	\|\nabla A_2&\|_{L^\infty_TL^N_x}\|\nabla \delta v\|_{L^{r}_T L^{\frac{Nr}{2(r-1)}}_x} + 
	\|\Id-A_2\|_{L^\infty_TL^\infty_x} \|\nabla^2 \delta v\|_{L^{r}_TL^{\frac{Nr}{3r-2}}_x} + \\&+
	\|\nabla \delta A\|_{L^{\infty}_TL^{\frac{Nr}{3r-2}}_x}\|\nabla v_1\|_{L^{\frac{2}{2-\ee}}_TL^\infty_x} + 
	\|\delta A\|_{L^\infty_T L^{\frac{Nr}{2(r-1)}}_x}\|\nabla^2 v_1\|_{L^{r}_T L^N_x}.
\end{split}
\end{equation*}
We deduce that there exists a continuous function $\chi_g(t)$ with $\chi_g(0)=0$ such that
\begin{equation}\label{uniqueness1_nabla_delta_g}
	\|\nabla \delta g\|_{L^{r}_TL^{\frac{Nr}{3r-2}}_x} \lesssim 
	\chi_g(T)\|(\delta v,\, \delta h,\,\nabla \delta P)\|_{\X_{r,T}}.
\end{equation} 

\noindent
\textit{Bounds for $\delta R $}. From the definition of $\delta R$ we have
\begin{equation*}
\begin{split}
	\|\delta R\|_{L^{r}_TL^{\frac{Nr}{3r-2}}_x}\lesssim 
	\|\nabla v_2&\|_{L^{2r}_TL^{\frac{Nr}{2r-1}}_x}\|\delta v\|_{L^{2r}_TL^{\frac{Nr}{r-1}}_x}+
	\|\Id-A_2\|_{L^{\infty}_T L^\infty_x}\|\partial_t \delta v\|_{L^{r}_TL^{\frac{Nr}{3r-2}}_x}+\\&+
	\|\nabla\delta v\|_{L^{2r}_TL^{\frac{Nr}{2r-1}}_x}\|v_1\|_{L^{2r}_TL^{\frac{Nr}{r-1}}_x}
	+\|\delta A \|_{L^{\infty}_T L^{\frac{Nr}{2(r-1)}}_x}\| \partial_t v_1 \|_{L^r_TL^N_x},
\end{split}
\end{equation*}
Thus, there exists  a continuous function $\chi_R(t)$ with $\chi_R(0)=0$ such that
\begin{equation}\label{uniqueness1_delta_R}
	\|\delta R\|_{L^{r}_TL^{\frac{Nr}{3r-2}}_x}\lesssim  
	\chi_R(T)\Big\{\|(\delta v,\, \delta h,\,\nabla \delta P)\|_{\X_{r,T}}+ 
	\|\partial_t\delta v\|_{L^{r}_TL^{\frac{Nr}{3r-2}}_x}\Big\}.
\end{equation}

\noindent	
\textit{Bounds for $\delta f_4 $}. From the definition of $\delta f_4$ it follows
\begin{equation*}
\begin{split}
	\|\delta f_4 \|_{L^{1}_T  L^{\frac{Nr}{2(r-1)}}_x}
	&\lesssim 
	\|\delta	h		\|_{L^{2r}_T L^{\frac{Nr}{r-1}}_x}
	\| (h_1,\,h_2)	\|_{L^{2r}_T L^{\frac{Nr}{r-1}}_x} 
	\| \omega_2\|_{L^\infty_T L^\infty_x}	+\\
	&+\| h_1\|_{L^{2}_T L^{\infty}_x}^2
	\|\delta \omega\|_{L^\infty_T L^\frac{Nr}{2(r-1)}_x}.
\end{split}
\end{equation*}
Therefore, we obtain
\begin{equation}\label{uniqueness1_delta_f_4}
	\|\delta 	f_4 \|_{L^{1}_T  L^{\frac{Nr}{2(r-1)}}_x}\lesssim \chi_4 (T)\Big\{
	\|(\delta v,\, \delta h,\,\nabla \delta P)\|_{\X_{r,T}} +
	\|\delta \omega\|_{L^\infty_TL^\frac{Nr}{2(r-1)}_x}
	\Big\}.
\end{equation}

\noindent	
\textit{Bounds for $\delta f_5 $}. From the definition of $\delta f_5$ it follows
\begin{equation*}
\begin{split}
	\|\delta f_5 \|_{L^{1}_T  L^{\frac{Nr}{2(r-1)}}_x}
	\lesssim 
	\|\delta A\|_{L^\infty_T L^{\frac{Nr}{2(r-1)}}_x}
	\|\nabla h_2\|_{L^{\frac{2}{2-\ee}}_t L^\infty_x} + 
	 t^{1-\frac{1}{r}}\|A_1\|_{L^\infty_T L^\infty_x}
	 \|\nabla \delta h\|_{L^{r}_t L^{\frac{Nr}{2(r-1)}}_x} 
\end{split}
\end{equation*}
Therefore, we obtain
\begin{equation}\label{uniqueness1_delta_f_5}
	\|\delta 	f_5 \|_{L^{1}_T  L^{\frac{Nr}{2(r-1)}}_x}\lesssim \chi_5 (T)
	\|(\delta v,\, \delta h,\,\nabla \delta P)\|_{\X_{r,T}}.
\end{equation}

\noindent Summarizing 
\eqref{uniqueness1_delta_f_1},
\eqref{uniqueness1_delta_f_2},
\eqref{uniqueness1_delta_f_3},
\eqref{uniqueness1_delta_f_6},
\eqref{uniqueness1_delta_f_7},
\eqref{uniqueness1_delta_f_8},
\eqref{uniqueness1_nabla_delta_g},
\eqref{uniqueness1_delta_R},
\eqref{uniqueness1_delta_f_4}
 and 
\eqref{uniqueness1_delta_f_5},
we deduce that there exists a continuous function $\chi (t)=\sum_i \chi_i(t)$ which assume $0$ for $t=0$, 
such that 
\begin{equation*}
\begin{split}
	\|(\delta v,\,\delta h,\,\delta P)&\|_{\X_{r,T}} + \|\partial_t v\|_{L^r_T L^{\frac{Nr}{3r-2}_x}}
	+ \|\delta \omega\|_{L^\infty_T L^{\frac{Nr}{2(r-1)}}_x}
	\lesssim \\&\lesssim
	\chi(T)
	\Big\{\,
		\|(\delta v,\,\delta h,\,\delta P)\|_{\X_{r,T}} + 
		\|\partial_t v\|_{L^r_T L^{\frac{Nr}{3r-2}_x}}  +  
		\|\delta \omega\|_{L^\infty_T L^{\frac{Nr}{2(r-1)}}_x}\,
	\Big\},
\end{split}
\end{equation*}
which yields the uniqueness of the solution to \eqref{Navier_Stokes_system} on a sufficiently small interval. 
Then uniqueness part can be completed by a bootstrap method. 

\subsection{Uniqueness: the general case}

\noindent Now let us consider the general case $1<p<N$, $\ee\in (0,\min\{1/r,\,1-1/r,\,N/p-1\}]$ and our initial data $(u_0,\,\nabla d_0)$ in $\BB_{p,r}^{N/p-1+\ee}\cap \BB_{p,r}^{N/p-1}$. We want to prove that the solution for \eqref{Navier_Stokes_system_1}, given by Theorem \ref{Main_Theorem_2}, is unique. Let us observe that our solution belongs to the functional framework of Theorem \ref{Main_Theorem_3}, thanks to proposition \ref{Theorem_solutions_smooth_dates_2}. We also recall Remark \ref{remark1} for the Lispschitz-estimates and suppose $T<1$. In order to prove the uniqueness we need the following Lemma
\begin{lemma}
	Let $\alpha_1,\,\beta_i,\,\gamma_j$ and $p_1,\,p_2,\,p_3$ be defined by Theorem \ref{Main_Theorem_2}, 
	for $i=1,2$ and $j=1,2,3$. If $t^{\alpha_1}f$, $t^{\alpha_1}\nabla g$ and $t^{\alpha_1}R$ belong to 
	$L^{2r}_TL^{p_1}_x$, then
	\begin{equation*}
	\begin{cases}
		\,\partial_t v - \Delta v + \nabla P = f		&(0,T)\times\RR^N,\\
		\,\Div v = g									&(0,T)\times\RR^N,\\
		\,\partial_t g = \Div\, R							&(0,T)\times\RR^N,\\
		v_{|t=0}=0										&\quad\quad\quad\;\;\RR^N,
	\end{cases}
	\end{equation*}
	has a unique solution such that
	\begin{equation}\label{Lemma1_uniq_estimate}
	\begin{split}
		\|t^{\gamma_1}v\|_{L^{2r}_TL^{p_3}_x} &+ 
		\|t^{\gamma_2}v\|_{L^{\infty}_TL^{p_3}_x} +
		\|t^{\gamma_3}v\|_{L^{2r}_TL^{3p_1}_x} +
		\|t^{\gamma_4}v\|_{L^{\infty}_TL^{3p_1}_x} + 
		\|t^{\beta_1}v\|_{L^{2r}_TL^{p_2}_x}+ \\ &+
		\|t^{\beta_2}v\|_{L^{\infty}_TL^{p_2}_x}+
		\|t^{\alpha_1}(\partial_t v,\,\nabla^2 v,\,\nabla P)\|_{L^{2r}_T L^{p_1}_x}
		\lesssim 
		\|t^{\alpha_1}(f,\,\nabla g,\,R)\|_{L^{2r}_T L^{p_1}_x}
	\end{split}
	\end{equation}
\end{lemma}
\begin{proof}
	The proof is basically equivalent to the one of Lemma \ref{Lemma2_uniq_estimate}
\end{proof}

\noindent
By \eqref{Navier_Stokes_Lagrangian_coordinates_difference} and the previous Lemma, it follows that 
\begin{equation*}
\begin{split}
	&\|t^{\beta_1}\nabla \delta v\|_{L^{2r}_TL^{p_2}_x}+
	\|t^{\beta_2}\nabla \delta v\|_{L^{\infty}_TL^{p_2}_x}+
	\|t^{\gamma_1}\delta v\|_{L^{2r}_TL^{p_3}_x} + 
	\|t^{\gamma_2}\delta v\|_{L^{\infty}_TL^{p_3}_x} +
	\|t^{\gamma_3}\delta v\|_{L^{2r}_TL^{3p_1}_x} +\\ &+
	\|t^{\gamma_4}\delta v\|_{L^{\infty}_TL^{3p_1}_x} +
	\|t^{\alpha_1}(\partial_t\delta  v,\,\nabla^2\delta  v,\,\nabla\delta  P)\|_{L^{2r}_T L^{p_1}_x}
	\lesssim 
	\|t^{\alpha_1}(\delta f_1,\,\delta f_2,\,\delta f_3,\,\nabla\delta  g,\,\delta R)
	\|_{L^{2r}_T L^{p_1}_x},
\end{split}
\end{equation*}
where we have also used that $\|b_0\|_{L^\infty_x}=\|a_0\|_{L^\infty_x}\leq \eta$.
Furthermore, by the second equation of \eqref{Navier_Stokes_Lagrangian_coordinates_difference} we get 
$t^{\alpha_1}\delta \omega\in L^{\infty}_T L^{p_1^*}_x$, where $p_1^*=p_1N/(N-p_1)$ is the Lebesgue exponent in the critical Sobolev embedding
\begin{equation*}
	W^{1,p_1}_x\hookrightarrow L^{p_1^*}_x.
\end{equation*}
Moreover
\begin{equation*}
	\|t^{\alpha_1}\delta \omega\|_{L^{\infty}_T L^{p_1^*}_x}\lesssim 
	\|t^{\alpha_1}(\delta f_4,\,\delta f_5)\|_{L^1_T L^{p_1^*}_x}.
\end{equation*}

\noindent By Theorem \ref{Maximal_regularity_Thm_weight_time}, Lemma \ref{Lemma4b} and Lemma \ref{Lemma5b}, we get
\begin{equation*}
\begin{split}
	\|t^{\gamma_1}&\delta h\|_{L^{2r}_TL^{p_3}_x} + 
	\|t^{\gamma_2}\delta h\|_{L^{\infty}_TL^{p_3}_x} +
	\|t^{\gamma_3}\delta h\|_{L^{2r}_TL^{3p_1}_x} +
	\|t^{\gamma_4}\delta h\|_{L^{\infty}_TL^{3p_1}_x} + \\ &+
	\|t^{\beta_1}\nabla\delta  h\|_{L^{2r}_TL^{p_2}_x} +
	\|t^{\beta_2}\nabla\delta  h\|_{L^{\infty}_TL^{p_2}_x}	+
	\|t^{\alpha_1}\nabla^2 \delta h\|_{L^{2r}_TL^{p_1}_x}
	\lesssim 
	\|t^{\alpha_1}(\delta f_6,\,\delta f_7,\,\delta f_8)\|_{L^{2r}_T L^{p_1}_x}.
\end{split}
\end{equation*}

\noindent
Summarizing the last inequalities, we deduce that we have to control the right-hand side of
\begin{equation}\label{uniqueness_what_we_need_to_estimate}
\begin{split}
	\|(\delta v,\, \delta h, \nabla \delta &P)\|_{\Y_{r,T}} + 
	\|t^{\alpha_1}\partial_t\delta v\|_{L^{2r}_TL^{p_1}_x}+
	\|t^{\alpha_1}\delta \omega\|_{L^\infty_TL^{p_1^*}_x}
	\lesssim 
	\|t^{\alpha_1}(\delta f_1,\,\delta f_2,\,\nabla\delta f_3)
	\|_{L^{2r}_TL^{p_1}_x}+\\ &+
	\|t^{\alpha_1}(\delta f_6,\, \delta f_7,\,\delta f_8,\nabla \delta g,\,\delta R)\|_{L^{2r}_TL^{p_1}_x}+ 
	\|t^{\alpha_1}(\delta f_4,\,\delta f_5)\|_{L^1_T L^{p_1^*}_x}.
\end{split}
\end{equation}
Let us now estimate the right-hand side of \eqref{uniqueness_what_we_need_to_estimate} term by term.  
\begin{remark}
	In what follows, we will use repeatedly the following estimates:
	\begin{equation*}
	\begin{aligned}
		&t^{\alpha_1}< t^{\alpha_1^\ee}\quad \text{for}\quad t\leq T<1,\\
		&\|\nabla \delta A\|_{L^\infty_T L^{p_1}_x}\lesssim \|\nabla^2 \delta v\|_{L^1_TL^{p_1}_x}
		\lesssim \|t^{\alpha_1}\nabla^2 \delta v\|_{L^{2r}_T L^{p_1}_x},\\
		&\| \delta A\|_{L^\infty_T L^{p_2}_x}\lesssim \|\nabla \delta v\|_{L^1_TL^{p_2}_x}
		\lesssim \|t^{\beta_1}\nabla \delta v\|_{L^{2r}_T L^{p_2}_x}.
	\end{aligned}
	\end{equation*}
	Moreover, if in addition we consider $p_3=\infty$ we get also the following estimate
	\begin{equation*}
	\|t^{\frac{1}{2}-\frac{1}{2r}}h_1\|_{L^{2r}_T L^{\infty}_x} +
	\|t^{\frac{1}{2}} h_1\|_{L^{\infty}_T L^{\infty}_x}<\infty.
	\end{equation*}
\end{remark}

\noindent	
\textit{Bounds for $t^{\alpha_1}\delta f_1$}. From the definition of $\delta f_1$, we readily get
\begin{equation*}
	\|t^{\alpha_1}\delta f_1\|_{L^{2r}_TL^{p_1}_x}\lesssim 
	\|\Id-\,^tA_2\|_{L^{\infty}_TL^{\infty}_x}
	\|t^{\alpha_1}\nabla \delta P\|_{L^{2r}_TL^{p_1}_x} + 
	\|\delta A\|_{L^\infty_T L^{p_1^*}_x}\|t^{\alpha_1}\nabla P\|_{L^{2r}_TL^{N}_x}, 
\end{equation*}
where $p_1^*=p_1N/(N-p_1)$ is the Lebesgue exponent in the critical Sobolev embedding
\begin{equation*}
	W^{1,p_1}_x\hookrightarrow L^{p_1^*}_x.
\end{equation*}
Consequently, recalling \eqref{uniqueness_formula_deltaA} and observing that $L^{N}_x\hookrightarrow L^{p_1}_x\cap L^{q_1}_x$, we obtain
\begin{equation*}
	\|t^{\alpha_1}\delta f_1\|_{L^{2r}_TL^{p_1}_x}\lesssim 
	\|t^{\alpha_1^\ee}\nabla v_2\|_{L^{2r}_T L^{\infty}_x}
	\|t^{\alpha_1}\nabla \delta P\|_{L^{2r}_TL^{p_1}_x} + 
	\|\nabla^2\delta v\|_{L^1_TL^{p_1}_x}
	\|t^{\alpha_1}\nabla P\|_{L^{2r}_TL^{p_1}_x}^\theta
	\|t^{\alpha_1}\nabla P\|_{L^{2r}_TL^{q_1}_x}^{1-\theta},	
\end{equation*}
for $\theta$ determined by $1/N=\theta/p_1 +(1-\theta)/q_1$. We get
\begin{equation*}
\begin{split}
	\|t^{\alpha_1}\delta f_1\|_{L^{2r}_TL^{p_1}_x}
	&\lesssim 
	\|t^{\alpha_1^\ee}\nabla v_2\|_{L^{2r}_T L^{\infty}_x}
	\|t^{\alpha_1}\nabla \delta P\|_{L^{2r}_TL^{p_1}_x} + 
	\|\nabla^2 \delta v\|_{L^1_TL^{p_1}_x}
	\|t^{\alpha_1}\nabla P\|_{L^{2r}_TL^{p_1}_x}^\theta
	\|t^{\alpha_1^\ee}\nabla P\|_{L^{2r}_TL^{q_1}_x}^{1-\theta},\\
	&\lesssim 
	\|t^{\alpha_1^\ee}\nabla v_2\|_{L^{2r}_T L^{\infty}_x}\|t^{\alpha_1}\nabla \delta P\|_{L^{2r}_TL^{p_1}_x} + 
	\|t^{-\alpha_1}\|_{L^{(2r)'}_T}\| t^{\alpha_1} \nabla^2 \delta v\|_{L^{2r}_TL^{p_1}_x}\\
	&\lesssim 
	\|t^{\alpha_1^\ee}\nabla v_2\|_{L^{2r}_T L^{\infty}_x}\|t^{\alpha_1}\nabla \delta P\|_{L^{2r}_TL^{p_1}_x} + 
	T^{1-\alpha_1(2r)'}\| t^{\alpha_1} \nabla^2 \delta v\|_{L^{2r}_TL^{p_1}_x}
\end{split}
\end{equation*}
Thus there exists a continuous function $t\rightarrow \chi_1(t)$, which goes to $0$ for $t\rightarrow 0$ and
\begin{equation}\label{uniqueness2_delta_f1}
	\| t^{\alpha_1}\delta f_1\|_{L{2r}_TL^{p_1}_x}\lesssim \chi_1(T)\|(\delta v,\,
	\delta h,\,\nabla \delta P)\|_{\Y_{r,T}}
\end{equation}

\noindent	
\textit{Bounds for $t^{\alpha_1}\delta f_2$.}  From the definition of $\delta f_2$  and observing that 
$A_2\,\,^tA_2-A_1\,\,^tA_1=\delta A\,^t A_2 +  A_1\,^t \delta A$, we deduce
\begin{equation*}
\begin{split}
	\|&t^{\alpha_1}\delta f_2\|_{L^{2r}_TL^{p_1}_x}\lesssim
	\|\nabla(A_2\,^tA_2)\|_{L^{\infty}_T L^{N}_x} 
	\|t^{\alpha_1}\nabla \delta v\|_{L^{2r}_T L^{p_1^*}_x} + 
	\|A_2\,^tA_2\|_{L^{\infty}_TL^\infty_x} \|t^{\alpha_1}\nabla^2 \delta v\|_{L^{2r}_TL^{p_1}_x} +\\&+
	\|\delta A\|_{L^\infty_TL^{p_1^*}_x}  
	(\|\nabla A_2\|_{L^{\infty}_T L^{N}_x} + \|\nabla A_1 \|_{L^{\infty}_T L^{N}_x})
	\|t^{\alpha_1^\ee}\nabla v_1\|_{L^{2r}_TL^\infty_x} +\\&+
	\|\nabla \delta A\|_{L^\infty_TL^{p_1}_x}  
	(\| A_2\|_{L^{\infty}_T L^{\infty}_x} + \| A_1 \|_{L^{\infty}_T L^{\infty}_x})
	\|t^{\alpha_1^\ee}\nabla v_1\|_{L^{2r}_TL^\infty_x}+\\&+
	\| \delta A\|_{L^\infty_TL^{p_1^*}_x}  
	(\| A_2\|_{L^{\infty}_T L^{\infty}_x} + \| A_1 \|_{L^{\infty}_T L^{\infty}_x})
	\| t^{\alpha_1} \nabla^2 v_1\|_{L^{2r}_TL^{N}_x}
\end{split}
\end{equation*}
Again by $L^{N}_x\hookrightarrow L^{p_1}_x\cap L^{q_1}_x$ and the critical Sobolev embedding, there exists 
a continuous function $t\rightarrow \chi_2(t)$ which goes to $0$ for $t\rightarrow 0$, such that 
\begin{equation}\label{uniqueness2_delta_f2}
	\|t^{\alpha_1}\delta f_2\|_{L^{2r}_TL^{p_1}_x}\lesssim 
	 \chi_2(T)\|(\delta v,\, \delta  h,\,\nabla \delta P)\|_{\Y_{r,T}}.
\end{equation} 

\noindent
\textit{Bounds for $t^{\alpha_1} \delta f_3$} From the definition of $\delta f_3$ we get
\begin{equation*}
\begin{split}
	\|&t^{\alpha_1} \delta f_3\|_{L^{2r}_T L^{p_1}_x}\lesssim
	\|\nabla \delta A \|_{L^\infty_T L^{p_1}_x} 
	\|t^{\frac{1}{2}-\frac{1}{2r}}h_1\|_{L^{2r}_T L^{\infty}_x} 
	\|t^{\frac{1}{2}} h_1\|_{L^{\infty}_T L^{\infty}_x}+ \\&+
	\|\delta A \|_{L^\infty_T L^{p_1^*}_x}\Big( 
	\|t^{\beta_1}\nabla h_1\|_{L^{2r}_T L^{p_2}_x}
	\|t^{\gamma_2}h_1\|_{L^{\infty}_T L^{p_3}_x}+
	\|t^{\beta_1^\ee}\nabla h_1\|_{L^{2r}_T L^{q_2}_x}
	\|t^{\gamma_2^\ee}h_1\|_{L^{\infty}_T L^{q_3}_x}\Big)+\\&+
	\|\nabla A_1\|_{L^\infty_TL^N_x}
	\|t^{\beta_1}\nabla \delta h\|_{L^{2r}_TL^{p_2}_x} 
	\|t^{\gamma_2}h_2\|_{L^\infty_T L^{p_3}_x}+
	\|\nabla A_1\|_{L^\infty_TL^N_x}
	\|t^{\gamma_1} \delta h\|_{L^{2r}_TL^{p_2}_x} 
	\|t^{\beta_2}h_2\|_{L^\infty_T L^{p_3}_x}+
	\\&+
	\|A_1\|_{L^\infty_TL^\infty_x}
	\|t^{\beta_1} \nabla \delta h\|_{L^{2r}_TL^{p_2}_x} 
	\|t^{\gamma_2}h_2\|_{L^\infty_T L^{p_3}_x} +
	\| A_1\|_{L^\infty_TL^\infty_x}
	\|t^{\gamma_1} \delta h\|_{L^{2r}_TL^{p_3}_x} 
	\|t^{\beta_2}\nabla h_2\|_{L^\infty_T L^{p_2}_x} +\\&+	
	\|\nabla A_1\|_{L^\infty_TL^N_x}
	\|t^{\beta_1}\nabla \delta h\|_{L^{2r}_TL^{p_2}_x} 
	\|t^{\gamma_2}h_1\|_{L^\infty_T L^{p_3}_x}+
	\|\nabla A_1\|_{L^\infty_TL^N_x}
	\|t^{\gamma_1} \delta h\|_{L^{2r}_TL^{p_2}_x} 
	\|t^{\beta_2}h_1\|_{L^\infty_T L^{p_3}_x}+
	\\&+
	\|A_1\|_{L^\infty_TL^\infty_x}
	\|t^{\beta_1} \nabla \delta h\|_{L^{2r}_TL^{p_2}_x} 
	\|t^{\gamma_2}h_1\|_{L^\infty_T L^{p_3}_x} +
	\| A_1\|_{L^\infty_TL^\infty_x}
	\|t^{\gamma_1} \delta h\|_{L^{2r}_TL^{p_3}_x} 
	\|t^{\beta_2}\nabla h_1\|_{L^\infty_T L^{p_2}_x}\end{split}
\end{equation*}
Hence, arguing exactly as for \eqref{uniqueness2_delta_f1} and \eqref{uniqueness2_delta_f2}, there exists $\chi_3(t)$ such that
\begin{equation}\label{uniqueness2_delta_f3}
	\|t^{\alpha_1}\nabla \delta f_3\|_{L^{2r}_T L^{p_1}_x}\lesssim 
	\chi_3(T)\|(\delta v,\, \delta h,\,\nabla \delta P)\|_{\Y_{r,T}}.
\end{equation}

\noindent
\textit{Bounds for $t^{\alpha_1}\delta f_6$} From the definition of $\delta f_6$ we get
\begin{equation*}
\begin{split}
	\|t^{\alpha_1} &\delta f_6\|_{L^{2r}_TL^{p_1}_x}\lesssim 
	\|t^{\beta_1}\delta A \|_{L^\infty_TL^{p_2}_x}
	\|t^{\alpha_1^\ee}\nabla v_2\|_{L^{2r}_TL^\infty_x}
	\|t^{\gamma_2}h_2\|_{L^\infty_TL^{p_3}_x} +\\& +
	\|A_1\|_{L^\infty_TL^\infty_x} 
	\|t^{\beta_1}\nabla \delta v\|_{L^{2r}_TL^{p_2}_x}
	\|t^{\gamma_2} h_2\|_{L^{\infty}_TL^{p_3}_x}
	+
	\|A_1\|_{L^\infty_TL^\infty_x} 
	\|t^{\beta_2}\nabla  v_1\|_{L^{2r}_TL^{p_2}_x}
	\|t^{\gamma_3}\delta  h\|_{L^{\infty}_TL^{p_3}_x},
\end{split}
\end{equation*}
Thus
\begin{equation}\label{uniqueness2_delta_f6}
	\|t^{\alpha_1} \delta f_6\|_{L^{2r}_TL^{p_1}_x}\lesssim
	\chi_6(T)\|(\delta v,\, \delta h,\,\nabla \delta P)\|_{\Y_{r,T}},
\end{equation}
for an opportune continuous function $\chi_6(t)$ which goes to $0$ when $t\rightarrow 0$.

\noindent	
\textit{Bounds for $t^{\alpha_1}\delta f_7$}. From the definition of $\delta f_7$ we have
\begin{equation*}
\begin{split}
	\|t^{\alpha_1}\delta f_7  \|_{L^{2r}_TL^{p_1}_x}\lesssim 
	&\|\nabla A_2\|_{L^\infty_T L^N_x}
	\|t^{\alpha_1} \nabla \delta h\|_{L^{2r}_T L^{p_1^*}_x}+
	\|A_2\|_{L^\infty_T L^\infty_x}
	\|t^{\alpha_1} \nabla^2 \delta h \|_{L^{2r}_T L^{p_1}_x}+ \\&+
	\| \delta A \|_{L^\infty_TL^{p_1}_x} \|t^{\alpha_1^\ee}\nabla h_1\|_{L^{2r}_TL^\infty_x} + 
	\|\delta A\|_{L^\infty_T L^{p_1^*}_x}\|t^{\alpha_1}\nabla^2 h_1\|_{L^{2r}_T L^N_x},
\end{split}
\end{equation*}
which yields that there exists a continuous function $\chi_7(t)\geq 0$, with $\chi_7(0)=0$, such that
\begin{equation}\label{uniqueness2_delta_f7}
	\|t^{\alpha_1}\delta f_7  \|_{L^{2r}_TL^{p_1}_x}\lesssim \chi_7 (T)
	\|(\delta v,\, \delta h,\,\nabla \delta P)\|_{\Y_{r,T}}.
\end{equation}

\noindent	
\textit{Bounds for $t^{\alpha_1}\delta f_8$}. From the definition of $\delta f_8$ we get 
\begin{equation*}
\begin{split}
	\|t^{\alpha_1}&\delta f_9\|_{L^{2r}_TL^{p_1}_x}\lesssim  
	\|t^{\beta_1}\delta h\|_{L^{2r}_TL^{p_3}_x}
	\|t^{\beta_2}\nabla h_2 \|_{L^\infty_T L^{p_2}_x}
	\|\omega_2\|_{L^\infty_TL^\infty_x} + \|t^{\beta_2} h_1\|_{L^{\infty}_TL^{p_3}_x}
	\|t^{\beta_1}\nabla \delta h \|_{L^{2r}_T L^{p_2}_x} 
	{\scriptstyle \times}
	\\&{\scriptstyle \times}
	\|\omega_2\|_{L^\infty_TL^\infty_x} +
	\|t^{\beta_2} h_1\|_{L^{\infty}_TL^{p_3}_x}
	\|t^{\beta_2}\nabla h_1 \|_{L^{\infty}_T L^{p_2}_x}\|t^{-\frac{1}{2}-\frac{1}{2r}}
	\|_{L^{\frac{2r}{r-1}-\bar{\ee}}_T}
	\|t^{\frac{1}{2}}\delta \omega\|_{L^\infty_TL^\infty_x} 	 
\end{split}
\end{equation*}
For $\bar{\ee}$ small enough, that is,
\begin{equation}\label{uniqueness2_delta_f8}
	\|t^{\alpha_1}\delta f_8  \|_{L^{2r}_TL^{p_1}_x}\lesssim \chi_9 (T)\Big\{
	\|(\delta v,\, \delta h,\,\nabla \delta P)\|_{\Y_{r,T}} +
	\|t^{\frac{1}{2}}\delta \omega\|_{L^\infty_TL^\infty_x}
	\Big\}.
\end{equation}

\noindent	
\textit{Bounds for $t^{\alpha_1}\nabla \delta g $}. By the definition of $\delta g$ we get
\begin{equation*}
\begin{split}
	\|t^{\alpha_1}\nabla \delta g\|_{L^{2r}_TL^{p_1}_x} \lesssim 
	\|\nabla A_2&\|_{L^\infty_TL^N_x}\|t^{\alpha_1}\nabla \delta v\|_{L^{2r}_T L^{p_1^*}_x} + 
	\|\Id-A_2\|_{L^\infty_TL^\infty_x} \|t^{\alpha_1}\nabla^2 \delta v\|_{L^{2r}_TL^{p_1}_x} + \\&+
	\|\nabla \delta A\|_{L^{\infty}_TL^{p_1}_x}\|t^{\alpha_1^\ee}\nabla v_1\|_{L^{2r}_TL^\infty_x} + 
	\|\delta A\|_{L^\infty_T L^{p_1^*}_x}\|t^{\alpha_1}\nabla^2 v_1\|_{L^{2r}_T L^N_x}.
\end{split}
\end{equation*}
We deduce that there exists a continuous function $\chi_g(t)$ with $\chi_g(0)=0$ such that
\begin{equation}\label{uniqueness2_delta_g}
	\|t^{\alpha_1}\nabla \delta g\|_{L^{2r}_TL^{p_1}_x} \lesssim 
	\chi_g(T)\|(\delta v,\, \delta h,\,\nabla \delta P)\|_{\Y_{r,T}}.
\end{equation} 

\noindent
\textit{Bounds for $t^{\alpha_1}\delta R $}. From the definition of $\delta R$ we have
\begin{equation*}
\begin{split}
	\|t^{\alpha_1}\delta R\|_{L^{2r}_TL^{p_1}_x}\lesssim 
	\|t^{\beta_2}\nabla v_2&\|_{L^\infty_TL^{p_2}_x}\|t^{\gamma_1}\delta v\|_{L^{2r}_TL^{p_3}_x}+
	\|\Id-A_2\|_{L^{\infty}_T L^\infty_x}\|t^{\alpha_1}\partial_t \delta v\|_{L^{2r}_TL^{p_1}_x}+\\&+
	\|t^{\beta_1}\nabla\delta v\|_{L^{2r}_TL^{p_2}_x}\|t^{\gamma_2}v_1\|_{L^\infty_TL^{p_3}_x}
	+\|\delta A \|_{L^{2r}_T L^{p_1^*}_x}\|t^{\alpha_1} \partial_t v_1 \|_{L^{2r}_TL^N_x},
\end{split}
\end{equation*}
Thus, there exists  a continuous function $\chi_R(t)$ with $\chi_R(0)=0$ such that
\begin{equation}\label{uniqueness2_delta_R}
	\|t^{\alpha_1}\delta R\|_{L^{2r}_TL^{p_1}_x}\lesssim  
	\chi_R(T)\Big\{\|(\delta v,\, \delta h,\,\nabla \delta P)\|_{\Y_{r,T}}+ 
	\|t^{\alpha_1}\partial_t\delta v\|_{L^{2r}_TL^{p_1}_x}\Big\}.
\end{equation}

\noindent	
\textit{Bounds for $t^{\alpha_1} \delta f_4 $}. From the definition of $\delta f_4$ it follows
\begin{equation*}
\begin{split}
	\|&t^{\alpha_1} 	\delta f_4 		\|_{L^{1}     _T L^{p_1^*} 	_x}			\lesssim 
	\| t^{\beta _1}		\nabla \delta	h		\|_{L^{2r}    _T L^{p_2} 	_x}
	\| t^{\gamma_2}			(h_1,\,h_2)			\|_{L^{\infty}_T L^{p_3} 	_x} 
	\| 					(\omega_1,\,\omega_2)	\|_{L^\infty  _T L^\infty	_x}	+
	\| t^{\gamma_1}				\delta	h		\|_{L^{2r}    _T L^{p_2} 	_x}		{\scriptstyle\times}
																				\\&	{\scriptstyle \times}
	\| t^{\gamma_2}		\nabla (h_1,\,h_2)		\|_{L^{\infty}_T L^{p_3} 	_x} 
	\| 					(\omega_1,\,\omega_2)	\|_{L^\infty  _T L^\infty	_x}	+
	\| t^{\beta _1}			 \delta	h			\|_{L^{2r}	  _T L^{p_2} 	_x}
	\| t^{\gamma_2}			(h_1,\,h_2)			\|_{L^{\infty}_T L^{p_3} 	_x} 
	\| 				\nabla(\omega_1,\,\omega_2)	\|_{L^{2r}    _T L^\infty	_x}	+\\&+
	\| t^{\frac{1}{2}}				h_1			\|_{L^{\infty}_T L^{\infty} _x}^2
	\|t^{\alpha_1} 			\delta \omega		\|_{L^{\infty}_T L^{p_1^*}	_x} 
\end{split}
\end{equation*}
Therefore, we obtain
\begin{equation}\label{uniqueness2_delta_f4}
	\| t^{\alpha_1}\delta 	f_4 \|_{L^{1}_T  L^{p_1^*}_x}\lesssim \chi_4 (T)\Big\{
	\|(\delta v,\, \delta h,\,\nabla \delta P)\|_{\Y_{r,T}} +\|
	t^{\alpha_1}\delta \omega\|_{L^\infty_TL^{p_1^*}_x}
	\Big\}.
\end{equation}

\noindent	
\textit{Bounds for $t^{\gamma_1}\delta f_5 $}. By definition
\begin{equation*}
\begin{split}
	\|t^{\gamma_1}\delta f_5\|_{L^{1}_T L^{p_1^*}_x}&\lesssim 
	\|t^{\alpha_1}\delta A\|_{L^\infty_T L^{p_1^*}_x}\|\nabla h_2\|_{L^1_TL^\infty_x}+
	\|A_1\|_{L^2r_T L^\infty_x} \|t^{\alpha_1}\nabla \delta h\|_{L^{2r}_T L^{p_1^*}_x}\\
	&\lesssim 
	\|t^{\alpha_1} \nabla^2 \delta v\|_{L^{2r}_T L^{p_1}_x}\|\nabla h_2\|_{L^1_TL^\infty_x}+
	\|A_1\|_{L^2r_T L^\infty_x}\|t^{\alpha_1}\nabla^2 \delta h\|_{L^{2r}_T L^{p_1}_x}
\end{split}
\end{equation*}
hence
\begin{equation}\label{uniqueness2_delta_f5}
	\|t^{\alpha_1}\delta f_5\|_{L^{1}_T L^{p_1^*}_x}\lesssim \chi_5 (T)
	\|(\delta v,\, \delta h,\,\nabla \delta P)\|_{\Y_{r,T}} .
\end{equation}
with $\chi_5$ as the previous functions.

\noindent
Summarizing points \eqref{uniqueness2_delta_f1}, \eqref{uniqueness2_delta_f2}, 
\eqref{uniqueness2_delta_f3}, \eqref{uniqueness2_delta_f6}, \eqref{uniqueness2_delta_f7}, 
\eqref{uniqueness2_delta_f8},  \eqref{uniqueness2_delta_g}, \eqref{uniqueness2_delta_R}, 
\eqref{uniqueness2_delta_f4} and \eqref{uniqueness2_delta_f5}, we finally obtain
\begin{equation*}
\begin{split}
	|(\delta v,\, \delta h, \nabla \delta P)&\|_{\Y_{r,T}} + 
	\|t^{\alpha_1}\partial_t\delta v\|_{L^{2r}_TL^{p_1}_x}+
	\|t^{\alpha_1}\delta \omega\|_{L^\infty_TL^{p_1^*}_x}\lesssim \\&
	\chi(T) \Big\{ |(\delta v,\, \delta h, \nabla \delta P)\|_{\Y_{r,T}} + 
	\|t^{\alpha_1}\delta \omega\|_{L^\infty_TL^{p_1^*}_x}+
	\|t^{\alpha_1}\partial_t\delta v\|_{L^{2r}_TL^{p_1}_x} \Big\}
\end{split}
\end{equation*}
where $\chi$ stands for $\sum_i\chi_i$. Thus, for $T$ sufficiently small, the left-hand side has to be $0$.
This proves the uniqueness at least in small time interval. Then uniqueness part can be completed by a bootstrap method.  This concludes the proof of Theorem \ref{Main_Theorem_3}.
\appendix

\section{Estimates}

\begin{lemma}\label{Lemma4}
	Let the operator $\Cc$ be defined as in Lemma \ref{Lemma3}. 
	Consider $T\in(0,\infty]$, $\ee\geq 0$ small enough, $1<\bar{r}<\infty$, and moreover suppose that 
	$q,\,\tilde{q}$ satisfy $N/2<q<N/(1-\ee)$, 
	$\,\max\{N, q\}<\tilde{q}\leq \infty$. Let $\alpha^\ee$, $\gamma^\ee$ and 
	$\bar{\gamma}^\ee$	be defined by
	\begin{equation*}
	\alpha^\ee:=\frac{1}{2}\big(3-\frac{N}{q}-\ee \big) -\frac{1}{\bar{r}},\quad	
	\gamma^\ee:=\frac{1}{2}\big(1-\frac{N}{\tilde{q}}-\ee \big)	-\frac{1}{\bar{r}}\quad\text{and}\quad
	\bar{\gamma}^\ee:=\frac{1}{2}\big(1-\frac{N}{\tilde{q}}-\ee \big).
	\end{equation*}
	If $t^{\alpha^\ee} f(t)$ belongs to $L^{\bar{r}}(0,T;L^{q}_x)$ then $t^{\gamma^\ee}\Cc f(t)$ 
	belongs to $L^{\bar{r}}(0,T;L^{\tilde{q}}_x)$. 
	Furthermore there exist $C_\ee=C_\ee(q, \tilde{q}, \bar{r})>0$ such that
	\begin{equation}\label{inquality_Lemma4}
		\|t^{\gamma^\ee} \Cc f(t) \|_{L^{\bar{r}}(0,T; L^{\tilde{q}}_x)}
		\leq 		
		C_\ee \|t^{\alpha^\ee}  f(t) \|_{L^{\tilde{r}}(0,T; L^{q}_x)}.
	\end{equation}
	Moreover, if $\bar{r}>2$ and $N\bar{r}/(2\bar{r}-2)<q$,  
	then $t^{\bar{\gamma}^\ee}\Cc f(t)$ belongs to 
	$L^{\infty}(0,T;L^{\tilde{q}}_x)$ and there exists a positive constant 
	$\bar{C}_\ee = \bar{C}_\ee (q, \tilde{q}, \bar{r})$ 
	such that
	\begin{equation}\label{inquality2_Lemma4}
		\|t^{\bar{\gamma}^\ee} \Cc f(t) \|_{L^{\infty}(0,T; L^{\tilde{q}}_x)}\leq 
		\bar{C}_\ee \|t^{\alpha^\ee}  f(t) \|_{L^{\bar{r}}(0,T; L^{q}_x)}.
	\end{equation} 
\end{lemma}
\begin{proof}
	Recalling \eqref{estimate_heatkernel} we have 
	\begin{equation*}
		\|t^{\gamma^\ee} \Cc f(t) \|_{L^{\tilde{q}}_x}\lesssim 
		\int_0^t \frac{1}
		{|t-s|^{\frac{N}{2}\big(\frac{1}{q}-\frac{1}{\tilde{q}} \big)}}\|f(s)\|_{L^{q}_x}\dd s.
	\end{equation*}		
	Defining $F(s):=\|s^{\alpha^\ee}f(s)\|_{L^{q}_x}$, by a change of variable $s=t\tau$ and 
	because $  \gamma^\ee -\alpha^\ee +1 = (1/q-1/\tilde{q}) N/2$, we get that
	\begin{equation*}
		t^{\gamma^\ee}\|\Cc f(t)\|_{L^{\tilde{q}}_x}\lesssim 
		\int_0^1 \frac{1}{|1-\tau|^{\frac{N}{2}\big(\frac{1}{q}-\frac{1}{\tilde{q}} \big)
		}}\tau^{-\alpha^\ee}F(t\tau)\dd \tau.  
	\end{equation*}
	Applying Minkowski inequality, we deduce that
	\begin{equation*}
		\|t^{\gamma^\ee} \Cc f(t) \|_{L^{\bar{r}}_t L^{\tilde{q}}_x}\lesssim 
		\int_0^1 \frac{1}
		{|1-\tau|^{\frac{N}{2}\big(\frac{1}{q}-\frac{1}{\tilde{q}}\big) }}\tau^{-\alpha^\ee}
		\Big( \int_0^{T\tau}F(t\tau)^{\bar{r}}\dd t \Big)^{\frac{1}{\bar{r}}}\dd \tau,
	\end{equation*}
	which yields
	\begin{equation*}
		\|t^{\gamma^\ee} \Cc f(t) \|_{L^{\bar{r}}_t L^{\tilde{q}}_x}\lesssim 
		\int_0^1 \frac{1}
		{|1-\tau|^{\frac{N}{2}\big(\frac{1}{q}-\frac{1}{\tilde{q}}\big)}}
		\tau^{-\alpha^\ee-\frac{1}{\bar{r}}}\dd \tau
		\|t^{\alpha_1^\ee}  f(t) \|_{L^{r_1}(0,T; L^{q_1}_x)}.
	\end{equation*}
	Thus, because $(1/q-1/\tilde{q})N/2<1$ and $0<1/2(3-N/q -\ee)<1$, we obtain inequality 
	\eqref{inquality_Lemma4}.

	On the other hand, observing that
	\begin{equation*}
		t^{\bar{\gamma}^\ee}\|\Cc f(t)\|_{L^{\tilde{q}}_x}\lesssim 
		\int_0^t \frac{t^{\bar{\gamma}^\ee} s^{-\alpha^\ee}}
		{|t-s|^{\frac{N}{2}\big(\frac{1}{q}-\frac{1}{\tilde{q}}\big)}} 
		F(s)\dd s\lesssim
		\Big(
		\int_0^t
		\frac{t^{\bar{\gamma}^\ee \bar{r}'}s^{-\alpha^\ee\bar{r}'}}
		{|t-s|^{\frac{N}{2}\big(\frac{1}{q}-\frac{1}{\tilde{q}}\big)\bar{r}'}}
		\dd s\Big)^{\frac{1}{\bar{r}'}}
		\|t^{\alpha^\ee}  f(t) \|_{L^{\bar{r}}(0,T; L^{q}_x)} . 
	\end{equation*}
	By a change of variable $s=t\tau$ and because $\bar{\gamma}^\ee \bar{r}'-\alpha^\ee \bar{r}'-
	(1/q-1/\tilde{q})N\bar{r}'/2+1=0$, 		
	we obtain
	\begin{equation*}
		t^{\gamma^\ee_2}\|\Cc f(t)\|_{L^{q_2}_x}\lesssim
		\Bigg(
		\int_0^1 \frac{\tau^{-\alpha^\ee\bar{r}'}}
		{|1-\tau|^{\frac{N}{2}\big(\frac{1}{q}-\frac{1}{\tilde{q}}\big)\bar{r}'}}
		\dd \tau \Bigg)^{\frac{1}{\bar{r}'}}
		\|t^{\alpha_1^\ee}  f(t) \|_{L^{\bar{r}}(0,T; L^{q_1}_x)}.
	\end{equation*}
	Since $q>N\bar{r}/(2\bar{r}-2)$ yields $(1/q-1/\tilde{q})N/2<1/\bar{r}'$ and $q<N/(1-\ee)$ implies 
	$\alpha^\ee \bar{r}'<1$, we obtain \eqref{inquality2_Lemma4}, which completes the proof of the Lemma.

\end{proof}

\begin{lemma}\label{Lemma5}
	Let the operators $\Bb$ be defined as in Lemma \ref{Lemma2}. 
	Consider $T\in (0,\infty]$, $\ee\geq 0$ small enough, $1<\bar{r}<\infty$, and moreover suppose that 
	$q,\,\bar{q}$ satisfy $N/2<q<N/(1-\ee)$ and $q\leq \bar{q}$ such that $1/q-1/\bar{q}<1/N$. Let 
	$\alpha^\ee$ be defined as in Lemma \ref{Lemma4} and $\beta^\ee$ and 	$\bar{\beta}^\ee$ be defined by
	\begin{equation*}
	\bar{\beta}^\ee:=\frac{1}{2}\big(2-\frac{N}{\bar{q}}-\ee \big)	\quad\text{and}\quad
	\beta^\ee:=\frac{1}{2}\big(1-\frac{N}{\bar{q}}-\ee \big)-\frac{1}{\bar{r}}.
	\end{equation*}
	If $t^{\alpha^\ee} f(t)$ belongs to $L^{\bar{r}}(0,T; L^{q}_x)$ then 
	$t^{\beta^\ee}\mathcal{B}f(t) $ 
	belongs to $ L^{\bar{r}}(0,T;L^{\bar{q}}_x)$ and there exists a positive constant
	$C_\ee=C_\ee (q, \bar{q}, \bar{r})$ such that
	\begin{equation}\label{inquality_Lemma5}
		\|t^{\beta^\ee} \mathcal{B} f(t) \|_{L^{\bar{r}}(0,T; L^{\bar{q}}_x)}\leq 
		\bar{C}_\ee \|t^{\alpha^\ee}  f(t) \|_{L^{\bar{r}}(0,T; L^{q}_x)}.
	\end{equation}
	Moreover, if $\bar{r}>2$, $N\bar{r}/(2\bar{r}-2)<q$ and $\bar{q}<Nr$ 
	then $t^{\bar{\beta}^\ee}\mathcal{B}f(t)$ belongs to 
	$L^{\infty}(0,T;L^{\bar{q}}_x)$ and there exists a positive constant 
	$\bar{C}_\ee = \bar{C}_\ee (q, \bar{q}, \bar{r})$ 
	such that
	\begin{equation}\label{inquality2_Lemma5}
		\|t^{\bar{\beta}^\ee} \mathcal{B} f(t) \|_{L^{\infty}(0,T; L^{\bar{q}}_x)}\leq 
		C_\ee \|t^{\alpha^\ee}  f(t) \|_{L^{\bar{r}}(0,T; L^{q}_x)}.
	\end{equation}
\end{lemma}
\begin{proof}
	At first, recalling \eqref{estimate_nablaheatkernel}, we get that
	\begin{equation*}
		t^{\beta^\ee}\|\mathcal{B}f(t)\|_{L^{\bar{q}}_x}\lesssim 
		t^{\beta^\ee}
		\int_0^t \frac{1}
		{|t-s|^{\frac{N}{2}\big(\frac{1}{q}-\frac{1}{\bar{q}} \big)+\frac{1}{2}}}\|f(s)\|_{L^{q}_x}\dd s.  
	\end{equation*}
	Defining $F(s):=\|s^{\alpha^\ee}f(s)\|_{L^{q}_x}$, by a change of variable $s=t\tau$ and 
	because $  \beta^\ee -\alpha^\ee +1 = 1/2 + (1/q-1/\bar{q}) N/2$, we get that
	\begin{equation*}
		t^{\beta^\ee}\|\mathcal{B}f(t)\|_{L^{\bar{q}}_x}\lesssim 
		\int_0^1 \frac{1}{|1-\tau|^{\frac{N}{2}\big(\frac{1}{q}-\frac{1}{\bar{q}} \big)+
		\frac{1}{2}}}\tau^{-\alpha^\ee}F(t\tau)\dd \tau.  
	\end{equation*}
	Applying Minkowski inequality, we deduce that
	\begin{equation}\label{estimate1_Lemma4}
		\|t^{\beta^\ee} \mathcal{B} f(t) \|_{L^{\bar{r}}_t L^{\bar{q}}_x}\lesssim 
		\int_0^1 \frac{1}{|1-\tau|^{\frac{N}{2}\big(\frac{1}{q}-\frac{1}{\bar{q}} \big)+
		\frac{1}{2}}}\tau^{-\alpha^\ee}
		\Big( \int_0^{T\tau}F(t\tau)^{\bar{r}}\dd t \Big)^{\frac{1}{\bar{r}}}\dd \tau,
	\end{equation}
	which yields
	\begin{equation*}
		\|t^{\beta^\ee} \mathcal{B} f(t) \|_{L^{\bar{r}}_t L^{\bar{q}}_x}\lesssim 
		\int_0^1 \frac{1}{|1-\tau|^{\frac{N}{2}\big(\frac{1}{q}-\frac{1}{\bar{q}} \big)+
		\frac{1}{2}}}\tau^{-\alpha^\ee-\frac{1}{\bar{r}}}\dd \tau
		\|t^{\alpha^\ee}  f(t) \|_{L^{\bar{r}}(0,T; L^{q}_x)}.
	\end{equation*}
	Because $0<(1/q-1/\bar{q})N/2+1/2<1$ and $0<1/2(3-N/q -\ee)<1$, we deduce inequality 
	\eqref{inquality_Lemma5}.

	\noindent For the second inequality, proceeding in a similar way of the previous Lemma, we obtain that
	\begin{equation*}
		t^{\beta^\ee_2}\|\mathcal{B}f(t)\|_{L^{\bar{q}}_x}\lesssim
		\Bigg(
		\int_0^t \frac{t^{\beta^\ee \bar{r}'}s^{-\alpha^\ee\bar{r}_1'}}
		{|t-s|^{\frac{N}{2}\big(\frac{1}{q}-\frac{1}{\bar{q}} \big)\bar{r}'+
		\frac{1}{2}\bar{r}'}}
		\dd s\Bigg)^{\frac{1}{\bar{r}'}}
		\|t^{\alpha^\ee}  f(t) \|_{L^{\bar{r}}(0,T; L^{q}_x)}
	\end{equation*}
	By a change of variable $s=t\tau$ and because $\beta^\ee\bar{r}'-\alpha^\ee\bar{r}'-
	(1/q-1/\bar{q})N\bar{r}'/2-\bar{r}'/2+1=0$, 		
	we obtain
	\begin{equation*}
		t^{\beta^\ee}\|\Bb f(t)\|_{L^{\bar{q}}_x}\lesssim
		\Bigg(
		\int_0^1 \frac{\tau^{-\alpha^\ee\bar{r}'}}
		{|1-\tau|^{\frac{N}{2}\big(\frac{1}{q}-\frac{1}{\bar{q}} \big)\bar{r}'+
		\frac{1}{2}\bar{r}'}}
		\dd \tau \Bigg)^{\frac{1}{\bar{r}'}}
		\|t^{\alpha^\ee}  f(t) \|_{L^{\bar{r}}(0,T; L^{q}_x)}.
	\end{equation*}
	Since by the hypotheses we can deduce $\alpha^\ee_1\bar{r}'<1$ 
	and $(1/q-1/\bar{q})N\bar{r}'/2+\bar{r}'/2<1 $ then there 
	exists $C_\ee>0$ such that
	\begin{equation*}
		\|t^{\beta^\ee} \mathcal{B} f(t) \|_{L^{\infty}(0,T; L^{\bar{q}}_x)}\leq C_\ee 
		\|t^{\alpha^\ee}  f(t) \|_{L^{\bar{r}}(0,T; L^{q}_x)}
	\end{equation*}
\end{proof}

\begin{lemma}\label{Lemma6}
	Let $1<\bar{r}<\infty$, $q>N\bar{r}/(2\bar{r}-2)$ and $\sigma:= 1-N/(2q)-1/\bar{r}$. Let us suppose that 
	$t^{\sigma} f$ belongs to $L^{\bar{r}}(0,T; L^q_x)$ with $T\in (0,\infty]$. Then $\Cc f$ belongs to 
	$L^\infty(0,T; L^\infty_x)$ and for every $t\in (0,T)$
	\begin{equation*}
		\| \Cc f(t) \|_{L^\infty_x}\leq C_{\bar{r}}\|s^{\sigma} f\|_{L^{\bar{r}}(0,t; L^q_x)},
	\end{equation*}
	where $C_{\bar{r}}r$ is a positive constant dependent only by $\bar{r}$.
\end{lemma}
\begin{proof}
	Recalling \eqref{estimate_heatkernel} we get
	\begin{equation*}
		\| \Cc f(t) \|_{L^\infty_x} \leq \int_0^t \frac{1}{|t-s|^{\frac{N}{2}\frac{1}{q}}s^{\sigma}} 
		F(s)\dd s,
	\end{equation*}
	for every $t\in (0,T)$, where $F(s)=s^{\sigma}\|f(s)\|_{L^\infty_x}$. By the change of variable 
	$s=t\, \tau$ we obtain
	\begin{equation*}
		\|\Cc f(t)\|_{L^\infty_x} \leq \int_0^1 \frac{1}{|1-\tau|^{\frac{N}{2}\frac{1}{q}}\tau^{\sigma}} 
		F(t\,\tau)t^{1-\frac{N}{2}\frac{1}{q}-\sigma}\dd s=
		\int_0^1 \frac{1}{|1-\tau|^{\frac{N}{2}\frac{1}{q}}\tau^{\sigma}} 
		F(t\,\tau)t^{\frac{1}{\bar{r}}}\dd s.
	\end{equation*}
	Hence, by H\"{o}lder inequality, it follows
	\begin{equation*}
		\| \Cc f(t) \|_{L^\infty_x} \leq
		\Big( \int_0^1 \frac{1}{|1-\tau |^{\frac{N}{2}\frac{1}{q}\bar{r}'}}\dd \tau \Big)^\frac{1}{\bar{r}'}
		\Big( \int_0^1 |F(t\, \tau)|^r t \dd \tau \Big)^\frac{1}{r}.
	\end{equation*}
	Since $\bar{r}' N/(2q)<1$, we finally get
	\begin{equation*}
	\| \Cc f(t) \|_{L^\infty_x} \leq C_{\bar{r}} \Big( \int_0^1 |F(t\, \tau)|^r t \dd \tau \Big)^\frac{1}{r}
	= C_{\bar{r}}\|s^{\sigma} f\|_{L^{\bar{r}}(0,t; L^q_x)}.
	\end{equation*}
\end{proof}
\noindent Finally we enunciate the following Lemma, which proof is basically equivalent to the previous one.
\begin{lemma}\label{Lemma7}
	Let $2<\bar{r}<\infty$, $q>N\bar{r}/(\bar{r}-2)$ and $\sigma:= (1-N/q)1/2-1/\bar{r}$. Let us suppose that 
	$t^{\sigma} f$ belongs to $L^{\bar{r}}(0,T; L^q_x)$ with $T\in (0,\infty]$. Then $\Bb f$ and 
	$t^{-1/2}\Cc f$ belong to 
	$L^\infty(0,T; L^\infty_x)$ and for every $t\in (0,T)$
	\begin{equation*}
		\| \Bb f(t),\, t^{-\frac{1}{2}}\Cc f(t) \|_{L^\infty_x}\leq 
		C_{\bar{r}}\|s^{\sigma} f\|_{L^{\bar{r}}(0,t; L^q_x)},
	\end{equation*}
	where $C_{\bar{r}}$ is a positive constant dependent only by $\bar{r}$.
\end{lemma}

\section{Technical Results for the Heat and Stokes equations}

\noindent We consider the following system, composed by an Heat equation and a free Stokes equation with a linear perturbation:

\begin{equation}\label{Appendix_Navier_Stokes_system}
	\begin{cases}
		\;\partial_t u + v\cdot \nabla u -\Delta u+\nabla \Pi =f_1		& \RR_+ \times\RR^N,\\
		\; \partial_t d -\Delta d= f_2									& \RR_+ \times\RR^N,\\
		\;\Div\, u = 0													& \RR_+ \times\RR^N,\\
		\;(u,\, d)_{|t=0} = (u_0,\, d_0)							& \;\;\quad \quad\RR^N,\\
	\end{cases}
\end{equation}
where $d_0\in L^{\infty}_x$ and $(u_0, \,\nabla d_0)$ belongs to $\BB_{p,r}^{N/p-1}$ with 
$1<p<N$ and $1<r<\infty$.   
Propositions \ref{Appx_Prop_exist_sol_scheme} and Proposition \ref{Appx_Prop_exist_sol_scheme_weight_in_time} concern the existence of a solution $(u,\,d,\,\nabla \Pi)$, which belong to $\X_{r,T}$ and $\Y_{r,T}$ respectively. For $p$ less than (or equal to) the critical exponent $Nr/(3r-2)$ we can solve our system in a functional framework based only on some regularizing effects for the heat kernel in $L^pL^q$ spaces. However, 
if $p$ exceeds this critical value, in order to handle this less of regularity we have the add a weight in time. 

\noindent
Proposition \ref{Theorem_solutions_smooth_dates} requires the following result:

\begin{prop}\label{Appx_Prop_exist_sol_scheme}
	Let $1<r\leq 2$ and $1<p\leq Nr/(3r-2)$. Suppose that $f_1,\, \nabla f_2$ belong to 
	$L^r_TL^{N3/(3r-2)}_x$, $f_2 \in L^1_TL^\infty_x\cap L^r_TL^{Nr/2(r-1)}_x$, 
	$\nabla f_2$ belongs to 
	$L^{6r/5}_T L^{3Nr/(6r-2)}_x$. Assume that $v$ belongs to $L^{2r}_T L^{Nr/(r-1)}_x$ 
	and its norm is small enough. Let us assume that $d_0$ takes value in 
	$\SSS^{N-1}$, $u_0,\,\nabla d_0$ belong to $\BB_{p,r}^{N/p-1}$ and condition 
	\eqref{smallness_condition} is satisfied. Then there exists $(u,\,d,\,\nabla \Pi)$ solution of  
	\eqref{Appendix_Navier_Stokes_system} such that $d$ belongs to $L^{\infty}_TL^\infty_x$,
	$(u,\,\nabla d,\,\nabla \Pi)$ belongs to $\X_{r,T}$ and $(u,\,\nabla d)$ belongs to 
	$L^{2}_T L^\infty_x$.
\end{prop}
\begin{proof}
	The case of the simple heat equation in $d$ is provided by the Mild formulation, namely
	\begin{equation}\label{appx_prop1_d}
		d(t) = e^{t\Delta }d_0 + \int_0^te^{(t-s)\Delta }f_2(s)\dd s=
		e^{t\Delta}d_0+\Cc f_2(t).
	\end{equation}
	We immediately get $d\in L^{\infty}_TL^\infty_x$ and its norm is bounded by $\|d_0
	\|_{L^\infty_x}+\int_0^t\|f_2(s)\|_{L^\infty_x}\dd s$. Moreover, by Corollary 
	\ref{Cor_Characterization_of_hom_Besov_spaces}, because $\nabla d_0\in 
	\BB_{p,r}^{N/p-1}\hookrightarrow \BB^{-1/r}_{Nr/(r-1),2r}$, we deduce that 
	$e^{t\Delta}\nabla d_0$ belongs $L^{2r}_T L^{Nr/(r-1)}_x$. the integral 
	part $\nabla \Cc f_2(t) = \Cc \nabla f_2(t)$ is handled by Lemma \ref{Lemma3} with 
	$r_1=r$, $r_2=2r$ 
	$q_1=Nr/(3r-2)$ and $q_2=Nr/(r-1)$. Similarly, because $\nabla^2 d_0$ belongs to 
	$\BB_{Nr/(2r-1),2r}^{-1/r}\cap \BB^{-2/r}_{Nr/(2r-2),r}$, we get 
	$\nabla e^{t\Delta}\nabla d_0\in L^{2r}_T L^{Nr/2r-1}_x\cap L^r_T L^{Nr/(2r-2)}_x$ 
	and by Lemma \ref{Lemma2} we obtain 
	$\Bb \nabla f_2\in  L^{2r}_T L^{Nr/2r-1}_x\cap L^r_T L^{Nr/(2r-2)}_x$. Observing also that 
	$\BB_{p,r}^{N/r-1}$ is embedded in $\BB^{-2/3r}_{3Nr/(3r-2), 3r}$, again by Corollary 
	\ref{Cor_Characterization_of_hom_Besov_spaces} we get that $e^{t\Delta}\nabla d_0$ 
	belongs to $L^{3r}_T L^{3Nr/(3r-2)}_x$. The same property is fulfilled by 
	$\Bb f_2 = \Cc \nabla f_2$, using Lemma \ref{Lemma2} with $r_1=6r/5$ and $q_1=3Nr/(6r-5)$.
	At last, since $\nabla^3 d_0\in \BB_{Nr/(3r-2),r}^{-2/r}$ we deduce again 
	by Corollary \ref{Cor_Characterization_of_hom_Besov_spaces} that 
	$\nabla^2 e^{t\Delta}\nabla d_0$ belongs to $L^{r}_T L^{Nr/(3r-2)}_x$, while the same result 
	is allowed for $\Aa\nabla f_2$ by Theorem \ref{Maximal_regularity_theorem}.
	Hence, $(u,\,\nabla d,\,\nabla \Pi)$ belongs to $\X_{r,T}$ at least for the terms 
	related to $d$. Furthermore, since $\nabla d_0$ belongs to $\BB_{\infty,2}^{-1}$ (here the 
	necessary condition $r\leq 2$) we get $e^{t\Delta }\nabla d_0\in L^2_T L^\infty_x$ and by 
	Lemma \ref{Lemma2} with $r_1=r$, $r_2=2$, $q_1=Nr/(r-1)$ and $q_2=\infty$, we deduce that 
	$\Bb f_2$ belongs to $L^2_TL^\infty_x$, that is $\nabla d\in L^2_T L^\infty_x$
	
	\noindent
	Concerning the Stokes equation with the $v$-linear perturbation the mainly idea is to use 
	the Fixed-Point Theorem on the space $\tilde{\X}_{r,T}$ determined by
	\begin{equation*}
		\tilde{\X}_{r,T} := \big\{\,(u,\,\nabla\Pi)\quad\text{such that}\quad
		(u,\,d,\,\nabla \Pi)\in \X_{r,T}\big\}.
	\end{equation*}
	Indeed, let $(\omega_i,\,\nabla P_i)$ belong to $\tilde{\X}_{r,T}$, for $i=1,2$, and let us 
	define
	\begin{equation}\label{appx_prop1_ui_Pii}
	\begin{aligned}
		u_i(t) &:= e^{t\Delta }u_0  + \int_0^t e^{(t-s)\Delta}\big\{ -v\cdot \nabla \omega_i -
		\nabla P_i +f_1(s)\big\}\dd s,\\
		\nabla \Pi_i &:= -RR\cdot \big\{ v\cdot\nabla \omega_i + f_1\big\},
	\end{aligned}	
	\end{equation}
	then we have $(u_i,\,\nabla P_i)\in \tilde{\X}_{r,T}$, by the same techniques used for $d$. 
	Moreover, subtracting in $i$, $\delta u := u_1-u_2$,  $\delta \nabla \Pi := \nabla \Pi_1 -
	\nabla \Pi_2$, $\delta \omega: = \omega_1-\omega_2$ and  $\delta \nabla P: = \nabla P_1 -
	\nabla P_2$, we get 
	\begin{equation*}
		\|(\delta u,\,\delta \nabla \Pi)\|_{\tilde{X}_{r,T}}\lesssim 
		\|v\|_{L^{2r}_T L^{\frac{Nr}{r-1}}_x}\|(\delta \omega,\,\delta P)\|_{\tilde{X}_{r,T}}.
	\end{equation*}
	Thus, by the Fixed-Point Theorem, on the condition $\|v\|_{L^{2r}_T L^{Nr/(r-1)}_x}$ small 
	enough, there exists $(u,\,d,\nabla \Pi)$ solution for 
	\eqref{Appendix_Navier_Stokes_system}, with the properties described by the statement. 
	This concludes the proof of Proposition \ref{Appx_Prop_exist_sol_scheme}. 
\end{proof}

\noindent Now we extend the range of $r$ to $(1,\, \infty)$ and we consider an index of integrability $p$ greater than the critical $Nr/(3r-2)$. As already mentioned, here the addition of a weight in time is necessary. The following result is used in proposition \ref{Theorem_solutions_smooth_dates_2}.

\begin{prop}\label{Appx_Prop_exist_sol_scheme_weight_in_time}
	Let $1<r< \infty$ and $Nr/(3r-2)<p<N$. Recalling the notation of Theorem \ref{Main_Theorem_2}, let us 
	suppose that $t^{\alpha_1}(f_1,\,\nabla f_2)$ belongs to $L^{2r}_T L^{p_1}_x$ and $t^{2\gamma_1}f_2$ 
	belongs to $L^{r}_T L^{p_3/2}_x$. Assume that $t^{\gamma_1}\in L^{2r}_T L^{p_3}_x$ and its norm is 
	small enough. Let $d_0$ and $u_0$ be defined as in Proposition \ref{Appx_Prop_exist_sol_scheme}. 
	Then there exists $(u,\,d,\,\nabla \Pi)\in \Y_{r,T}$ solution of \eqref{Appendix_Navier_Stokes_system}, 
	with $d\in L^\infty_T L^\infty_x$.
\end{prop}
\begin{proof}
	The proof is basically equivalent to the one of Proposition \ref{Appx_Prop_exist_sol_scheme}. At first, 
	by \eqref{appx_prop1_d} and Lemma \ref{Lemma6}, we get 
	\begin{equation*}
	\|d\|_{L^\infty_T L^\infty_x}\lesssim \|d_0
	\|_{L^\infty_x}+\|t^{2\gamma_1}f_2\|_{L^r_T L^{\frac{p_3}{2}}_x}.
	\end{equation*}
	Recalling Theorem \ref{Characterization_of_hom_Besov_spaces}, by $\nabla d_0\in \BB_{p_3,2r}^{N/p_3-1}$, 
	$\nabla^2d_0 \in \BB_{p_2, 2r}^{N/p_2-1}$ and $\nabla^3d_0\in \BB_{p_1,2r}^{N/p_1-1}$, we get that
	$t^{\gamma_1}e^{t\Delta}\nabla d_0\in L^{2r}_TL^{p_3}_x$, $t^{\beta_1}\nabla e^{t\Delta}\nabla d_0
	\in L^{2r}_TL^{p_2}_x$ and $t^{\alpha_1}e^{t\Delta}\nabla d_0\in L^{2r}_TL^{p_1}_x$. Similarly 
	$t^{\gamma_2}e^{t\Delta}\nabla d_0\in L^{\infty}_TL^{p_3}_x$ and we get also 
	$t^{\beta_2}\nabla e^{t\Delta}\nabla d_0
	\in L^{\infty}_TL^{p_2}_x$. Because $\nabla d_0 \in \BB_{3p_1, 2r}^{N/(3p_1)-1}$ we get $t^{\gamma_3} 
	e^{t\Delta}\nabla d_0\in L^{2r}_T L^{3p_1}_x$ and $t^{\gamma_4}e^{t\Delta}\nabla d_0$ belongs to 
	$L^{\infty}_T L^{3p_1}_x$. 
	
	\noindent Using Lemma \ref{Lemma4} and Lemma \ref{Lemma5} with $\ee=0$, $q=p_1$, $\tilde{q}=p_2$, 
	$\bar{q}=p_3$ or $\bar{q}=3p_1$, we deduce the previous results for $\Cc \nabla f_2$ instead of 
	$e^{t\Delta}\nabla d_0$ (observing also that $\nabla\Cc =\Bb$ and $\nabla^2\Cc=\Aa$). Thus 
	$\nabla d$ fulfils all the condition imposed by $\Y_{r,T}$. 
	
	\noindent To conclude the proof, we use the Fixed-Point Theorem. Denoting $\tilde{Y}_{r,T}$ the set 
	composed by the	couples $(u,\,\nabla\Pi)$ such that $(u,\,d,\,\nabla\Pi)$ belongs to $\Y_{r,T}$, 
	we consider $(\omega_i,\, \nabla P_i)\in \tilde{Y}_{r,T}$, for $i=1,2$. Thus, defining $(u_i, \nabla \Pi_i
	)$ by \eqref{appx_prop1_ui_Pii}, we have
	\begin{equation*}
		\|(\delta u,\,\delta \nabla \Pi)\|_{\tilde{X}_{r,T}}\lesssim 
		\|t^{\gamma_1}v\|_{L^{2r}_T L^{p_3}_x}\|(\delta \omega,\,\delta P)\|_{\tilde{X}_{r,T}},
	\end{equation*}
	hence there exists $(u, \,d,\,\nabla \Pi)\in \Y_{r,T}$ solution of \eqref{Appendix_Navier_Stokes_system}, 
	and this concludes the proof.
\end{proof}
\pagestyle{empty}
\bibliographystyle{amsplain}
\providecommand{\bysame}{\leavevmode\hbox to3em{\hrulefill}\thinspace}
\providecommand{\MR}{\relax\ifhmode\unskip\space\fi MR }
\providecommand{\MRhref}[2]{%
  \href{http://www.ams.org/mathscinet-getitem?mr=#1}{#2}
}
\providecommand{\href}[2]{#2}

\end{document}